\documentclass[12pt,reqno]{article}

\usepackage[authoryear]{natbib}

\usepackage[Symbol]{upgreek}
\usepackage[T1]{fontenc}
\usepackage[latin9]{inputenc}
\usepackage[letterpaper]{geometry}
\usepackage{color}
\usepackage{amsmath}
\usepackage{amsthm}
\usepackage{amssymb}
\usepackage{setspace}
\usepackage[authoryear]{natbib}
\usepackage[unicode=true,pdfusetitle,
 bookmarks=true,bookmarksnumbered=false,bookmarksopen=false,
 breaklinks=false,pdfborder={0 0 0},pdfborderstyle={},backref=false,colorlinks=true]
 {hyperref}
\hypersetup{
 linkcolor=magenta,citecolor=blue}
 
\usepackage{graphicx}
\graphicspath{ {./graphs/} }
\usepackage{babel}
\providecommand{\tabularnewline}{\\}
\makeatother
\usepackage{array}
\usepackage{booktabs}
\usepackage{multirow}

\makeatletter

\usepackage{babel}

\makeatletter
\numberwithin{equation}{section}
\numberwithin{figure}{section}
\theoremstyle{remark}
\newtheorem{rem}{\protect\remarkname}[section]
\theoremstyle{definition}
\newtheorem{defn}{\protect\definitionname}[section]
\theoremstyle{plain}
\newtheorem{thm}{\protect\theoremname}[section]
\theoremstyle{plain}
\newtheorem{prop}{\protect\propositionname}[section]
\theoremstyle{definition}
\newtheorem{example}{\protect\examplename}[section]
\theoremstyle{definition}
\newtheorem{condition}{\protect\conditionname}

\newtheorem{lem}{\protect\lemmaname}[section]

\newtheorem*{assumptionDQM}{Assumption DQM}
\newtheorem*{assumptionC}{Assumption C}
\newtheorem*{assumptionG}{Assumption G}

\@ifundefined{date}{}{\date{}}
\usepackage{amsmath}
\usepackage{graphicx,psfrag,epsf}
\usepackage{enumerate}
\usepackage{natbib}
\usepackage{url} 

\usepackage{amsthm}
\usepackage{amsfonts}
\usepackage{amstext}

\usepackage{pdflscape}
\usepackage{threeparttable}

\usepackage{natbib}
\makeatother

\providecommand{\conditionname}{Condition}
\providecommand{\definitionname}{Definition}
\providecommand{\examplename}{Example}
\providecommand{\lemmaname}{Lemma}
\providecommand{\propositionname}{Proposition}
\providecommand{\remarkname}{Remark}
\providecommand{\theoremname}{Theorem}

\usepackage{todonotes}
\usepackage{verbatim}

\begin{document}
\sloppy

\newgeometry{verbose,tmargin=1in,bmargin=1in,lmargin=1.25in,rmargin=1.25in,footskip=1cm}

\title{Treatment Choice with Nonlinear Regret{\thanks{%
We would like to thank Isaiah Andrews, Tommaso Denti, Takashi Hayashi, Keisuke Hirano, Charles Manski, Francesca Molinari, Jos\'e Luis Montiel Olea, Collin Raymond,  J\"org Stoye, 
Aleksey Tetenov, Davide Viviano, Kohei Yata, and the participants at ASSA 2023, Bristol Econometric Study Group Conference, Boston College, Chicago, Cowles Foundation Conference, Cornell Brown Bag Seminar, Greater NY Metropolitan Area Econometrics Colloquium, IAAE 2022,
Indiana Bloomington, Interactions Conference 2022, Michigan State, NASMES 2022, Northwestern and Penn State for helpful comments.
The authors gratefully acknowledge financial support from ERC grants (numbers 715940 for Kitagawa and 646917 for Lee), the
ESRC Centre for Microdata Methods and Practice (CeMMAP) (grant number RES-589-28-0001) and NSF grant (number SES-2315600).}}}
\author{Toru Kitagawa\thanks{ Department of Economics, Brown University and Department of Economics, University College London. Email: toru\_kitagawa@brown.edu} \and Sokbae Lee\thanks{Department of Economics, Columbia University. Email: sl3841@columbia.edu} \and Chen Qiu\thanks{Department of Economics, Cornell University. Email: cq62@cornell.edu}}
\date{September 2024}
\maketitle

\begin{abstract}

The literature focuses on the mean of welfare regret, which can lead to undesirable treatment choice due to sensitivity to sampling uncertainty. We propose to minimize the mean of a \textit{nonlinear transformation} of regret and show that singleton rules are not essentially complete for nonlinear regret. Focusing on mean square regret, we derive closed-form fractions for finite-sample Bayes and minimax optimal rules. Our approach is grounded in decision theory and extends to limit experiments. The treatment fractions can be viewed as the strength of evidence favoring treatment. We apply our framework to a normal regression model and sample size calculations.\\ 

\end{abstract}

\restoregeometry

\newpage{}

\onehalfspacing

\section{Introduction}
Evidence-based policy making using randomized control trial data is becoming increasingly common in various fields of economics. How should we use data to inform an optimal policy decision in terms of social welfare? Building on the framework of statistical decision theory as laid out in \cite{Wald50}, the literature on statistical treatment choice initiated by \cite{manski2004statistical} analyzes how to use data to inform a welfare optimal policy. Following \cite{Savage51} and \cite{manski2004statistical}, researchers often focus on the average of welfare regret across the sampled data (called \emph{expected regret}) and obtain an optimal decision rule by minimizing a worst-case expected regret. 

When it comes to the ranking of different statistical decision rules, once we eliminate those that are stochastically dominated, it becomes less obvious how we should compare decision rules that do not stochastically dominate each other. Focusing on the expected regret, as suggested by \cite{manski2004statistical}, provides a natural starting point.\footnote{There is, however, no compelling argument why we should limit our attention to the mean of regret, as has been acknowledged by \citet{manski2014quantile} and \citet{manski2021econometrics}. }
In general, regardless of whether we consider a Bayes or minimax criterion, optimal decision rules defined in terms of their expected regret are singleton rules, i.e., given a sample, optimal decision rules either treat everyone, or no-one in the population. As an artificial example, suppose that the outcome of interest is $+1$ or $-1$ (success or failure) and imagine that we observe 100 successes and 99 failures (the status quo is zero for everyone). The empirical success (ES) rule, which is asymptotically optimal in terms of the mean of regret, suggests that everyone in the entire population should be treated.
If there is a swing of one outcome from $+1$ to $-1$ though, then the same ES rule now dictates that no-one should be treated. Such high sensitivity and aggressiveness of treatment decisions with respect to sampling uncertainty can incur a large welfare loss due to random sampling errors, especially when the sample size is small. Axiomatically, the mean regret paradigm fails to capture some important and empirically relevant factors in decision making, e.g.,  a decision maker may be worried about large welfare loss in particularly bad states of the world.


To address these concerns in the mean regret framework, this paper proposes a novel, simple and tractable approach to treatment choice with finite data  by optimizing a \textit{nonlinear transformation} of welfare regret. In what follows, we let $g(\cdot)$ be a nonlinear transformation of the regret. We assess the performance of each treatment rule via the expected value of the transformed regret loss that it delivers. In the spirit of \cite{Wald50}, this average nonlinear regret over realizations of the sampling process becomes the risk function. We refer to this risk as a \emph{nonlinear regret risk}. Due to the nonlinearity of $g(\cdot)$, information relating to other moments of the regret distribution is encoded in the risk function. 
For example, when $g(r)=r^2$, the associated risk function is the sum of the squared expected regret and the variance of regret, penalizing decision rules that lead to a high variance of regret.
We refer to this nonlinear regret risk as \emph{mean square regret}.\footnote{As shown in Remark \ref{rem:v.regret.equal.v.welfare}, the variance regret is equivalent to the variance of welfare. Thus, compared to the expected regret criterion, our mean square regret can also be viewed as penalizing rules with a high variance of welfare.} We stress that our nonlinear regret approach is not an ad-hoc modification of the existing mean regret paradigm. In fact, it finds counterparts in decision theory from  \cite{hayashi2008regret} and \cite{stoye2011axioms} who build a rich  axiomatic model for a class of regret-driven choices, including mean square regret and many other nonlinear regret criteria (corresponding to what \cite{hayashi2008regret} calls \emph{regret aversion}). These criteria can better accommodate decision maker's aversion to large welfare loss in bad states of the world. We discuss the connection of our approach and the results of \cite{hayashi2008regret} more in detail in Section \ref{sec:discussion.axiom}. 


This shift of criterion towards a nonlinear transformation of regret changes optimal rules drastically. We show that, for many nonlinear transformations and a large class of distributions, singleton rules are either incomplete or inadmissible, offering novel decision-theoretic justifications for implementing fractional treatment assignment rules. Our approach also stands out in terms of the tractability, simplicity and interpretability of the associated optimal rules. We are able to provide general formula on Bayes and minimax optimal rules based on nonlinear regret risks. For mean square regret, we derive closed-form solutions for both Bayes and minimax optimal decision rules, not only in Gaussian finite samples with known variance but also asymptotically. These optimal rules are fractional and surprisingly easy to calculate and interpret. For example, the minimax optimal treatment assignment fraction has the following  logistic structure:
\[
\frac{\exp\left(2\cdot1.23\cdot\text{t-statistic}\right)}{\exp\left(2\cdot1.23\cdot\text{t-statistic}\right)+1},
\]
where the $t$-statistic is for the average treatment effect estimated from experimental data, and which  coincides with the posterior probability-matching assignment under the least favorable
prior.\footnote{The posterior probability-matching assignment, known as the Thompson sampling algorithm \citep{thompson1933likelihood}, possesses a desirable exploration-exploitation property in bandit problems. Our results show that the posterior probability-matching assignment can be justified in terms of minimax mean square regret, even in the static treatment choice problem where the exploration motive does not exist.} For example, our asymptotically minimax optimal rule for the previous artificial example would only allocate 54\% of the population to the treatment, dropping to 46\% if one outcome switches. Due to their fractional nature, our rules are useful even beyond the treatment decision paradigm: researchers may conveniently interpret our rule as a summary statistic that quantifies the strength of evidence in support of the treatment versus control. 
See Section \ref{sec:discussion.interpretation} for further discussions on this matter.

Given a nonlinear regret risk and a prior for the underlying potential outcome distributions, we obtain the Bayes optimal rules. 
Consistent with our incompleteness and inadmissibility results, Bayes optimal rules are, in general, also fractional rules. 
For mean square regret, we show that the Bayes optimal rule is a \textit{tilted} posterior-probability matching rule, where the probability of random assignment corresponds to the posterior probability tilted by a certain weighting term. In a special case where the prior for the average treatment effect is supported only on two symmetric points, the tilting term is nullified and the Bayes optimal rule boils down to the Thompson-sampling type posterior-probability matching rule. 
For the minimax optimal rule in a Gaussian experiment with known variance, we can show that a least favorable prior is supported on two symmetric points. 
Hence, the minimax optimal rule follows the posterior-probability matching assignment rule, and is a logistic transformation of the sample mean. This minimax mean square regret rule is easy to compute and tuning- or hyper-parameter free. 

Imagine the outcome of interest now follows a normal distribution $N(1,1)$ with unit mean and unit variance, whereas the status quo is zero for everyone. In this scenario, the infeasible optimal rule is to treat everyone and the regret of any decision rule is supported on $[0,1]$. Suppose the planner observes one observation from the $N(1,1)$ distribution and needs to make a treatment choice. The ES rule is optimal in terms of expected regret, but could be far from ideal in terms of other features of the regret and welfare distribution. 
In fact, if the planner adopted ES rule, then there would be a mass of 16\% probability that she ended up with the largest possible regret of one (and the smallest possible welfare of zero). In contrast with the mean regret criterion commonly used in the literature, our mean square regret criterion penalizes rules with large variance of the regret distribution (and equivalently, large variance of welfare). If, instead, the planner implemented our proposed minimax rule, she could avert such high chance of welfare loss: the probability of incurring a regret larger than 0.95 (and a welfare smaller than 0.05) is only 1.4\%.  Also see Figure \ref{fig:regret.distribution} for a comparison of the distributions of the regret and welfare for ES rule and our proposed mean square regret minimax optimal rule.

\begin{figure}[htbp]
\begin{tabular}{ccccc}
\toprule 
\multirow{2}{*}{} & ES rule  & Our proposed minimax rule \tabularnewline
\midrule
Mean of regret & 0.1587 & 0.2087 &\tabularnewline
Standard deviation of regret & 0.3654 & 0.2698 \tabularnewline
Mean of welfare & 0.8413 & 0.7913 \tabularnewline
Standard deviation of welfare & 0.3654 & 0.2698 \tabularnewline
Mean square regret & 0.1587 & 0.1133 \tabularnewline
\bottomrule
\end{tabular}
    \centering
    \includegraphics[scale=0.5]{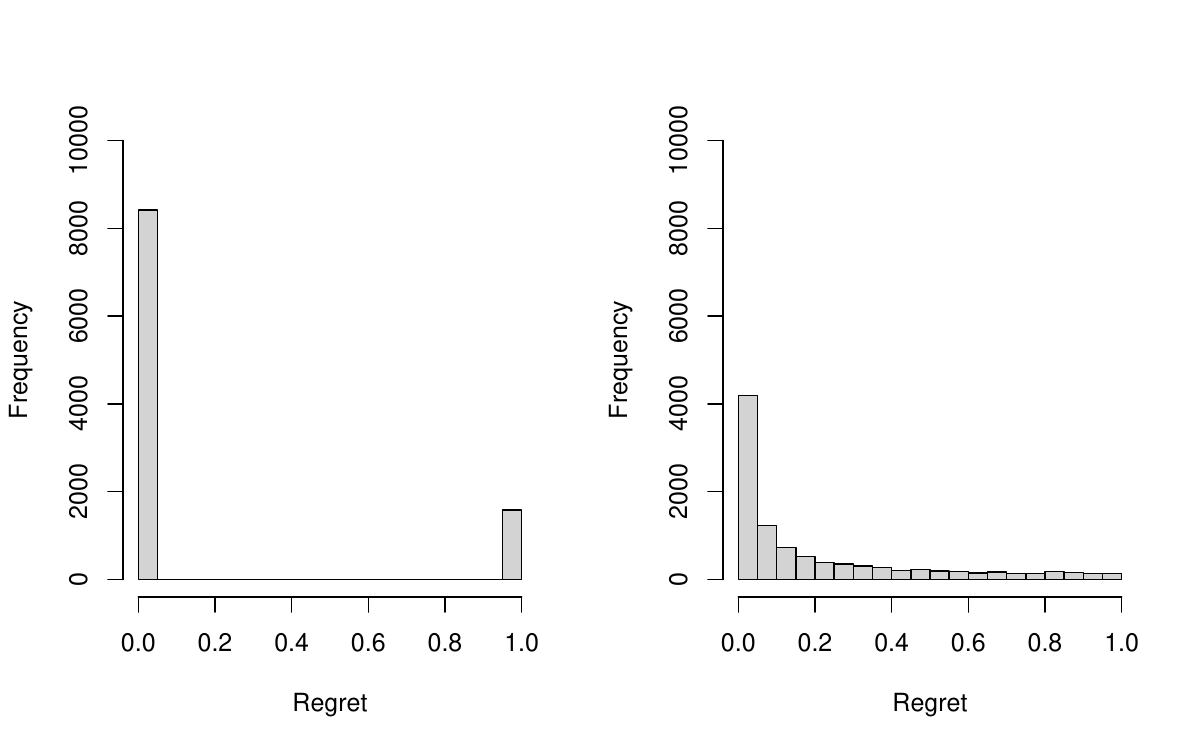}
     \includegraphics[scale=0.5]{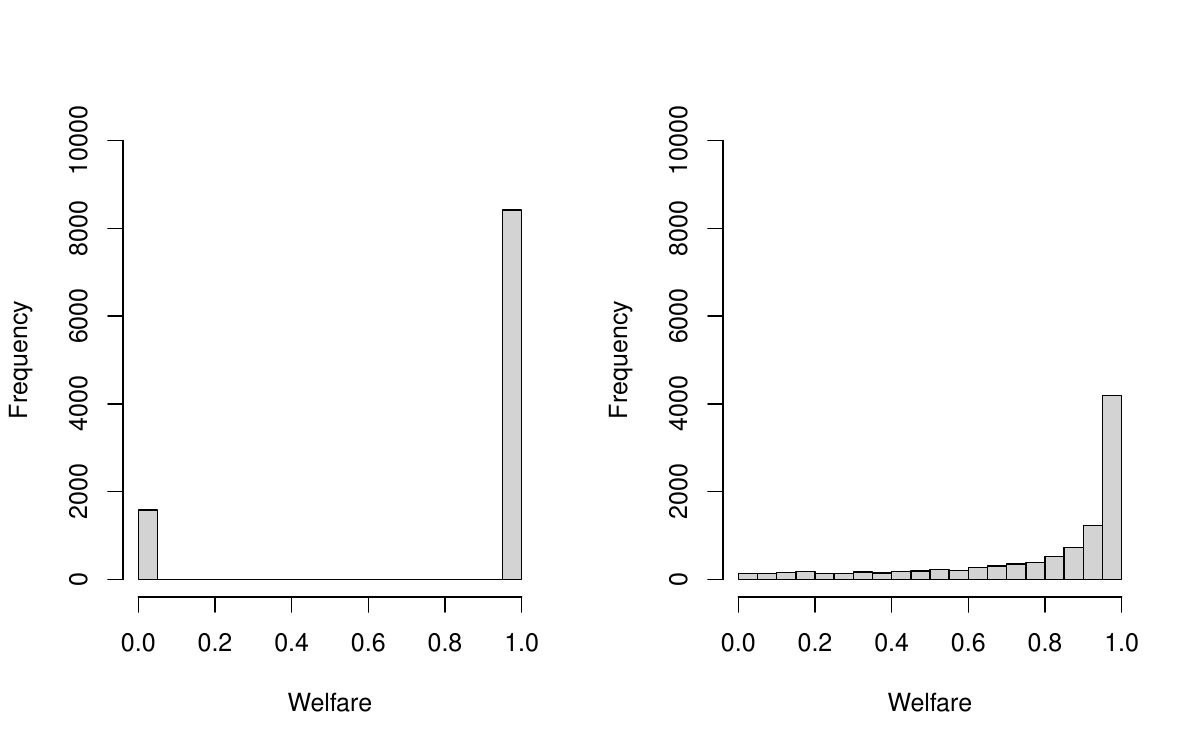}
    \caption{Summary statistics and empirical distributions of regret and welfare for the ES rule (left) and our proposed minimax optimal rule (right) in one $N(1,1)$ experiment, 10000 simulations.}
    \label{fig:regret.distribution}
\end{figure}

We demonstrate the usefulness of our approach in two applications. First, using our general theory, we derive a minimax optimal rule in a normal regression model with binary treatment, a specification frequently used by many practitioners. Second,  in practice, the planner often has a preference for singleton rules, and calculates a sufficient sample size for their randomized experiment based on these singleton rules. We show that implementing these singleton rules can lead to a large efficiency loss in terms of mean square regret. 

Following \cite{HiranoPorter2009, HiranoPorter2020}, we extend our finite sample results to a large sample setting by engaging with the limit experiments framework introduced by \cite{le2012asymptotic}. Even when potential outcome distributions are non-Gaussian but belong to a regular parametric class, we can obtain a Gaussian limit experiment with known variance.
Therefore, we can apply our results from a finite sample Gaussian experiment to a limit experiment and find feasible and asymptotically optimal rules with some efficient estimator of the parameters. Interestingly, in the limit experiment, the Bayes optimal rule under the mean square regret remains different from the minimax optimal rule, although the resulting mean square regret is quantitatively similar between the two rules. This is in contrast with the \emph{linear} regret risk, for which it is known that the Bayes optimal and minimax optimal rules in the limit experiment are the same empirical success rule.


Our justification for implementing a fractional treatment assignment rule differs from those given in the existing literature so far, which all use the conventional expected regret as the criterion. See, for example, \cite{manski20092009} for a detailed review of fractional rules with standard regret under ambiguity and other non-standard settings, including nonlinear welfare, interacting treatments, learning and other non-cooperative aspects. More specifically, when the linear welfare is partially identified, \cite{Manski2000,manski2005social,manski2007identification, Manski2007} shows that minimax regret optimal rules are fractional even with the true knowledge of the identified set.\footnote{There are two approaches to go without assuming the true knowledge of the identified set. One approach is to plug in an estimate of the identified set, treating it \emph{as if} it is the true object \citep{manski2013public,cassidy2019tuberculosis, manski2021probabilistic}. The other approach is to directly consider finite sample minimax regret optimal rules, which can be also fractional \citep{ stoye2012minimax,yata2021,manski2022identification} if model ambiguity is sufficiently large compared to statistical uncertainty.}  \cite{manski2007admissible} and \cite {manski20092009} justify fractional rules via a nonlinear welfare in a point-identified setting. \cite*{kock2022functional,kock2023treatment}   show that optimal treatment rules may be fractional if the decision maker targets a functional of the outcome distribution that is not quasi-convex. Fractional rules also arise when agents response with strategic behavior \citep{munro2020learning}. Our results justify fractional rules in a standard setting without model ambiguity, nonlinear welfare or other strategic aspects.

In a series of papers, Manski and Tetenov have explored optimal treatment rules in frameworks that go beyond the classical paradigm of the statistical decision theory laid out by \cite{Wald50}. \cite{manski1988ordinal,manski2011actualist} argues to maximize a functional of the welfare distribution that at least weakly respects stochastic dominance.  \cite{manski2014quantile} consider the performance of a statistical treatment rule measured in terms of quantiles of the welfare. Our approach is distinct from the approaches taken by the aforementioned papers. In particular, we  select treatment rules based on the distributions of their regret. Motivated by the risk aversion of policy makers, \cite{manski2007admissible} consider a concave and monotone transformation of welfare measured in terms of a binary outcome, and define regret in terms of the transformed welfare. That is, \cite{manski2007admissible} take a concave transformation of the welfare and keep the associated regret linear, while our approach looks at a nonlinear transformation of regret, keeping the welfare linear. These two approaches share a similar motivation: they advocate that decision makers may wish to take other aspects of the regret distribution into consideration when ranking decision rules. 
However, the strength of our approach lies in its technical tractability and practical implementability. See Online Appendix \ref{sec:app.compare} for further discussions on these matters.

The literature on the treatment choice problem has become an area of active research since the pioneering works of \cite{Manski2000, manski2002treatment,manski2004statistical} and \cite{Dehejia2005} introduced a decision theoretic framework to the problem. When the welfare is point-identified, minimax regret treatment choice rules for finite samples are derived, in different settings, by  \cite{schlag2006eleven}, \cite{stoye2009minimax},  \cite{tetenov2012statistical}, \cite{masten2023minimax}, and \cite{chen2024note}.  \cite{manski2014quantile} and \cite*{guggenberger2024minimax} study treatment choice problems when quantiles are of interest. \cite{HiranoPorter2009,HiranoPorter2020} introduce an asymptotic framework to analyze treatment rules with limit experiments. When the welfare is partially identified, \cite{Manski2000,manski2005social,manski2007identification, Manski2007, manski20092009} analyzes the treatment choice given the knowledge of the identified set. Treatment allocation analyses without the knowledge of the identified set but with finite sample data include \cite{stoye2012minimax}, \cite{christensen2020}, \cite{ishihara2021}, \cite{yata2021}, \cite{manski2022identification}, \cite{ishihara2023bandwidth} and \cite{montielolea2023decision}.  \cite{Chamberlain2011} investigates a Bayesian approach to treatment choice, and \cite{christensen2020} and \cite{Giacomini2021} discuss a robust Bayesian approach.

There is a growing literature on learning in the context of individualized treatment rules that map an individual's observable characteristics to a treatment. See \cite{manski2004statistical}, \cite{BhattacharyaDupas2012}, \cite{kitagawa2018should, KT21}, \cite{MT17}, \cite{AW17}, \cite{adjaho2022externally}, \cite{han2023optimal}, and \cite{cui2023individualized}, among others, for analyses in different settings.  Our analysis does not incorporate individuals' observable covariates. Since the nonlinear regret risk aggregates the conditional nonlinear regret risk additively, it is straightforward to incorporate observable discrete covariates into our analysis, i.e., an optimal individualized fractional assignment rule that applies an optimal fractional assignment rule to each subpopulation of individuals sharing the same covariate value. 

The rest of the paper is organised as follows. Section \ref{sec:setup} introduces our setup. Section \ref{sec:admit} studies the admissibility and completeness of decision rules with nonlinear regret risk. Section \ref{sec:finite} presents finite sample results on Bayes and minimax optimal decision rules. In Section \ref{sec:discuss}, we discuss the axiomatic foundation of our criteria and the interpretation of our rules as a measure of strength of evidence. Section \ref{sec:asymptotic} extends our results to the limit experiment framework and derives asymptotically optimal decision rules. Section \ref{sec:applications} applies our theory to an example of treatment choice in a normal regression model and an example of sufficient sample size calculation in randomized control trials. Section \ref{sec:conclude} concludes. Proofs and lemmas are reserved for the Appendix.

\section{Setup} \label{sec:setup}

Consider the assignment of a binary treatment $D \in \left\{ 1,0\right\}$ to an infinitely large population of individuals whose treatment effects can be heterogeneous. Let $Y(1)$ be the potential outcome when $D=1$ (with treatment)  and $Y(0)$ be the potential outcome when  $D=0$ (no treatment). Denote by $P\in \mathcal{P}$ the joint distribution of $(Y(1),Y(0))$, where $\mathcal{P}$ is a set of distributions under consideration. Define $\mu_{1} := \mathbb{E}[Y(1)]$ and $\mu_{0} := \mathbb{E}[Y(0)]$
as the means of the potential outcomes $Y(1)$ and $Y(0)$ under the distribution $P$. We assume that the welfare of the planner is determined by the mean outcome in the population. Defining the population average treatment effect as $\tau:=\tau (P):=\mu_{1}-\mu_{0}$, the infeasible optimal treatment policy is as follows:  allocate $D=1$ to each individual in the population if $\tau\geq0$ and allocate everyone $D=0$ otherwise.

We also assume the decision problem is nontrivial in the sense that the image of the mapping $\tau(P),P\in\mathcal{P}$ contains both positive and negative values.
Since the sign of $\tau$ is unknown, the planner collects an experimental sample of the observed outcomes of $n$ units randomly drawn from the population $P$, and the experimental design is known to the planner. The experiment generates a random vector 
$Z_{n}:=\left\{ Y_{i},D_{i}\right\} _{i=1}^{n} \in \mathbf{Z}_{n}$, where  $Y_{i}$ is the observed outcome of unit $i$,  $D_{i}$ is the treatment status of unit $i$, and $\mathbf{Z}_{n}$ is the sampling space. Let  $P^{n}$ be the sampling distribution of $Z_{n}$, which depends on $P$ as well as the known experimental design.\footnote{For example, for a randomized control trial with a known treatment probability $0<\pi<1$, the joint likelihood of $Z_n$ is written as $P^n(Z_n)=\prod_{i=1}^{n}(\pi f_1(Y_i(1)))^{D_i}((1-\pi)f_0(Y_i(0)))^{1-D_i}$, where $f_1$ and $f_0$ are the marginal densities of $Y(1)$ and $Y(0)$, respectively. Our setup also accommodates other experimental designs. The derivation of $P^n$ is analogous but may be more involved if the experimental design is complicated. Also, we use the uppercase letter $Z_n$ to denote a random vector and use lowercase letter $z_n$ to denote a realized value of $Z_n$.}
After observing data $Z_{n}$, the planner chooses a statistical
treatment rule $\hat{\delta}$ that maps $Z_{n} \in \mathbf{Z}_{n}$ 
to a real number between 0 and 1, i.e., 
\[
\hat{\delta}:\mathbf{Z}_{n}\to[0,1],
\]
where $\hat{\delta}(z_{n})$ is the proportion of the population receiving the treatment. Denote by $\mathcal{D}$ the set of all statistical decision rules under consideration. 
\begin{rem}
In our setting, the action space of the planner is $[0,1]$, instead of $\{0,1\}$. That is, the planner is allowed to make fractional treatment allocation and to differentiate the treatment statuses of individuals in the population. Following the terminology from \cite{manski2004statistical,manski2021econometrics}, we say $\hat{\delta}$ is a \emph{singleton} 
rule if $\hat{\delta}(z_{n})\in \{0,1\}$ for almost all $z_{n}\in\mathbf{Z}_{n}$. We say $\hat{\delta}$ is \emph{fractional} if $\hat{\delta}$ is not a singleton rule.
Intuitively, after observing data, a singleton rule either treats everyone, or no one in the population, whereas a fractional rule allocates an interior fraction of the population to the treatment, leaving the rest of the population untreated. A fractional rule $\hat{\delta}(z_{n})$ may be implemented according to some randomization device after observing $Z_{n}=z_{n}$.
\end{rem}
Applying the statistical treatment rule  $\hat{\delta}$ to the population yields a welfare of
\[
W(\hat{\delta}):=W(\hat{\delta},P):=\mu_{1}\hat{\delta}+\mu_{0}(1-\hat{\delta})
\]
to the planner. The infeasible optimal treatment policy that maximizes welfare is $\delta^{*}:=\mathbf{1}\{\tau\geq0\}$.
Following \cite{Savage51} and \cite{manski2004statistical}, we define the regret of $\hat{\delta}$ as its welfare compared
to the welfare of $\delta^{*}$, i.e.,
\[
Reg(\hat{\delta}):=Reg(\hat{\delta},P):=\tau[1\{\tau\geq0\}-\hat{\delta}].
\]
Since $Reg(\hat{\delta})$ is a random object that depends on realizations of the random vector $Z_{n}$,  \cite{manski2004statistical} follows \cite{Wald50} in measuring the performance of $\hat{\delta}$ using its risk, i.e., the expected regret across realizations of the sampling process:
\[
R(\hat{\delta},P):=\mathbb{E}_{P^{n}}[Reg(\hat{\delta})]:=\int_{z_{n}\in \mathbf{Z}_{n}} Reg(\hat{\delta}(z_{n}))dP^{n}(z_{n}),
\]
where $\mathbb{E}_{P^{n}}$ denotes the expectation with respect to $P^{n}$.

The risk criterion $R(\hat{\delta},P)$ ranks treatment rules according to their mean regret. We, instead, consider a planner whose assessment of the performance of statistical treatment rules depends not only on the mean of regret but also on some other features of the regret distribution. To take other features of the regret distribution into consideration, we look at the nonlinear transformation of regret:
\[
g(Reg(\hat{\delta})),
\]
where $g:\mathbb{R}^{+}\rightarrow\mathbb{R}^{+}$ is some nonlinear function. The planner's preference over statistical decision rules $\hat{\delta}$ is measured by the expected value of $g(Reg(\hat{\delta}))$ with respect to realizations of $Z_n$:
\begin{equation}\label{eq:nonlinearregretrisk}
R_g(\hat{\delta},P) := \mathbb{E}_{P^n}[ g(Reg(\hat{\delta}))].
\end{equation}
We refer to the criterion $R_g(\hat{\delta},P)$ as the \textit{nonlinear regret risk}.\footnote{The conventional mean regret risk function is a special case of our approach by taking $g$ to be linear.} Due to the nonlinearity of $g(\cdot)$, $R_g(\hat{\delta},P)$ depends not only on the mean but also on other features of the regret distribution, including its higher-order moments.  For instance, if we specify the quadratic function $g(r)=r^2$,  the squared regret is
\[
(Reg(\hat{\delta}))^2=\tau^{2}[1\{\tau\geq0\}-\hat{\delta}]^{2}.
\]
This squared regret constitutes the new loss function, and we can evaluate the performance of $\hat{\delta}$ via \emph{mean square regret}:
\[
R_{sq}(\hat{\delta},P):=\tau^{2}\mathbb{E}_{P^{n}}[1\{\tau\geq0\}-\hat{\delta}]^{2}.
\] 

\begin{rem}\label{rem:v.regret.equal.v.welfare}
Similar to the classical estimation theory, we can decompose 
\begin{align*}
R_{sq}(\hat{\delta},P) & =\left[R(\hat{\delta},P)\right]^{2}+V(\hat{\delta},P),
\end{align*}
where $R(\hat{\delta},P)$ is the mean regret risk, and $V(\hat{\delta},P)$
is the variance of the regret $Reg(\hat{\delta})$ defined as
\begin{align*}
V(\hat{\delta},P) & :=\mathbb{E}_{P^{n}}\left[\tau(1\{\tau\geq0\}-\hat{\delta})-\tau\mathbb{E}_{P^{n}}[1\{\tau\geq0\}-\hat{\delta}]\right]^{2}\\
 & =\mathbb{E}_{P^{n}}\left[\tau\hat{\delta}-\tau\mathbb{E}_{P^{n}}[\hat{\delta}]\right]^{2}.
\end{align*}
Thus, ranking
treatment rules by the mean square regret criterion has the benefit
of penalizing rules with large regret variance. As
\[
\tau\hat{\delta}-\tau\mathbb{E}_{P^{n}}[\hat{\delta}]=W(\hat{\delta},P)-\mathbb{E}_{P^{n}}\left[W(\hat{\delta},P)\right],
\]
for each rule $\hat{\delta}$, the variance of its regret
in fact equals the variance of the resulting welfare. The mean square
regret criterion has the following equivalent representation
\begin{align*}
R_{sq}(\hat{\delta},P) & =\left[R(\hat{\delta},P)\right]^{2}+\underset{\text{variance of welfare}}{\underbrace{\mathbb{V}_{P^n}\left[W(\hat{\delta},P)\right]}},
\end{align*}
where $\mathbb{V}_{P^n}[\cdot]$ denotes the variance under sampling distribution $P^n$. Thus, compared to the mean regret criterion $R(\hat{\delta},P)$,
the mean square regret penalizes rules that lead to a large variance
of welfare. A decision maker who uses mean square regret criterion
is averse to the volatility of their welfare (with respect to sampling
uncertainty), while a decision maker who prefers the mean
regret criterion ignores the volatility of their welfare.
\end{rem}

\section{Incompleteness and inadmissibility of singleton rules} \label{sec:admit}

Viewing the nonlinear regret risk $R_g(\hat{\delta},P)$ defined in (\ref{eq:nonlinearregretrisk}) as the risk criterion within Wald's framework of statistical decision theory, we introduce the following standard definition of admissibility of a statistical treatment rule and the essential completeness of a class of rules. 
\begin{defn}[Admissibility, inadmissibility and essential completeness under nonlinear regret risk]\  
\begin{itemize}
\item[(i)] A statistical treatment choice rule $\hat{\delta}:\mathbf{Z}_n \to [0,1]$ is admissible under the nonlinear regret risk $R_g(\hat{\delta},P) = \mathbb{E}_{P^n}[ g(Reg(\hat{\delta}))]$ if no $\hat{\delta}^{\prime} \neq \hat{\delta}$ dominates $\hat{\delta}$, i.e., there is no $\hat{\delta}^{\prime}$ such that $R_g(\hat{\delta}^{\prime},P) \leq R_g(\hat{\delta},P)$ holds for all $P\in \mathcal{P}$ with the inequality strict for some $P\in\mathcal{P}$. 

\item[(ii)] A statistical treatment choice rule $\hat{\delta}:\textbf{Z}_n \to [0,1]$ is inadmissible under the nonlinear regret $R_g(\hat{\delta},P)$ if there exists
a decision rule $\hat{\delta}^{\prime} \neq \hat{\delta}$ that dominates $\hat{\delta}$.

\item[(iii)] A class $C^*$ of statistical treatment rules, $C^*\subset\mathcal{D}$, is essentially complete in $\mathcal{P}$, if, for an arbitrary rule $\hat{\delta}\in\mathcal{D}$ not in $C^*$, there exists a $\hat{\delta}_0\in C^*$ such that 
$R_g(\hat{\delta},P)\geq R_g(\hat{\delta}_0,P)$ for all $P\in\mathcal{P}$.
\end{itemize}

\end{defn}
As in standard statistical decision theory, admissibility defined through nonlinear regret
is a minimal requirement that a desirable statistical treatment rule should satisfy. The notion of an essentially complete class simplifies the task of finding a good decision rule, as there is no need to consider rules outside this class.

\begin{assumptionG}[Nonlinear transformation]
The nonlinear transformation $g:\mathbb{R}^+ \to \mathbb{R^{+}}$ is continuously differentiable and $g(\cdot)$ is strictly increasing on $\mathbb{R}^{+}$ with $g'(0) = 0$.
\end{assumptionG}

Assumption G puts mild restrictions on the shape of the nonlinear transformation. Together with Assumption G, our first theorem shows that, in terms of nonlinear regret risk, a wide class of singleton rules are not essentially complete, implying that fractional rules cannot be eliminated from consideration in general. Recall for each $P\in\mathcal{P}$, $P^n$ denotes the distribution of data $Z_n$ and depends on $P$ while keeping the experimental design intact.\footnote{For example, for a randomized control trial with treatment assignment probability $\pi$, we have that  for $P_*\in\mathcal{P}$, 
$P^n_*(Z_n)=\prod_{i=1}^{n}(\pi f_{1}^*(Y_i(1)))^{D_i}((1-\pi)f^*_{0}(Y_i(0)))^{1-D_i}$.} For an event $A$, write $\mathbb{P}_{P^{n}}(A):=\int\mathbf{1}\{z_{n}\in A\}dP^{n}(z_{n})$ as the probability of event $A$ under the sampling distribution $P^n$.

\begin{thm} \label{thm:incomplete}





Let $\mathcal{D}_{S}$ be a class of singleton decision rules such that, for some $\{P_1, P_2 \} \subset \mathcal{P}$ that satisfy $\tau(P_1) = c > 0$ and $\tau(P_2) = - c < 0$, $\mathcal{D}_{S}$ can be represented as  $\mathcal{D}_{S}=\mathcal{D}_{S}^{1}\cup\mathcal{D}_{S}^{2}\cup\mathcal{D}_{S}^{3}$, where
\begin{align*}
\mathcal{D}_{S}^1 & := \{ \hat{\delta}_S \in \{0,1\} : \mathbb{P}_{P_1^n}( \hat{\delta}_S = 1) < 1, \mathbb{P}_{P_2^n}( \hat{\delta}_S = 1) > 0 \}, \\
\mathcal{D}_{S}^2 &:= \{ \hat{\delta}_S \in \{0,1\} : \mathbb{P}_{P_1^n}( \hat{\delta}_S = 1) = \mathbb{P}_{P_2^n}( \hat{\delta}_S = 1) = 1 \}, \\
\mathcal{D}_{S}^3 & := \{ \hat{\delta}_S \in \{0,1\} : \mathbb{P}_{P_1^n}( \hat{\delta}_S = 1) = \mathbb{P}_{P_2^n}( \hat{\delta}_S = 1) = 0 \}.
\end{align*}
If Assumption G holds, then $\mathcal{D}_{S}$ is not essentially
complete in $\mathcal{P}$.
\end{thm}

The incompleteness result in Theorem $\ref{thm:incomplete}$ is general, without relying on parametric specifications for $P$ or functional form restrictions on the decision rules.
The class of singleton rules considered in Theorem $\ref{thm:incomplete}$ is also large, comprising both ``non-degenerate singleton rules'' ($\mathcal{D}_S^1$) and ``degenerate singleton rules'' ($\mathcal{D}_S^2$ and $\mathcal{D}_S^3$) under some pair of distributions, one  with positive average treatment effect ($P_1$), and the other negative ($P_2$).\footnote{The assumption that the pair share the same absolute value of treatment effect is not essential. The proof still goes through as long as we have $\tau(P_1) = c_1 > 0$ and $\tau(P_2) = - c_2 < 0$ for some $c_1,c_2>0$.} Intuitively, any singleton rule with a positive probability of making mistakes under such pair of distributions is contained in $\mathcal{D}_S^1$. For the same pair of distributions, $\mathcal{D}_S^2$ would be the class of rules that treats everyone in the whole population with probability one, e.g.,  $\hat{\delta}_S=1$, which treats everyone and ignores data. In contrast, $\mathcal{D}_S^3$ is the class of rules that treats no one in the  population with probability one under $P_1$ and $P_2$, and contains the trivial rule $\hat{\delta}_S=0$. 

The results of Theorem $\ref{thm:incomplete}$ contrast sharply with the known result on the essentially completeness of singleton rules in more standard formulations of the treatment choice problem, where the (negative) expected welfare corresponds to the risk criterion in Wald's framework of statistical decision theory. For hypothesis testing problems with monotone likelihood ratio distributions, \cite{karlin1956theory} show that the class of singleton threshold rules is essentially complete. As exploited in \cite{HiranoPorter2009} and \cite{tetenov2012statistical}, the essential completeness of singleton threshold rules carries over to the treatment choice problem, implying that decision makers can discard fractional rules when in search of a good rule. Applying Theorem \ref{thm:incomplete} to monotone likelihood ratio distributions, we can conclude that the same class of singleton threshold rules are no longer essentially complete when it comes to many nonlinear regret risk criteria. Therefore, fractional rules cannot be eliminated from the consideration. Theorem \ref{thm:incomplete} is established with the following important observation:

\begin{lem}\label{lem:main}
Suppose conditions of Theorem \ref{thm:incomplete} hold. For each singleton rule $\hat{\delta}_{S}\in\mathcal{D}_{S}^1$, there exists
a fractional rule 
\begin{equation}\label{eq:fractional.rule.form}
\hat{\delta}_{S,\lambda}:=(1-\lambda)\hat{\delta}_{S}+\lambda(1-\hat{\delta}_{S}),  
\end{equation}
for some $\lambda\in(0,1)$ such that
\begin{equation}
R_{g}(\hat{\delta}_{S,\lambda},P)<R_{g}(\hat{\delta}_{S},P),\text{ for each }P\in\{P_1,P_2\},\label{eq:incomplete}
\end{equation}
where $\{P_1,P_2\}$ are defined in Theorem \ref{thm:incomplete}.
\end{lem}
Lemma \ref{lem:main} reveals that if we focus on a binary class of distributions $\{P_1,P_2\} \subset \mathcal{P}$ with opposite treatment effects, then the class of nondegenerate singleton rules $\mathcal{D}_S^1$ are dominated.  Theorem \ref{thm:incomplete} utilizes this dominance result in the restricted class of distributions $\{P_1,P_2\}$ to show that $\mathcal{D}_S$ is not essentially complete in the unrestricted (possibly infinite) class of distributions $\mathcal{P}$. 

In general, proving the inadmissibility of a rule (in $\mathcal{P}$) is often more demanding, as one needs to find a dominating fractional rule not only for $\{P_1,P_2\}$ but all $P\in\mathcal{P}$. The next theorem confirms that a large set of singleton threshold rules are indeed inadmissible, if we impose stronger shape restrictions on $g$ and more distributional assumptions for the data.


\begin{thm} \label{thm:inadmissibility}
Suppose the nonlinear transformation $g:\mathbb{R}^{+}\to\mathbb{R}^{+}$ is a nonzero, homogeneous function of degree $\alpha>1$. Moreover, there exists an estimator $\hat{\tau}:\mathbf{Z}_{n}\rightarrow\mathbb{R}$ for $\tau$, and the pdf or pmf of $\hat{\tau}$ is a member of the following one-parameter exponential family: 
\begin{equation}\label{eq:exponential.family}
f(x|\tau)=A(\tau)\exp^{x\tau}h(x),\text{ }x\in\mathbb{R},
\end{equation}
where $A(\cdot)$, $h(\cdot)$ are real valued functions, and $\tau\in[-\bar{\tau},\bar{\tau}]\subseteq\mathbb{R}$ for some $\bar{\tau}>0$. Then, any singleton threshold rule
\begin{equation}\label{eq:singelton.exponential.family}
\hat{\delta}_{t}(\hat{\tau})=\mathbf{1}\left\{ \hat{\tau}\geq t\right\},
\end{equation}
where $t$ is in the interior of the support of $\hat{\tau}$, is inadmissible.
\end{thm}

Theorem \ref{thm:inadmissibility} is again in sharp contrast with the existing literature. For the same one-parameter exponential family distributions, results from \cite{karlin1956theory} imply that singleton rules of the form \eqref{eq:singelton.exponential.family} are in fact admissible for the decision criterion of mean regret. Our results offer an opposite conclusion, demonstrating that the ranking of statistical treatment rules crucially depends on which features of the regret distribution the decision maker wishes to exploit. The homogeneity assumption on the shape of $g$ is a particularly convenient one that enables us to explicitly pin down a dominating rule in the form of \eqref{eq:fractional.rule.form} over $\hat{\delta}_{t}$, but is not required and can be  considerably weakened. This class of homogeneous functions also coincides with what \cite{hayashi2008regret} calls regret aversion.\footnote{Our proof is constructive, and requires boundedness of the parameter space to pin down explicitly a fractional rule that dominates \eqref{eq:singelton.exponential.family} for all values of $\tau$ in the parameter space.}

\section{Finite sample optimality} \label{sec:finite}

We measure the performance of a rule $\hat{\delta}\in \mathcal{D}$ by its nonlinear regret  risk $R_{g}(\hat{\delta},P)$, which depends on the true unknown $P$. In this section we look at two optimality criteria and derive general results on optimal rules for these criteria. We illustrate the usefulness of our results using specific parametric models.

\subsection{Bayes optimality}

\begin{defn} [Bayes nonlinear risk and the Bayes optimal rule] \label{def:bayes.optimality}
Let $\pi$ be a prior distribution on $P\in\mathcal{P}$. The Bayes nonlinear (regret) risk of $\hat{\delta}$ with respect to the prior $\pi$ is
\[
r_{g}(\hat{\delta},\pi):=\int_{P\in\mathcal{P}}R_{g}(\hat{\delta},P)d\pi(P).
\]
A Bayes optimal rule $\hat{\delta}_{\pi}$ with respect
to the prior $\pi$ is such that
\[
r_{g}(\hat{\delta}_{\pi},\pi)=\inf_{\hat{\delta}\in\mathcal{D}}r_{g}(\hat{\delta},\pi).
\]
Moreover, we say that a prior distribution $\pi$ is \emph{least favorable} if $r_{g}(\pi,\hat{\delta}_{\pi})\geq r_{g}(\pi^{\prime},\hat{\delta}_{\pi^{\prime}})$
for all prior distributions $\pi^{\prime}$.\footnote{Our notion of Bayes optimality is specifically attached to the risk function $R_{g}(\hat{\delta},P)$, and is not an acronym for Bayes welfare criterion. See also \cite{ferguson1967mathematical,berger1985statistical} for precedents of attaching Bayes optimality to a specific risk function in statistical decision theory.}
\end{defn}

We now characterize the Bayes optimal rule for the Bayes nonlinear risk. It turns out that under mild restrictions on the nonlinear transformation $g$, the associated Bayes optimal rule is also fractional. To proceed, let $\pi(P|z_{n})$ be the posterior distribution of $P$ given a prior
$\pi$ and $Z_{n}=z_{n}$. 

\begin{thm} \label{thm:bayes.finite} Suppose Assumption G holds, and the following conditions are true:
\begin{itemize}

\item[(i)] There exists some treatment rule $\tilde{\delta}\in\mathcal{D}$ such that
$R_{g}(\tilde{\delta},P)$ is finite.

\item[(ii)] For almost all $z_{n}\in\mathbf{Z}_{n}$, the posterior distribution
$\pi(P|z_{n})$ puts nonzero probability mass on both $\left\{ P\in\mathcal{P}:\tau(P)>0\right\} $
and $\left\{ P\in\mathcal{P}:\tau(P)<0\right\} $.
\end{itemize}
%
Then for almost all $z_{n}\in\mathbf{Z}_{n}$, the Bayes optimal rule $\hat{\delta}_{\pi}$  exists, is fractional, and satisfies 

\begin{equation}
\int\left[\tau(P)g^{\prime}\Big(\tau(P)(\mathbf{1}\{\tau(P)\geq0\}-\hat{\delta}_{\pi})\Big)\right]d\pi(P|z_{n})=0.\label{eq:bayes.2} 
\end{equation}
\end{thm}


\begin{rem}
In general, the Bayes optimal rule depends on the nonlinear transformation $g$ and the model specification for $P$. The calculation of the posterior expectation in (\ref{eq:bayes.2}), which requires integration with respect to the posterior distribution of $P$, can be complicated. To gain further insight, consider the simple case where $g(r)=r^2$ and $P=P_{\tau}$ is parameterized by the one dimensional parameter $\tau\in\mathbb{R}$, where $\tau=\mathbb{E}[Y(1)]-\mathbb{E}[Y(0)]$ (for example, the outcome is normal with known variance). It follows that the prior distribution is indexed by $\tau$ and written as $\pi(\tau)$, and the Bayes optimal rule with respect to the Bayes mean square regret
\[
r_{sq}(\hat{\delta},\pi):=\int R_{sq}(\hat{\delta},P_{\tau})d\pi(\tau)
\]
is characterized as 

\begin{equation}
\int\left[\tau^2 (\mathbf{1}\{\tau\geq0\}-\hat{\delta}_{\pi})\right]d\pi(\tau|z_{n})=0,
\end{equation}
where $\pi(\tau|z_{n})$ is the posterior distribution of $\tau$ given the prior $\pi(\tau)$ and data $Z_{n}=z_{n}$, with $Z_{n}\sim P^{n}_{\tau}$.

Further to this, if the prior $\pi(\tau)$ is supported on two symmetric points $\tau\in\left\{ a,-a\right\} $
for some $a>0$, it follows that
\[
\hat{\delta}_{\pi}(z_{n})=\frac{\int a^{2}1\{\tau\geq0\}d\pi(\tau|z_{n})}{\int a^{2}d\pi(\tau|z_{n})}=\underset{\text{posterior probability that treatment effect is non-negative}}{\underbrace{\int1\{\tau\geq0\}d\pi(\tau|z_{n})}},
\]
which is the exact form of the posterior probability matching rule, as used by \citet{thompson1933likelihood}. If the prior is
not supported on two symmetric points, it holds that 
\begin{equation}\label{eq:bayes.parametric}
\hat{\delta}_{\pi}(z_{n}) =\frac{\int\tau^{2}1\{\tau\geq0\}d\pi(\tau|z_{n})}{\int\tau^{2}d\pi(\tau|z_{n})}=\underset{\text{posterior probability matching}}{\underbrace{\int1\{\tau\geq0\}d\pi(\tau|z_{n})}}\underset{\text{weight}}{\underbrace{\frac{\int\tau^{2}d\pi(\tau|z_{n},\tau\geq0)}{\int\tau^{2}d\pi(\tau|z_{n})}}},
\end{equation}
where $\pi(\tau|z_{n},\tau\geq0)$ denotes the posterior distribution
of $\tau$ conditional on $\tau\geq0$. Thus, for the mean square regret,
the Bayes optimal rule is a \emph{tilted} version of the posterior probability matching rule. 
\end{rem}
\begin{rem}
In contrast, for the \emph{linear} regret risk $R(\hat{\delta},P)$,
the Bayes optimal rule is
\[
\hat{\delta}(z_{n})=\begin{cases}
\hat{\delta}(z_{n})=1, & \int\tau(P) d\pi(P|z_{n})>0,\\
\hat{\delta}(z_{n})\in[0,1], & \int\tau(P) d\pi(P|z_{n})=0,\\
\hat{\delta}(z_{n})=0, & \int\tau(P) d\pi(P|z_{n})<0,
\end{cases}
\]
which is a singleton rule in general. 
\end{rem}

We now provide a simple example for which we derive the finite sample Bayes optimal rule with respect to a flat prior.  This example also sheds some light on the form of the Bayes optimal rule in large samples, which is discussed in Section \ref{sec:asymptotic}.

\begin{example}[Testing an innovation with normal outcome and mean square regret]\label{exa:1} 
Let $g(r)=r^{2}$. Suppose the distribution of $Y(0)$ is
known to the planner and without loss of generality, $\mathbb{E}[Y(0)]=0$.
Therefore, the planner only needs to learn $\mathbb{E}[Y(1)]$ and in the experimental design she allocates all units to the treatment. Let $\bar{Y}_{1}$
be the sample average of observed outcomes. Assume  $\bar{Y}_{1}\sim N(\tau,1)$ is normally distributed
with an unknown mean $\tau\in\mathbb{R}$
and a known variance normalized to one, with the likelihood function
\begin{equation}
f(\bar{y}_{1}|\tau)=\sqrt{\frac{1}{2\pi}}\exp\left(-\frac{1}{2}\left[\left(\bar{y}_{1}-\tau\right)^{2}\right]\right),\forall\bar{y}_{1}\in\mathbb{R}.\label{eq:pdf.normal}
\end{equation}
\end{example}

\begin{prop} \label{prop:bayes.flat}
In Example \ref{exa:1}, consider the uniform (improper) prior $\pi_{f}$ on $\tau$. Then the Bayes treatment
rule with respect to the mean square regret is 
\[
\hat{\delta}_{\pi_{f}}(\bar{Y}_{1})=\Phi\left(\bar{Y}_{1}\right)\left[ 1+\bar{Y}_{1}\cdot \Psi(\bar{Y}_{1}) \right],
\]
where $\Psi(x):= \frac{\phi\left(x\right)}{\Phi\left(x\right) \left( 1+x^{2} \right)} > 0$ for any $x \in \mathbb{R}$, and where $\Phi(\cdot)$ and $\phi(\cdot)$ are the cdf and pdf of a standard normal random variable, respectively.
\end{prop}

Proposition \ref{prop:bayes.flat} is a direct application of Theorem \ref{thm:bayes.finite}.  Since the prior is flat, the `posterior density' is proportional to the likelihood (\ref{eq:pdf.normal}). 
The form of the Bayes optimal rule then follows (\ref{eq:bayes.parametric}). The Bayes optimal rule $\hat{\delta}_{\pi_{f}}$ is a product of two terms. The first term, $\Phi(\bar{Y}_{1})$, is the posterior probability that the treatment effect is positive given the uninformative prior, and corresponds to the posterior probability matching rule. The second term, $( 1+\bar{Y}_{1}\cdot \Psi(\bar{Y}_{1}))$, adjusts the first term upwards if  $\Bar{Y}_{1}>0$, and adjusts it downwards if $\Bar{Y}_{1}<0$ (note that $\Psi(x) > 0$). Therefore, this Bayes optimal rule tilts the posterior probability matching rule and assigns treatment with a probability closer to zero or one. Also see Table \ref{tab:1} and Figure \ref{fig:2} for the magnitudes of the probability assignment of the Bayes optimal rule and posterior probability matching rule with respect to the uniform prior.

\subsection {Minimax optimality}
As an alternative to Bayes rule, this section studies minimax optimal rule for nonlinear regret risk.\footnote{As shown by \cite{Savage51,manski2004statistical}, maximin welfare criterion can be ultra-pessimistic and often leads to an optimal rule that always treats no-one in the population. Such ultra-pessimism carries over to maximin criterion applied to a nonlinear transformation of welfare. Echoing previous findings on minimax expected regret \citep[e.g.,][]{stoye2009minimax,tetenov2012statistical}, our results show that the ultra-pessimism also does not occur for the minimax criterion applied to many nonlinear regret risks.}

\begin{defn} [Minimax optimal rule] \label{def:minimax.optimality}
A minimax optimal rule $\hat{\delta}^{*}$ is such that
\[
\sup_{P\in\mathcal{P}}R_{g}(\hat{\delta}^{*},P)=\underset{\hat{\delta}\in\mathcal{D}}{\inf}\sup_{P\in\mathcal{P}}R_{g}(\hat{\delta},P).
\]
\end{defn}

The following proposition characterizes the minimax optimal rule as a Bayes rule under a least favorable prior.
\begin{prop}[\citet{lehmann2006theory}]\label{prop:minimax}
Suppose $\pi$ is a distribution on $P$ such
that
\[
r_{g}(\hat{\delta}_{\pi},\pi)=\sup_{P\in\mathcal{P}}R_{g}(\hat{\delta}_{\pi},P).
\]
Then: (i) $\hat{\delta}_{\pi}$ is minimax; (ii) $\pi$ is least favorable.
\end{prop}

Proposition \ref{prop:minimax} is a direct result of \citet[Theorem 5.1.4]{lehmann2006theory}.
%
Using Proposition \ref{prop:minimax}, we can attempt to find the minimax optimal rule by adopting a `guess-and-verify' approach: guess a least favorable prior and derive its associated Bayes optimal rule; verify that the resulting Bayes nonlinear regret risk equals the worst frequentist nonlinear regret risk of the Bayes optimal rule. In general, it can still be difficult to guess the least favorable distribution. However, in many parametric models, the support of the least favorable distribution is often discrete and finite, or the minimax optimal rule has a constant frequentist risk across its parameter space. See, for example, \citet{kempthorne1987numerical}. This greatly simplifies the problem. We now demonstrate the minimax optimal rule for Example \ref{exa:1}.

\begin{thm}
\label{thm:minimax}In Example \ref{exa:1}, a finite sample minimax
treatment rule is
\[
\hat{\delta}^{*}(\bar{Y}_{1})=\frac{\exp\left(2\tau^{*}\bar{Y}_{1}\right)}{\exp\left(2\tau^{*}\bar{Y}_{1}\right)+1},
\]
where $\tau^{*}\approx 1.23$, which solves

\begin{equation}
\sup\limits _{\tau\in[0,\infty)}\frac{1}{2}\tau^{2}\mathbb{E}\left[\frac{1}{\exp\left(2\tau\bar{Y}_{1}\right)+1}\right],\label{eq:minimax.1}
\end{equation}
or, equivalently, solves
\begin{equation}
\sup\limits _{\tau\in[0,\infty)}\tau^{2}\mathbb{E}\left[\left(\frac{1}{\exp\left(2\tau\bar{Y}_{1}\right)+1}\right)^{2}\right],\label{eq:minimax.2}   
\end{equation}
where the expectation is with respect to $\bar{Y}_{1}\sim N(\tau,1)$. Moreover, a least favorable prior $\pi^{*}$ on $\tau$ is a two-point prior such that $\pi^{*}(\tau^*)= \pi^{*}(-\tau^*)=\frac{1}{2}$.
\end{thm}

\begin{rem}The minimax optimal rule is a simple logistic transformation of the sample mean and is straightforward to calculate. Moreover, the minimax optimal rule agrees with the posterior probability matching rule, i.e., the treatment probability equals the posterior probability that the treatment effect is positive with respect to the least favorable prior, which is supported on two symmetric points around zero. In this way, we justify the posterior probability matching rule in a static environment without multiple exploration phases.
\end{rem}

\begin{rem}\label{rem:proof}
On a more technical note, the proof of Theorem \ref{thm:minimax} relies on some different techniques from the existing treatment choice literature. For the mean regret criterion, singleton threshold rules form an essential complete class, so the minimax optimal rule with respect to mean regret can be found by directly minimizing the worst-case regret with respect to the threshold, without figuring out a least favorable prior \citep{tetenov2012statistical}. \cite{stoye2009minimax} finds that a least favourable prior must be two-point symmetrically supported around zero. For this class of priors, the Bayes rule is always the ES rule. Therefore, the result of \cite{stoye2009minimax} also does not need to pin down the exact location of the least favourable prior. However, for mean square regret, singleton rules are inadmissible, and it becomes essential to find the exact location of a least favorable prior. To find a least favorable prior and by observing the form of the mean square regret, we first conclude that  a least favorable prior
$\pi^{*}$ for mean square regret is also symmetric, such that
\[
\pi^{*}\left(\tau\right)=\frac{1}{2},\pi^{*}\left(-\tau\right)=\frac{1}{2},
\]
for some $0<\tau<\infty$. Within this set of candidate least favorable priors $\pi_{\tau}^{*}$ indexed by $\tau$, Theorem \ref{thm:bayes.finite} implies the Bayes optimal rules admit the form $\hat{\delta}_{\pi^{*}_{\tau}}(\bar{Y_{1}})=\frac{\exp\left(2\tau\bar{Y}_{1}\right)}{\exp\left(2\tau\bar{Y}_{1}\right)+1}$. Furthermore, $r_{sq}(\hat{\delta}_{\pi^{*}_{\tau}},\pi^{*}_{\tau})$  follows the form in (\ref{eq:minimax.1}), and is equivalent to the form in (\ref{eq:minimax.2}). Then, we guess that the least favorable prior is 
\[
\pi^{*}\left(\tau^{*}\right)=\frac{1}{2},\pi^{*}\left(-\tau^{*}\right)=\frac{1}{2},
\]
where $\tau^{*}$ solves (\ref{eq:minimax.1}) or (\ref{eq:minimax.2}).\footnote{The exact location of the least favorable prior for the mean regret criterion is also different from $\tau^*$ and solves $\inf_{\tau\in[0,\infty)}\tau\Phi(-\tau)$ for Example \ref{exa:1}. } With this guess of the least favorable prior, we further establish that the following condition holds:
\begin{condition}\label{cond:1}
 $r_{sq}(\hat{\delta}^{*},\pi^{*})=\sup_{\tau\in[0,\infty)}R_{sq}\left(\hat{\delta}^{*},P_{\tau}\right)$.
\end{condition}
The left-hand side of Condition \ref{cond:1} is the Bayes mean square regret of $\hat{\delta}^{*}$ with respect to our hypothesized least favorable prior $\pi^{*}$, and the right-hand side of Condition \ref{cond:1} is the worst mean square regret of $\hat{\delta}^{*}$. Thus, Proposition \ref{thm:minimax} implies that $\hat{\delta}^{*}$ is a minimax optimal rule and $\pi^{*}$ is least favorable. See also Figure \ref{fig:1} for a graphical 
illustration.
The full proof of Theorem \ref{thm:minimax} is left to Appendix \ref{sec:App.A}. 
\end{rem}
\begin{figure}[ht]
    \centering
    \includegraphics[scale=0.3]{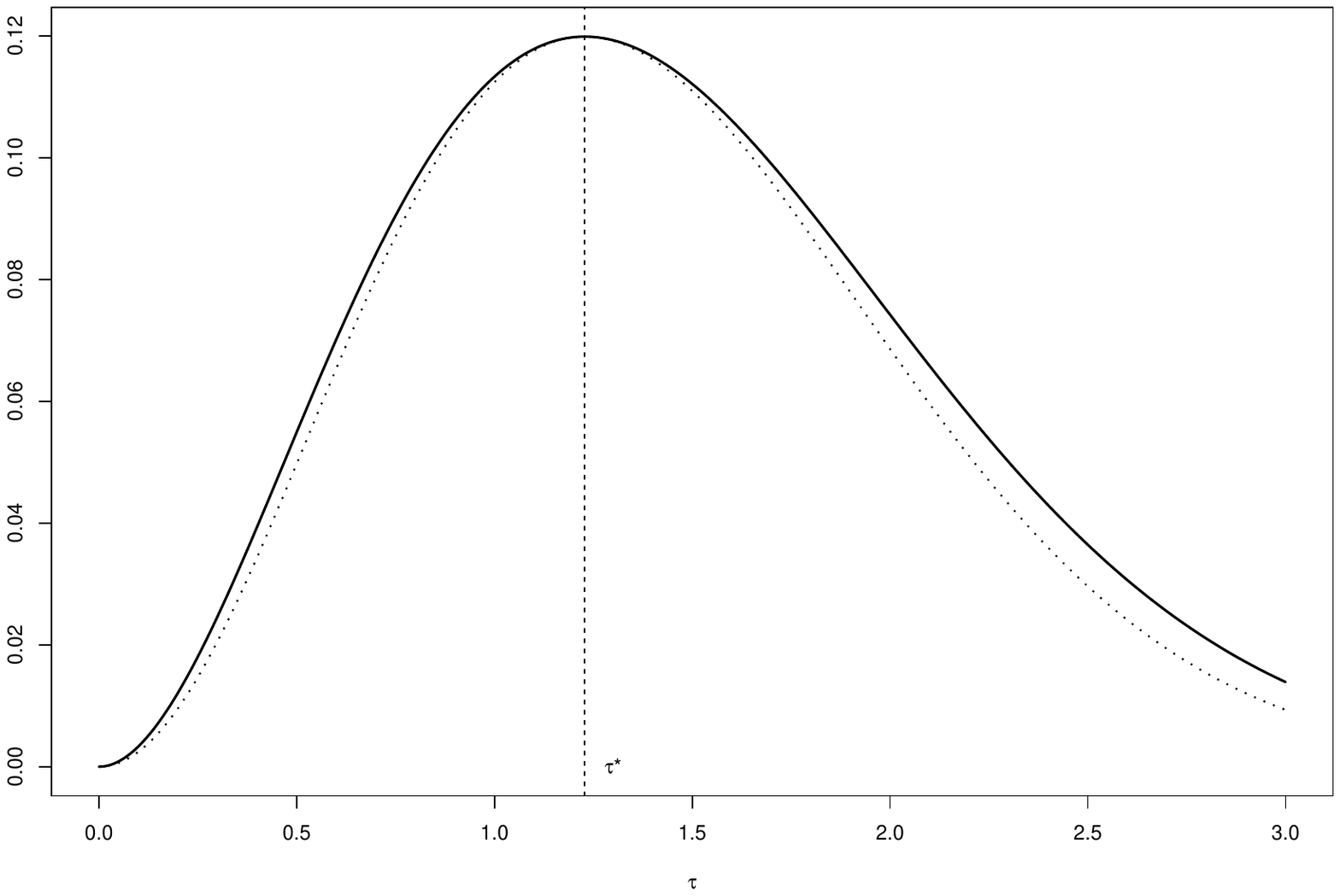}
    \caption{Illustration of Condition \ref{cond:1} for Theorem \ref{thm:minimax}. The dotted line is $r_{sq}(\hat{\delta}_{\pi^{*}_{\tau}},\pi^{*}_{\tau})$ as a function of $\tau$; the solid line is $R_{sq}\left(\hat{\delta}^{*},P_{\tau}\right)$ as a function of $\tau$.}
    \label{fig:1}
\end{figure}

\section{Further discussions}\label{sec:discuss}

\subsection{Microeconomic foundation}\label{sec:discussion.axiom}

\cite{hayashi2008regret} axiomatizes a class of regret-driven choices, including mean square regret and many other nonlinear transformations of regret. We briefly discuss how the results of \cite{hayashi2008regret} provide a microeconomic justification of our
approach. Let $\mathcal{S}:= \mathbf{Z}_n \times \Theta$ be the product space of sample $z_n \in \mathbf{Z}_n$ and the space of parameters indexing the distribution of the sample $\theta \in \Theta$.  Let $\mathcal{D}$ be the set of statistical treatment choice rules $\hat{\delta} : \mathbf{Z}_n \to [0,1]$ and $\mathcal{D}^{\ast}$ be the extended set of treatment choice rules $\delta : \mathcal{S} \to [0,1]$ which includes the infeasible oracle rule $\delta^{\ast} = 1\{ \tau \geq 0 \}$.\footnote{Therefore, the set of statistical decision rules $\mathcal{D}$ available to a decision maker is allowed to be not as rich as $\mathcal{D}^*$.} We can view $\mathcal{S}$ as the state space, $\delta$ as acts, and $\mathcal{D}$ and $\mathcal{D}^*$ as menus in decision theory. To facilitate translation of our notation to that used in \cite{hayashi2008regret}, write $W(\delta,\theta):=W(\delta)$, viewed as the utility of an act. For nonlinear transformation $g(x) = x^{\alpha}$, a nonlinear regret criterion with prior $\pi$ for $\theta$ writes as
\begin{equation*}
\int_{s \in \mathcal{S}}  \left( \max_{\delta \in \mathcal{D}^{\ast}} W(\delta(s)) - W(\hat{\delta}(s)) \right)^{\alpha} d P(z_n|\theta) d \pi(\theta).
\end{equation*}
The Bayes optimal nonlinear regret rule solves the following minimization
\begin{equation}
\inf_{\hat{\delta} \in \mathcal{D}}\int_{s \in \mathcal{S}}  \left( \max_{\delta \in \mathcal{D}^{\ast}} W(\delta(s)) - W(\hat{\delta}(s)) \right)^{\alpha} d P(z_n|\theta) d \pi(\theta).\label{eq:axiom.bayes.nonlinear}
\end{equation}
We can view the minimax optimal nonlinear regret rule as a solution to the following minimization:
\begin{equation}
\inf_{\hat{\delta} \in \mathcal{D}} \sup_{\pi \in \Pi} \int_{s \in \mathcal{S}}  \left( \max_{\delta \in \mathcal{D}^{\ast}} W(\delta(s)) - W(\hat{\delta}(s)) \right)^{\alpha} P(z_n|\theta) d \pi(\theta),\label{eq:axiom.minimax.nonlinear}
\end{equation}
where $\Pi$ denotes the set of probability distributions on $\Theta$.

Let $\left|\mathcal{S}\right|$ be the cardinality of $\mathcal{S}$, assumed to be finite in \cite{hayashi2008regret}.
\citet[Theorem 1, ][]{hayashi2008regret} builds a general axiomatic model where the choice of a decision
maker is represented by the following minimization:
\begin{equation}
\min_{\hat{\delta}\in \mathcal{D}}\varPhi\left(\max_{\delta\in \mathcal{D}}W(\delta(\cdotp))-W(\hat{\delta}(\cdotp))\right),\label{eq:hayashi.regret.general}
\end{equation}
where $\varPhi:\mathbb{R}_{+}^{\left|\mathcal{S}\right|}\rightarrow\mathbb{R}_{+}$
is a homothetic function. 
The function $\varPhi$ is an aggregator that collects the decision maker's regret in different states of the world. Note both (\ref{eq:axiom.bayes.nonlinear}) and (\ref{eq:axiom.minimax.nonlinear}) may be viewed as special cases of (\ref{eq:hayashi.regret.general}) (subject to caveats discussed below). While \cite{hayashi2008regret} also discussed (\ref{eq:axiom.bayes.nonlinear}), the minimax nonlinear regret criterion (\ref{eq:axiom.minimax.nonlinear}) has not been considered elsewhere in the literature to the best of our knowledge. The case of $\alpha>1$ is called regret aversion in \cite{hayashi2008regret}.

Axiomatic results in decision theory, like (\ref{eq:hayashi.regret.general}), focus on decision making without sample data. Our criteria (\ref{eq:axiom.bayes.nonlinear}) and (\ref{eq:axiom.minimax.nonlinear}) are tailored for decision making with sample data. We recognize two caveats when interpreting (\ref{eq:axiom.bayes.nonlinear}) and (\ref{eq:axiom.minimax.nonlinear}) as special cases of (\ref{eq:hayashi.regret.general}).
First, in (\ref{eq:axiom.bayes.nonlinear}) and (\ref{eq:axiom.minimax.nonlinear}),  the menu used to calculate regret, $\mathcal{D}^*$, includes the infeasible oracle rule and is allowed to be different from the actual menu $\mathcal{D}$ available to the decision maker.  Second, in (\ref{eq:axiom.bayes.nonlinear}) and (\ref{eq:axiom.minimax.nonlinear}) the utility function $W$ is usually state dependent  while in (\ref{eq:hayashi.regret.general}) the utility function is state independent. 
Fully reconciling the differences between decision theory and statistical decision theory is beyond the scope of this paper.  \citet[p. 2831]{manski2021econometrics} wrote: ``As in decisions without sample data, there is no clearly best way to choose
among admissible statistical decision functions (SDFs). Statistical decision theory has mainly studied the same  criteria as has decision theory without sample data.''
In view of \cite{manski2021econometrics}, our nonlinear regret criteria can find their counterparts in decision theory without sample data from \cite{hayashi2008regret}.

\subsection{Optimal rule as a summary statistic}\label{sec:discussion.interpretation}
The treatment probability of our suggested minimax optimal rule is always between zero and one. As such, our rule can be naturally viewed as a summary statistic that measures the \emph{strength of evidence} in favor of treatment versus control. Therefore, the usefulness of our approach does not hinge on the decision theoretic framework. One does not have to literally make a treatment decision  based on $\hat{\delta}^{*}$. Instead, empirical researchers may view $\hat{\delta}^{*}$ as a degree of confidence gathered from data about the performance of the treatment in terms of welfare. Given finite sample from a single phase experiment, a larger value of $\hat{\delta}^{*}$ means we are more in favor of the treatment, while a smaller value of $\hat{\delta}^{*}$ signals less evidence supporting implementing the treatment. In contrast, in the standard mean regret paradigm, optimal rules are singleton and not fractional. Hence, it is not possible for applied researchers to solicit a measure of evidence strength from an optimal decision rule. Consider a scenario where $\Bar{Y}_{1}$ is only slightly larger than zero. The empirical success rule would dictate everyone in the population to be treated, even though we may think that the evidence reflected from data in favor of the treatment is not entirely strong. 

Viewing $\hat{\delta}^{*}$ as a measure of the strength of evidence in a binary treatment setup, applied researchers may report $\hat{\delta}^{*}$ as an alternative summary statistic to the widely used P value. Despite its popularity, P value is known to be unfit for a measure of support for its hypothesis \citep{schervish1996p}. Consider the setup in Example \ref{exa:1} again. The P value for one-sided hypotheses $\mathbb{H}_{0}:\tau=0, \text{ v.s. } \mathbb{H}_{1}: \tau>0$ is $1-\Phi(\Bar{Y}_{1})$.  Since $\Phi(\Bar{Y}_{1})$ is in fact the posterior probability with respect to the flat prior, reporting the P value corresponds to reporting the posterior probability  under a flat prior. However, reporting a posterior probability under a specific prior is not necessarily associated with any optimality criterion. Different from the P value, $\hat{\delta}^{*}$ is an optimal treatment fraction under our mean square regret criterion. At the same time, $\hat{\delta}^{*}$ is also a posterior probability under a least favorable prior. Note given $\Bar{Y}_{1}>0(<0)$, $\hat{\delta}^{*}$ is quantitatively larger (smaller) than the P value. Therefore, reporting the P value would be more conservative than reporting $\hat{\delta}^{*}$ in our mean square regret framework. In Section \ref{sec:app.normal.regression}, we discuss how to calculate $\hat{\delta}^{*}$ in a normal regression model with binary treatment.


\section{Asymptotic optimality with mean square regret} \label{sec:asymptotic}

In this section we derive asymptotically optimal rules via the limit 
experiment framework \citep{le2012asymptotic}, following the
approach taken by \cite{HiranoPorter2009}. We first consider a local
parametrization of the statistical model $P$ so that, in large samples,
the treatment choice problem is equivalent to a simpler problem in a Gaussian limit experiment. Then, we examine and normalize
our nonlinear regret in the limit, and find the corresponding
optimal treatment rule. A feasible and asymptotically
optimal treatment rule also follows if there exists an efficient estimator of the parameters
in the original statistical model $P$. For a review, see \citet{HiranoPorter2020}.

\subsection{Limit experiments}

For simplicity, we focus on regular parametric models
of $P\in\mathcal{P}$ with mean square regret $R_{sq}$. Semiparametric
models and other nonlinear regret criteria can also be considered, albeit necessitating more technical analysis. Without loss of generality, consider a
case where the distribution of $Y(0)$ is known and the mean of $Y(0)$
is zero. Suppose now the distribution of $Y(1)$, denoted by $P$, is
parameterized by a finite dimensional parameter $\theta\in\Theta\subseteq\mathbb{R}^{k}$.
Hence, the population average treatment effect is 
\[
\tau(\theta)=\int zdP_{\theta}(z).
\]

Data $Z_{n}=\{Z_{i}\}_{i=1}^{n}$ is independently and identically drawn from $P_{\theta}$.
In particular, $Z_{i}\sim P_{\theta}$, where $Z_{i}\in\mathbf{Z}$
and $\mathbf{Z}$ is the support of $Z_{i}$. We now imagine a sequence
of experiments $\mathcal{E}_{n}:=\{P_{\theta}^{n},\theta\in\Theta\}$
in which the sample size $n$ grows. Let $\theta_{0}\in\Theta$ satisfy $\tau(\theta_{0})=0$. We consider a sequence of local alternative
parameters of the form $\theta_{0}+\frac{h}{\sqrt{n}}$,
$h\in\mathbb{R}^{k}$,
the most challenging case in which to
determine the optimal treatment rule, even in large samples. 

\begin{assumptionDQM}[Differentiability in Quadratic Mean] There exists a function $s:\mathbf{Z}\rightarrow\mathbb{R}^{k}$
such that 
\[
\int\left[dP_{\theta_{0}+h}^{\frac{1}{2}}(z)-dP_{\theta_{0}}^{\frac{1}{2}}(z)-\frac{1}{2}h^{\prime}s(z)dP_{\theta_{0}}^{\frac{1}{2}}(z)\right]^{2}=o(\left\Vert h\right\Vert ^{2}),\text{ as }h\rightarrow0,
\]
and $I_{0}:=\mathbb{E}_{\theta_{0}}\left[ss^{\prime}\right]$ is nonsingular.
\end{assumptionDQM}

Assumption DQM is a standard assumption in the limit experiment framework \citep[e.g.,][]{van2000asymptotic}. The function $s$ can usually be interpreted as the derivative of
the loglikelihood function so that $I_{0}$
is the Fisher information under $P_{\theta_{0}}$.

\begin{assumptionC}[Convergence]
A sequence
of treatment rules $\hat{\delta}_{n}$ in the experiments $\mathcal{E}_{n}$ is such that $\beta_{n}(h,1):=\mathbb{E}_{P_{\theta_{0}+\frac{h}{\sqrt{n}}}^{n}}[\hat{\delta}_{n}]\rightarrow\beta(h,1)$
and $\beta_{n}(h,2):=\mathbb{E}_{P_{\theta_{0}+\frac{h}{\sqrt{n}}}^{n}}[(\hat{\delta}_{n})^{2}]\rightarrow\beta(h,2)$
for every $h$ as $n\rightarrow\infty$.
\end{assumptionC}

Compared to mean regret criterion, our mean square regret additionally depends on the second moment of decision rules. Thus, Assumption C assumes convergence of both first and second moments of decision rules, differing from \cite{HiranoPorter2009}, who only look at convergence of the first moment of decision rules. Under Assumptions
DQM and C, we first establish the following result that allows us to simplify
the original treatment problem to a Gaussian experiment
in large samples.
\begin{prop}[\citet{van2000asymptotic}]\label{prop:limit} Suppose $\mathcal{E}_{n}$
satisfy Assumption DQM and a sequence
of treatment rules $\hat{\delta}_{n}$ in $\mathcal{E}_{n}$ satisfy Assumption C. Then there exists a function
$\hat{\delta}:\mathbb{R}^{k}\rightarrow[0,1]$ such that for every
$h\in\mathbb{R}^{k}$,
\[
\beta(h,1)=\int\hat{\delta}(\varDelta)dN(\varDelta|h,I_{0}^{-1}),\text{ }\beta(h,2)=\int\left(\hat{\delta}(\varDelta)\right)^{2}dN(\varDelta|h,I_{0}^{-1}),
\]
where $N(\varDelta|h,I_{0}^{-1})$ is a multivariate normal distribution
with mean $h$ and variance $I_{0}^{-1}$.
\end{prop}

Proposition \ref{prop:limit} is a special case of  \citet[Theorem 15.1 and Theorem 7.10]{van2000asymptotic} applied to the mean square regret setup, following
\citet[Proposition
3.1]{HiranoPorter2009}. To use Proposition \ref{prop:limit},
note for any treatment rule $\hat{\delta}_{n}$ in the experiments
$\mathcal{E}_{n}$, the mean square regret is 
\[
\mathbb{E}_{P_{\theta_{0}+\frac{h}{\sqrt{n}}}^{n}}\left[\tau\left(\theta_{0}+\frac{h}{\sqrt{n}}\right)^{2}\left(1\left\{ \tau\left(\theta_{0}+\frac{h}{\sqrt{n}}\right)\geq0\right\} -\hat{\delta}_{n}\right)^{2}\right],
\]
which depends on $\hat{\delta}_{n}$ only through $\mathbb{E}_{P_{\theta_{0}+\frac{h}{\sqrt{n}}}^{n}}[\hat{\delta}_{n}]$
and $\mathbb{E}_{P_{\theta_{0}+\frac{h}{\sqrt{n}}}^{n}}[\hat{\delta}_{n}^{2}]$,
to which we can apply Proposition \ref{prop:limit}. Thus, in terms of the mean square regret, any converging sequence of treatment rules is matched by
some treatment rule in a simpler Gaussian experiment with unknown mean
$h$ and known variance $I_{0}^{-1}$. 

Let $\dot{\tau}$ be the partial derivative of $\tau(\theta)$ at
$\theta_{0}$. Since $\tau\left(\theta_{0}\right)=0$, it follows that
$\sqrt{n}\tau\left(\theta_{0}+\frac{h}{\sqrt{n}}\right)\rightarrow \dot{\tau}^{\prime}h$ as $n\rightarrow\infty$.
Thus, for any rule $\delta$,
\begin{align*}
\sqrt{n}Reg\left(\delta,\left(\theta_{0}+\frac{h}{\sqrt{n}}\right)\right)&\rightarrow\dot{\tau}^{\prime}h\left[1\left\{ \dot{\tau}^{\prime}h\geq0\right\} -\delta\right] := Reg_{\infty}(\delta,h),
\end{align*}
and $n\left[Reg\left(\delta,\left(\theta_{0}+\frac{h}{\sqrt{n}}\right)\right)\right]^{2}\rightarrow\left(Reg_{\infty}(\delta,h)\right)^{2}$
as $n\rightarrow\infty$. Hence, normalizing by $n$, for any converging rule $\hat{\delta}_{n}$ in the sense of Proposition \ref{prop:limit}, we
define the corresponding limit mean square regret as
\begin{align}
R_{sq}^{\infty}(\hat{\delta},h) & :=\int\left(Reg_{\infty}(\hat{\delta}(\varDelta),h)\right)^{2}dN(\varDelta|h,I_{0}^{-1})\label{eq:limit.risk}\\
 & =\mathbb{E}_{\varDelta\sim N(h,I_{0}^{-1})}\left[Reg_{\infty}(\hat{\delta}(\varDelta),h)\right]^{2}.\nonumber 
\end{align}

With (\ref{eq:limit.risk}) as the mean square regret in the limit
experiment, we can apply our finite sample results in Section \ref{sec:finite}
and derive a feasible and asymptotically optimal treatment rule via an efficient estimator
of the parameters.

\subsection{Feasible and asymptotically optimal rules}

We first present results in terms of minimax optimality. Denote
$\overset{h}{\rightsquigarrow}$ as convergence in distribution under
the sequence of probability measures $P_{\theta_{0}+\frac{h}{\sqrt{n}}}^{n}$.
Define $\sigma_{\tau}:=\sqrt{\dot{\tau}^{\prime}I_{0}^{-1}\dot{\tau}}$ to be
the standard deviation of $\dot{\tau}^{\prime}\varDelta$, where $\varDelta\sim N(h,I_{0}^{-1})$.
\begin{thm} \label{thm:minimax.feasible}
Suppose Proposition \ref{prop:limit} holds, $\tau(\theta_{0})=0$,
and $\tau(\theta)$ is differentiable at $\theta_{0}$. 
\begin{itemize}
\item[(i)] The minimax optimal rule in the limit experiment is
\[
\hat{\delta}^{*}(\varDelta)=\frac{\exp\left(\frac{2\tau^{*}}{\sigma_{\tau}}\dot{\tau}^{\prime}\varDelta\right)}{\exp\left(\frac{2\tau^{*}}{\sigma_{\tau}}\dot{\tau}^{\prime}\varDelta\right)+1},
\]
where $\tau^{*}\approx1.23$, and which solves (\ref{eq:minimax.1}).

\item[(ii)] If, in addition, there exists a best regular estimator $\hat{\theta}$
such that

\begin{equation}\label{eq:regular}
\sqrt{n}\left(\hat{\theta}-\theta_{0}-\frac{h}{\sqrt{n}}\right)\overset{h}{\rightsquigarrow}N(0,I_{0}^{-1}),\text{ for all }h\in\mathbb{R}^{k},
\end{equation}
and there exists some estimator $\hat{\sigma}_{\tau}\overset{p}{\rightarrow}\sigma_{\tau}$
under $\theta_{0}$, the feasible treatment rule 
\[
\hat{\delta}_{F}^{*}(Z_{n})=\frac{\exp\left(\frac{2\tau^{*}}{\hat{\sigma}_{\tau}}\sqrt{n}\tau(\hat{\theta})\right)}{\exp\left(\frac{2\tau^{*}}{\hat{\sigma}_{\tau}}\sqrt{n}\tau(\hat{\theta})\right)+1}
\]
is locally asymptotically minimax optimal in terms of mean square regret:
\[
\sup_{J}\liminf_{n\rightarrow\infty}\sup_{h\in J}nR_{sq}(\hat{\delta}_{F}^{*},\theta_{0}+\frac{h}{\sqrt{n}})=\inf_{\hat{\delta}\in\mathcal{D}}\sup_{J}\liminf_{n\rightarrow\infty}\sup_{h\in J}nR_{sq}(\hat{\delta},\theta_{0}+\frac{h}{\sqrt{n}}),
\]
where $J$ is a finite subset of $\mathbb{R}^{k}$ and $\mathcal{D}$ is the set of all decision rules that satisfy Assumption C (slightly abusing notation).
\end{itemize}
\end{thm}

Theorem \ref{thm:minimax.feasible} extends our finite sample results to a large sample setting. Given a regular parametric model,  the maximum likelihood estimator (MLE) usually satisfies (\ref{eq:regular}). Thus, Theorem \ref{thm:minimax.feasible} suggests a simple way to construct an asymptotically minimax optimal rule in terms of mean square regret: estimate the parameters of $P_{\theta}$ via MLE, calculate a $t$-statistic for the mean, and then carry out a simple logit transformation for the $t$-statistic. This rule is always fractional and very easy to implement for practitioners. We expect that our result can also be extended to regular semiparametric models.

Next, we derive a feasible rule that is locally asymptotically Bayes optimal. Let $\pi(\theta)$ be a positive and continuous prior density on $\Theta$ (slightly abusing notation). For
a treatment rule $\hat{\delta}_{n}$ that satisfies Assumption C, the normalized
Bayes mean square regret is
\begin{align*}
nr_{sq}(\hat{\delta}_{n},\pi) & =\int nR_{sq}(\hat{\delta}_{n},\theta_{0}+\frac{h}{\sqrt{n}})\pi(\theta_{0}+\frac{h}{\sqrt{n}})dh.
\end{align*}
We define the Bayes mean square
regret in the limit experiment when $n\rightarrow\infty$ as
\[
r_{sq}^{\infty}(\hat{\delta}):=\pi(\theta_{0})\int R_{sq}^{\infty}(\hat{\delta},h)dh.
\]
That is, as the Bayes mean square regret with respect to an uninformative prior. Then we can apply Theorem \ref{thm:bayes.finite} to derive the Bayes optimal rule for the limit experiment. Given an MLE estimate of the parameters in $P_{\theta}$, Theorem \ref{thm:bayes.feasible} further implies that a feasible and
asymptotically optimal Bayes rule also follows with a simple transformation of the $t$-statistic for the mean.
\begin{thm} \label{thm:bayes.feasible}
Suppose Proposition \ref{prop:limit} holds, $\tau(\theta_{0})=0$
and $\tau(\theta)$ is differentiable at $\theta_{0}$. Let $\pi(\theta)$
be the density of a prior distribution on $\Theta$ that is continuous
and positive at $\theta_{0}$. 
\begin{itemize}
\item[(i)] The Bayes optimal rule in terms of mean square regret in the
limit experiment is
\begin{align}
\hat{\delta}_{B}\left(\varDelta\right) & =\Phi\left(\frac{\dot{\tau}^{\prime}\varDelta}{\sigma_{\tau}}\right)\left(1+ \frac{\dot{\tau}^{\prime}\varDelta}{\sigma_{\tau}}\Psi\left(\frac{\dot{\tau}^{\prime}\varDelta}{\sigma_{\tau}}\right)\right),\label{eq:bayes.limit}
\end{align}
where $\Psi(\cdot)$ is defined in Proposition \ref{prop:bayes.flat}.
That is, $r_{sq}^{\infty}(\hat{\delta}_{B})=\inf_{\delta\in\mathcal{D}_{\infty}}r_{sq}^{\infty}(\delta)$,
where $\mathcal{D}_{\infty}$ is the set of all treatment rules in the
$N(h,I_{0}^{-1})$ limit experiment.

\item[(ii)] If, in addition, there exists a best regular estimator $\hat{\theta}$
such that
\[
\sqrt{n}\left(\hat{\theta}-\theta_{0}-\frac{h}{\sqrt{n}}\right)\overset{h}{\rightsquigarrow}N(0,I_{0}^{-1}),\text{ for all }h\in\mathbb{R}^{k},
\]
and there exists some estimator $\hat{\sigma}_{\tau}\overset{p}{\rightarrow}\sigma_{\tau}$
under $\theta_{0}$, the feasible treatment rule 
\[
\hat{\delta}_{B,F}(Z_{n})=\Phi\left(\frac{\sqrt{n}\tau(\hat{\theta})}{\hat{\sigma}_{\tau}}\right)\left[1+\frac{\sqrt{n}\tau(\hat{\theta})}{\hat{\sigma}_{\tau}}\Psi\left(\frac{\sqrt{n}\tau(\hat{\theta})}{\hat{\sigma}_{\tau}}\right)\right]
\]
is locally asymptotically Bayes optimal, i.e.,
\[
\lim_{n\rightarrow\infty}nr_{sq}(\hat{\delta}_{B,F},\pi)=\inf_{\hat{\delta}\in\mathcal{D}}\liminf_{n\rightarrow\infty}nr_{sq}(\hat{\delta},\pi).
\]
\end{itemize}
\end{thm}

In the limit, the Bayes optimal rule is a \emph{tilted} posterior probability matching rule with respect to the uninformative prior. Compared to the posterior probability matching rule, the Bayes optimal rule  assigns treatment with a probability closer to zero or one. Compared to the limit minimax optimal rule, the Bayes optimal rule also assigns treatment with a probability close to zero or one. This contrasts with the case of \emph{linear} regret risk, where it is known that the Bayes optimal and minimax optimal rules are the same empirical success rule. See Figure \ref{fig:2} and Table \ref{tab:1} for various rules in a Gaussian limit experiment with unit variance. 
It can be seen that all three fractional rules approach one as $\bar{Y}_{1}$ gets large. 
For sufficiently large positive values of $\bar{Y}_{1}$  (e.g., 2.33), the Bayes and minimax optimal rules are to effectively treat everyone. Even with a modest value of $\bar{Y}_{1} = 0.84$, the Bayes optimal rule recommends a probability of treatment of 0.94, which is quite high when compared with the corresponding probability of $0.8$ recommended by the posterior probability matching rule.
Figures \ref{fig:3}, \ref{fig:4} and  \ref{fig:5} present the mean square regret, mean regret and standard deviation of regret of the optimal rules in the same Gaussian limit experiment with unit variance. We make several observations: firstly, although they admit different forms, our Bayes optimal and minimax optimal rules in the limit experiment exhibit a similar performance in terms of the mean square regret (Figure \ref{fig:3}); secondly, the ES rule is minimax optimal in terms of mean regret (Figure \ref{fig:4}), but its excessive variance (Figure \ref{fig:5}) in those states where mean regret is high implies that it is not optimal in terms of mean square regret.

\begin{table}[htbp]
    \centering
    \begin{tabular}{ccccc}
\toprule 
\multirow{2}{*}{$\bar{Y}_{1}$} & Minimax  & Bayes & Posterior probability  & \multirow{2}{*}{ES rule}\tabularnewline
 & optimal rule &  optimal rule & matching rule (flat prior) & \tabularnewline
\midrule
0 & 0.5 & 0.5 & 0.5 & $[0,1]$\tabularnewline
0.2533 & 0.6507 & 0.6920 & 0.6 & 1\tabularnewline
0.5244 & 0.7838 & 0.8430 & 0.7 & 1\tabularnewline
0.8416 & 0.8877 & 0.9379 & 0.8 & 1\tabularnewline
1.2816 & 0.9588 & 0.9851 & 0.9 & 1\tabularnewline
1.6449 & 0.9827 & 0.9958 & 0.95 & 1\tabularnewline
2.3263 & 0.9967 & 0.9997 & 0.99 & 1\tabularnewline
\bottomrule
\end{tabular}
    \caption{Treatment assignment probabilities in the Gaussian limit experiment with unit variance}
    \label{tab:1}
\end{table}

\begin{figure}[htbp]
    \centering
    \includegraphics[scale=0.4]{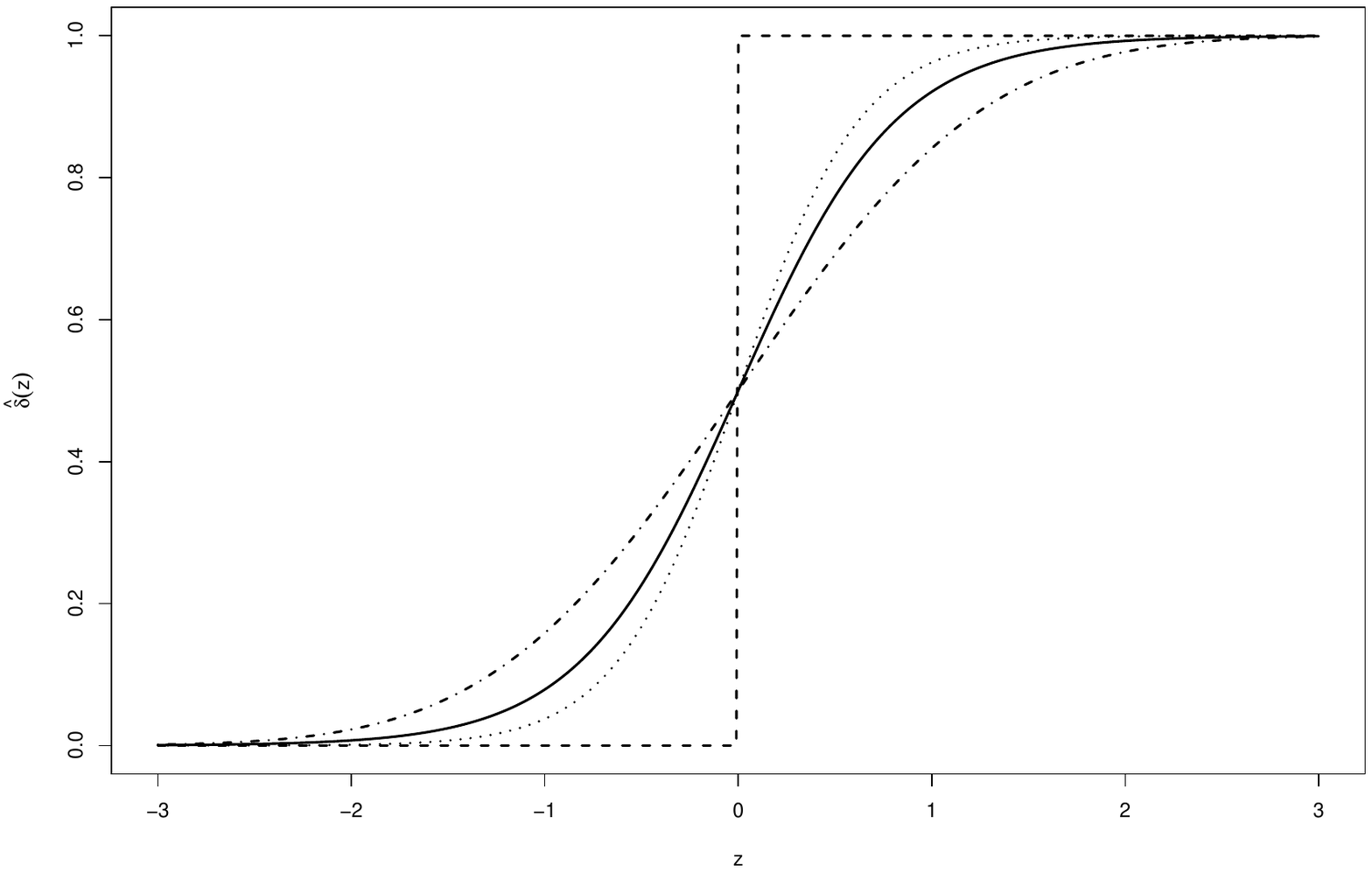}
    \caption{Optimal rules in the Gaussian limit experiment with unit variance. Solid line: minimax optimal rule for mean square regret; Dotted line: Bayes optimal rule for mean square regret; Dot-dashed line: posterior probability matching rule with respect to a flat prior; Dashed line: Empirical success (ES) rule. }
    \label{fig:2}
        \includegraphics[scale=0.4]{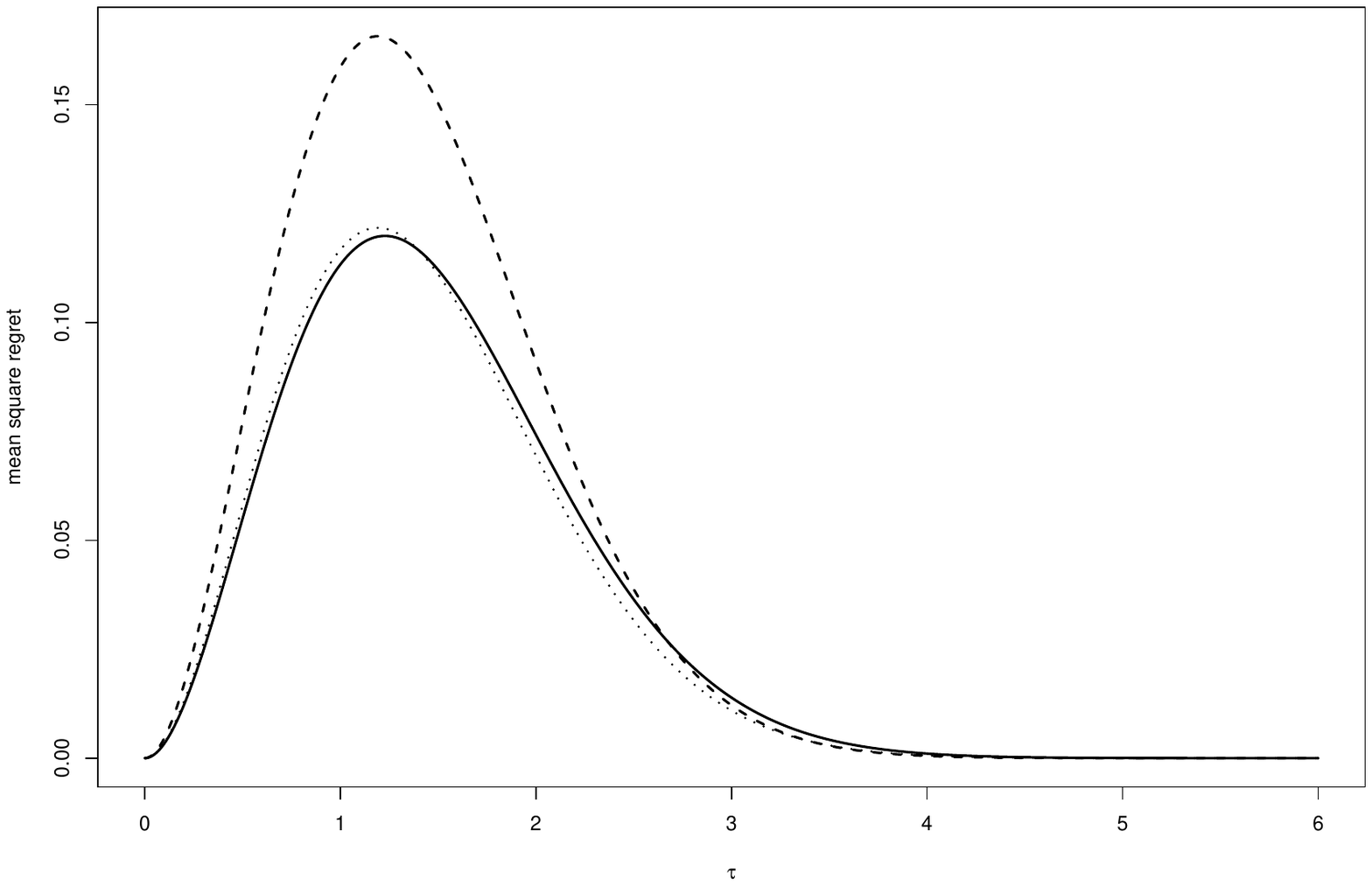}
    \caption{Mean square regret in the Gaussian limit experiment with unit variance. Solid line: minimax optimal rule; Dotted line: Bayes optimal rule with respect to a flat prior; Dashed line: ES rule. }
    \label{fig:3}
\end{figure}

\begin{figure}[htbp]
    \centering
    \includegraphics[scale=0.4]{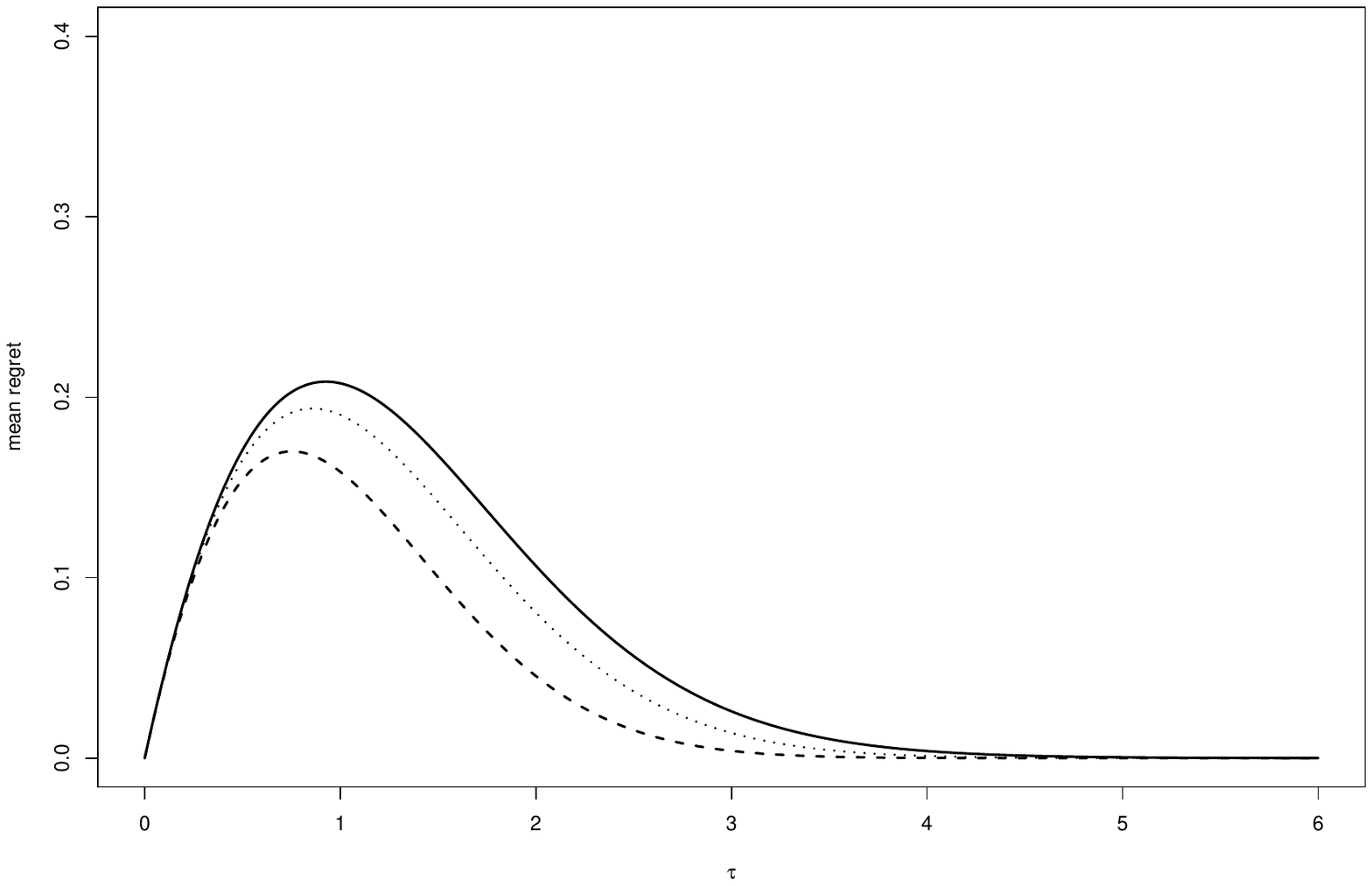}
    \caption{Mean regret in the Gaussian limit experiment with unit variance. Solid line: minimax optimal rule; Dotted line: Bayes optimal rule with respect to a flat prior; Dashed line: ES rule. }
    \label{fig:4}
    \includegraphics[scale=0.4]{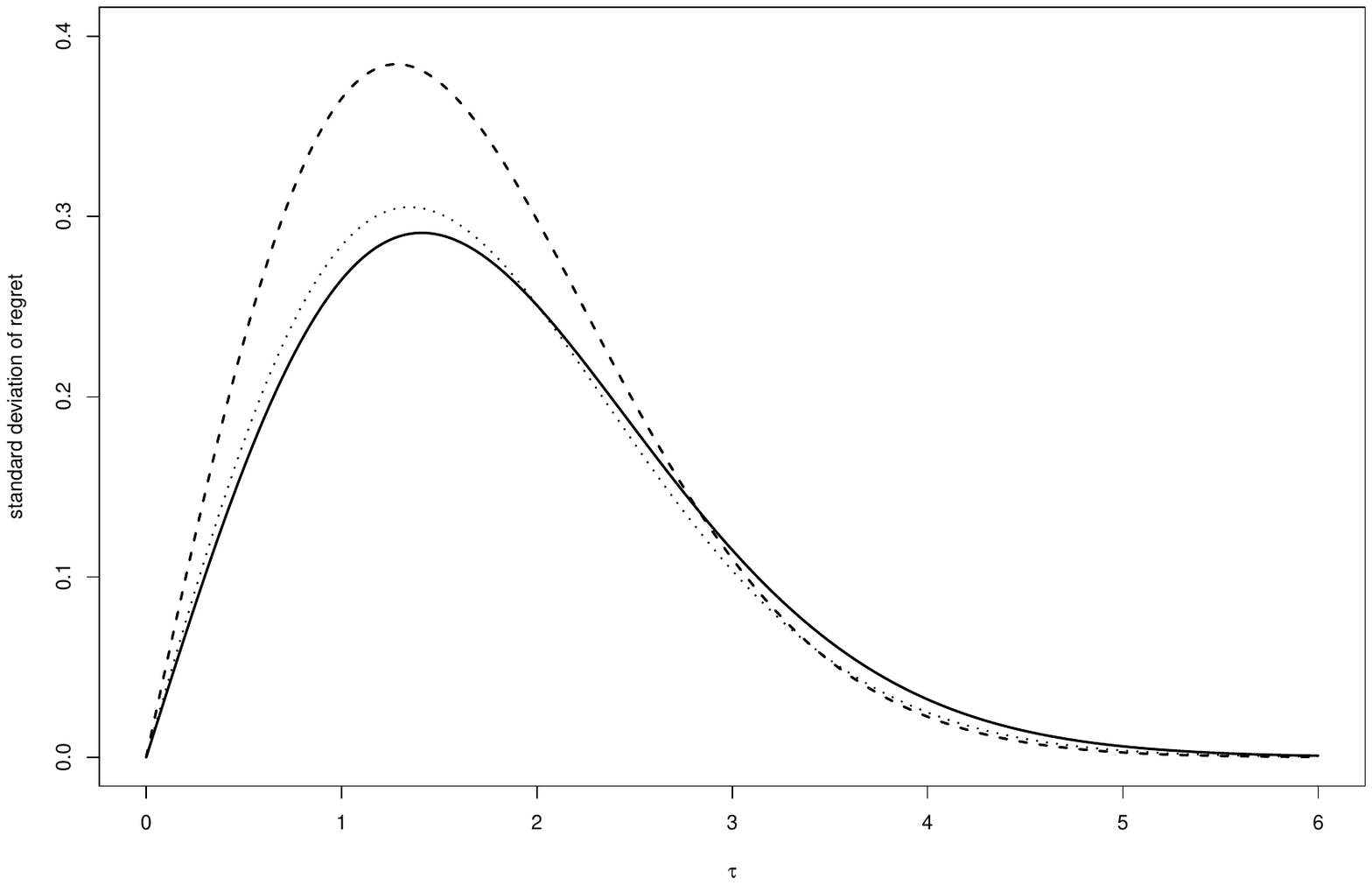}
    \caption{Standard deviation of regret in the Gaussian limit experiment with unit variance. Solid line: minimax optimal rule; Dotted line: Bayes optimal rule with respect to a flat prior; Dashed line: ES rule. }
    \label{fig:5}
\end{figure}

\section{Applications}\label{sec:applications}

\subsection{Treatment choice in a normal regression model}\label{sec:app.normal.regression}

Consider the following normal regression model frequently used by applied researchers:
\begin{equation}
Y=\tau D+\beta^{\prime}X+e,\quad e\sim N(0,\sigma^{2}),\label{eq:normal.regression}
\end{equation}
where $Y$ is the outcome variable, $D$ is the binary treatment and
$X$ is a vector of covariates (including the intercept). Suppose the treatment effect is homogeneous.
Then, the parameter $\tau\in\mathbb{R}$ is the population average
treatment effect. Let $\theta=(\tau,\beta^{\prime})^{\prime}$ and
$Z=(D,X^{\prime})^{\prime}$. (\ref{eq:normal.regression}) implies
that the conditional density of $Y$ given $Z$ follows the parametric form
\[
f(y|z)=\frac{1}{\sqrt{2\pi\sigma^{2}}}\exp\left(-\frac{1}{2\sigma^{2}}(y-\theta^{\prime}z)^{2}\right).
\]

For now, assume the variance term $\sigma^{2}$ is known to focus on the finite-sample analysis. Given
a random sample $\left\{ (Y_{i},Z_{i}^{\prime})^{\prime}\right\} _{i=1}^{n}$, the MLE estimator for $\theta$ is 
\begin{equation}
\hat{\theta}=\left(\sum_{i=1}^{n}Z_{i}Z_{i}^{\prime}\right)^{-1}\left(\sum_{i=1}^{n}Z_{i}Y_{i}\right),\label{eq.normal.mle.1}
\end{equation}
the usual OLS estimator. Let $\hat{\tau}:=\hat{\theta}_{1}$, the
first entry of $\hat{\theta}$.
It follows by standard algebra that 
\[
\hat{\tau}\mid Z_{1},\ldots,Z_{n}\sim N\left(\tau,\sigma^{2}\left[\left(\sum_{i=1}^{n}Z_{i}Z_{i}^{\prime}\right)^{-1}\right]_{11}\right),
\]
where $\left[M\right]_{ij}$ denotes the $(i,j)$th entry of matrix
$M$. By Theorem \ref{thm:minimax}, the finite sample minimax optimal rule is 
\[
\hat{\delta}^{*}=\frac{\exp\left(\frac{2\tau^{*}\hat{\tau}}{\sqrt{\sigma^{2}\left[\left(\sum_{i=1}^{n}Z_{i}Z_{i}^{\prime}\right)^{-1}\right]_{11}}}\right)}{\exp\left(\frac{2\tau^{*}\hat{\tau}}{\sqrt{\sigma^{2}\left[\left(\sum_{i=1}^{n}Z_{i}Z_{i}^{\prime}\right)^{-1}\right]_{11}}}\right)+1},
\]
where $\tau^{*}$ is defined in Theorem \ref{thm:minimax}. Even if $\sigma^{2}$
is unknown, the MLE estimator for $\left(\theta^{\prime},\sigma^{2}\right)^{\prime}$
is $\left(\hat{\theta}^{\prime},\hat{\sigma}^{2}\right)^{\prime}$,
where $\hat{\theta}$ is defined in (\ref{eq.normal.mle.1}), and
\[
\hat{\sigma}^{2}=\frac{1}{n}\sum_{i=1}^{n}\left(Y_{i}-Z_{i}^{\prime}\hat{\theta}\right)^{2}.
\]
Applying Theorem \ref{thm:minimax.feasible}, we may find a feasible asymptotically minimax
optimal treatment rule as
\[
\hat{\delta}_{F}^{*}=\frac{\exp\left(\frac{2\tau^{*}\hat{\tau}}{\sqrt{\hat{\sigma}^{2}\left[\left(\sum_{i=1}^{n}Z_{i}Z_{i}^{\prime}\right)^{-1}\right]_{11}}}\right)}{\exp\left(\frac{2\tau^{*}\hat{\tau}}{\sqrt{\hat{\sigma}^{2}\left[\left(\sum_{i=1}^{n}Z_{i}Z_{i}^{\prime}\right)^{-1}\right]_{11}}}\right)+1}.
\]

Practitioners may report $\hat{\delta}_{F}^{*}$ as an alternative to the P value  associated with $\tau$. Also see Section \ref{sec:discussion.interpretation} for more discussions on this issue.

\subsection{Sample size calculations} \label{sec:samplesize}

In practice, the planner often has a preference for singleton rules
like the empirical success (ES) rule or the hypothesis testing (HT) rule,
and calculates what is a sufficient sample size based on these singleton
rules. In this section we discuss the implications for the efficiency
loss in terms of mean square regret if singleton rules were
implemented instead of our proposed minimax optimal rules. Compared to our minimax optimal rule, these singleton rules often require significantly more data and thus are much less efficient. 
A similar discussion can be had for the Bayes optimal rule, but we omit this for brevity.

Consider the Gaussian experiment in Example \ref{exa:1}, but suppose now  $\bar{Y}_{1}\sim N(\tau,\frac{\sigma^{2}}{n})$ is the sample average calculated from experimental data
with a sample size of $n$ and known variance $\sigma^{2}>0$. In this case the minimax
optimal rule in terms of mean square regret is 
\[
\hat{\delta}^{*}(\bar{Y}_{1})=\frac{\exp(2\tau^{*}\frac{\sqrt{n}}{\sigma}\bar{Y}_{1})}{\exp(2\tau^{*}\frac{\sqrt{n}}{\sigma}\bar{Y}_{1})+1},
\]
where $\tau^{*}$ solves (\ref{eq:minimax.1}). Given each
$\varepsilon>0$, we can select $n$ such that 
\[
\sqrt{\sup_{\tau\in[0,\infty)}  R_{sq}(\hat{\delta}^{*},P_{\tau})}\leq\varepsilon,
\]
i.e., the square root of the worst case mean square regret does not exceed
$\varepsilon$. The worst case mean square regret
can be calculated as 
\begin{align*}
\sup_{\tau\in[0,\infty)}R_{sq}(\hat{\delta}^{*},P_{\tau})=\left(\frac{\sigma^{2}}{n}\right)R_{sq}^{*}(1),
\end{align*}
where $R_{sq}^{*}(1)\approx0.1199$ is the worst case mean square regret
of the minimax optimal rule in Example \ref{exa:1}. Thus, the worst case mean square regret shrinks
to zero at a rate of $\frac{1}{n}$. In practice, we can choose $\varepsilon$
to be proportional to $\sigma$, e.g., $0.01\sigma$, so that the square
root of the worst case mean square regret does not exceed 1\% of the standard
deviation.

\subsection*{Comparison with the ES rule}

\cite{manski2016sufficient} choose a sufficient sample size for the ES rule
via the $\varepsilon-$optimal approach: a policy $\hat{\delta}$
is $\varepsilon-$optimal if, for all states of the world,
\[
W(\delta^{*})-\mathbb{E}_{P^{n}}[W(\hat{\delta})]\leq\varepsilon,
\]
where $\delta^{*}$ is the infeasible optimal treatment rule or, equivalently,
\begin{equation}
\mathbb{E}_{P^{n}}[Reg(\hat{\delta})]\leq\varepsilon,\label{eq:ES-ss}
\end{equation}
for all states of the world. Given our Gaussian experiment $\bar{Y}_{1}\sim N(\tau,\frac{\sigma^{2}}{n})$, the worst case
mean regret of the ES rule $\widehat{\delta}_{ES}=\mathbf{1}\{\bar{Y}_{1}\geq0\}$
can be calculated exactly as 
\[
\sup_{\tau\in[0,\infty)}\tau\left(1-\Phi\left(\frac{\sqrt{n}\tau}{\sigma}\right)\right)=\frac{\sigma}{\sqrt{n}}\sup_{\tau\in[0,\infty)}\tau\left(1-\Phi\left(\tau\right)\right)=0.1700\frac{\sigma}{\sqrt{n}}.
\]
If the planner has a preference for the ES rule and decides to choose the
sample size so that (\ref{eq:ES-ss}) holds with some $\varepsilon>0$,
then the sample size should be at least
\[
n_{ES}=0.0289\frac{\sigma^{2}}{\varepsilon^{2}}.
\]
The worst case mean square regret of the ES rule, however, is
\begin{align*}
\sup_{\tau\in[0,\infty)}R_{sq}(\widehat{\delta}_{ES},P_{\tau}) & =\frac{\sigma^{2}}{n}R_{sq}^{ES}(1),
\end{align*}
where $R_{sq}^{ES}(1)=\sup_{\tau\in[0,\infty)}\tau^{2}\mathbb{E}_{\bar{Y}_{1}\sim N(\tau,1)}\left[\left(1-\mathbf{1}\{\bar{Y}_{1}\geq0\}\right)^{2}\right]\approx0.1657$.
Hence, at $n_{ES}$, the worst case mean square regret of $\widehat{\delta}_{ES}$
is $\frac{\sigma^{2}}{n_{ES}}0.1657=5.7355\frac{\varepsilon^{2}}{\sigma^{2}}$.
If, instead, the planner uses our minimax optimal rule, she only needs
a sample size of $n^{*}=0.0209\frac{\sigma^{2}}{\varepsilon^{2}}$
for the worst case mean square regret not to exceed $5.7355\frac{\varepsilon^{2}}{\sigma^{2}}$.
Thus, to guarantee the same worst case mean square regret, the ES rule requires
nearly 40\% more observations than our minimax optimal rule.

\subsection*{Comparison with the HT rule}

Practitioners who prefer the HT rule often select sample size by balancing
Type \mbox{I} and \mbox{II} errors. In the Gaussian experiment $\bar{Y}_{1}\sim N(\tau,\frac{\sigma^{2}}{n})$,
if the planner uses a size $\alpha$ ($<0.5$) HT rule
\[
\hat{\delta}_{HT}=\mathbf{1}\left\{ \frac{\sqrt{n}\bar{Y}_{1}}{\sigma}\geq z_{(1-\alpha)}\right\} ,
\]
where $z_{(1-\alpha)}$ is the $(1-\alpha)$ quantile of a standard
normal, then it is common for her to select  sample size so that the power of the
test is at least $\beta$ ($>0.5$), i.e., under the alternative $\tau>0$, the
probability of rejection is
\[
\text{Pr}\left\{ \frac{\bar{Y}_{1}-\tau}{\frac{\sigma}{\sqrt{n}}}>z_{(1-\alpha)}-\frac{\tau}{\frac{\sigma}{\sqrt{n}}}\right\} =\beta.
\]
Then the sample size should be at least
\[
n_{HT}=\frac{\sigma^{2}}{\tau^{2}}(z_{(1-\alpha)}-z_{(1-\beta)})^{2}.
\]
At this $n_{HT}$, we can also calculate the worst case mean square regret of the HT rule,
which is approximately $\frac{\tau^{2}}{(z_{(1-\alpha)}-z_{(1-\beta)})^{2}}1.4458.$
However, at this $n_{HT}$, the worst case mean square regret of our minimax
rule is only $0.1199\frac{\sigma^{2}}{n_{HT}}=0.1199\frac{\tau^{2}}{(z_{(1-\alpha)}-z_{(1-\beta)})^{2}}$.
That is to say, with the same sample size $n_{HT}$, our minimax optimal rule
guarantees that the worst case mean square regret is only around 8.3\%
of the corresponding value for the HT rule. Equivalently, to guarantee the same worst case mean
square regret, the HT rule requires around 11 times more observations than our minimax optimal rule.

\section{Conclusions}\label{sec:conclude}
Our paper proposes a novel approach to measure the performance of statistical decision rules by considering a nonlinear transformation of regret. Such a shift of criterion can incorporate other features of the regret distribution (e.g., second- or higher-order moments and tail probabilities) into the decision-making process, and yields optimal rules that are drastically different from the existing literature. For a large class of nonlinear transformations, optimal rules are fractional, allocating only a proportion of the population to the treatment. For the mean square regret criterion, we also derive Bayes optimal and minimax optimal rules both for finite Gaussian samples and in asymptotic limit experiments. These rules have a simple and insightful form, and can be calculated easily by practitioners. As an extension,  \cite*{kitagawa2023treatment} apply our mean square regret criterion to study optimal treatment choice problems when the welfare is only partially identified, and find that the fractional nature  and the fundamental form of the minimax optimal rules found in our paper is preserved.

Our approach suggests that decision makers may display \emph{regret aversion}, a notion related to but different from ambiguity aversion \citep{klibanoff2005smooth,denti2022model}. In particular, our nonlinear regret criteria can find their counterparts in decision theory from the work of \cite{hayashi2008regret} and thus are justified in terms of its microeconomic foundation. 

Since our rules are always fractional, they naturally provide a degree of confidence in the performance of the treatment versus control. In that sense, our rule is  useful for practitioners even outside the treatment choice paradigm. Implementing our rules also has the additional benefit of getting more data from randomized experiments that can be helpful for the inference of treatment effect, which would not be possible if singleton rules were implemented. 



\appendix

\section{\label{sec:App.A} Proofs of main results}

\subsection*{Proof of Theorem \ref{thm:incomplete}}
By Lemma \ref{lem:admit.1}, there exists a rule $\tilde{\delta}\in\mathcal{D}$ that is not dominated in $\{P_1,P_2\}$, and is such that 
\begin{equation}
\begin{aligned}
\text{for each } \hat{\delta}_{-\infty}\in\mathcal{D}_{S}^2,\text { it holds }R_{g}(\tilde{\delta},P)\neq R_{g}(\hat{\delta}_{-\infty},P)\text{ for some }P\in\{P_{1},P_{2}\},\\
\text{and for each } \hat{\delta}_{\infty}\in\mathcal{D}_{S}^3, \text{ it holds }R_{g}(\tilde{\delta},P)\neq R_{g}(\hat{\delta}_{\infty},P)\text{ for some }P\in\{P_{1},P_{2}\}.
\end{aligned}
\label{eq:different.risk.profile}
\end{equation}
Suppose, on the contrary, that $\mathcal{D}_{S}$ is essentially complete
in $\mathcal{P}$. Then, there must exist some $\hat{\delta}_{S}\in\mathcal{D}_{S}$
such that 
\[
R_{g}(\hat{\delta}_{S},P)\leq R_{g}(\tilde{\delta},P),\text{ for all }P\in\mathcal{P},
\]
and in particular, 
\begin{equation}
R_{g}(\hat{\delta}_{S},P)\leq R_{g}(\tilde{\delta},P),\text{ for each }P\in\left\{ P_{1},P_{2}\right\} ,\label{eq:incomplete.4}
\end{equation}
as $\left\{ P_{1},P_{2}\right\} \subset \mathcal{P}$. Below, we show that no
rule in $\mathcal{D}_{S}$ can satisfy (\ref{eq:incomplete.4}), forming a contradiction and completing the proof. 

First of all, note both $\mathcal{D}_S^2$ and $\mathcal{D}_S^3$ are non-empty, as $\mathcal{D}_S^2$ must contain the rule that treats everyone in the population irrespective of data, and  $\mathcal{D}_S^3$  must contain the rule that no one in the population irrespective of data.

Now, suppose $\mathcal{D}_S^1$ is empty. Then, $\mathcal{D}_S=\mathcal{D}_S^2\cup\mathcal{D}_S^3$.  However, no rule in  $\mathcal{D}_S^2$ or $\mathcal{D}_S^3$ can
satisfy (\ref{eq:incomplete.4}). To this this, suppose $\hat{\delta}_{-\infty}\in\mathcal{D}_S^2$
did satisfy (\ref{eq:incomplete.4}). Then, it must hold (as $\tilde{\delta}$
is not dominated in $\{P_{1},P_{2}\}$)
\[
R_{g}(\hat{\delta}_{-\infty},P)=R_{g}(\tilde{\delta},P),\text{ for each }P\in\{P_{1},P_{2}\},
\]
which violates (\ref{eq:different.risk.profile}). Similarly, if any $\hat{\delta}_{\infty}\in\mathcal{D}_S^3$
did satisfy (\ref{eq:incomplete.4}), it also must hold
\[
R_{g}(\hat{\delta}_{\infty},P)=R_{g}(\tilde{\delta},P),\text{ for each }P\in\{P_{1},P_{2}\},
\]
which violates (\ref{eq:different.risk.profile}) too.  Thus, a contradiction
is formed  and we conclude that $\mathcal{D}_{S}=\mathcal{D}_S^2\cup\mathcal{D}_S^3$ is not essentially
complete in $\mathcal{P}$.

Lastly, consider the case when $\mathcal{D}_S^1$ is not empty. Then, $\mathcal{D}_S=\mathcal{D}_S^1\cup\mathcal{D}_S^2\cup\mathcal{D}_S^3$. Clearly, the previous arguments still apply and no rule in $\mathcal{D}_S^2\cup\mathcal{D}_S^3$ satisfies (\ref{eq:incomplete.4}). Furthermore, by Lemma \ref{lem:main},  each $\hat{\delta}_{S}\in\mathcal{D}_{S}^1$  is also dominated in $\left\{ P_{1},P_{2}\right\}$. If any of the rule $\hat{\delta}_{S}\in\mathcal{D}_{S}^1$  did satisfy \eqref{eq:incomplete.4}, then $\tilde{\delta}$ must be dominated in $\left\{ P_{1},P_{2}\right\}$  as well due to \eqref{eq:incomplete.4}, which violates the fact that $\tilde{\delta}$ is not dominated in $\left\{ P_{1},P_{2}\right\}$. Thus, a contradiction
is also formed and we conclude that $\mathcal{D}_{S}$ is not essentially
complete in $\mathcal{P}$.

\subsection*{Proof of Lemma \ref{lem:main}}
Given each $\hat{\delta}_{S}\in\mathcal{D}_{S}^1$, consider
the fractional rule
$\hat{\delta}_{S,\lambda}:=(1-\lambda)\hat{\delta}_{S}+\lambda(1-\hat{\delta}_{S})$
 for some $0<\lambda<1$. When $\hat{\delta}_{S,\lambda}$ is implemented,
the regret of $\hat{\delta}_{S,\lambda}$ is 
\begin{align*}
Reg(\hat{\delta}_{S,\lambda}) & =\tau\left[1\{\tau\geq0\}-(1-\lambda)\hat{\delta}_{S}-\lambda(1-\hat{\delta}_{S})\right]\\
 & =(1-\lambda)Reg(\hat{\delta}_{S})+\lambda Reg(1-\hat{\delta}_{S}).
\end{align*}
The nonlinear regret risk of $\hat{\delta}_{S,\lambda}$ is 
\[
R_{g}(\hat{\delta}_{S,\lambda},P)=\mathbb{E}_{P^{n}}\left[g\left((1-\lambda)Reg(\hat{\delta}_{S})+\lambda Reg(1-\hat{\delta}_{S})\right)\right].
\]
 The derivative of $R_{g}(\hat{\delta}_{S,\lambda},P)$ with respect
to $\lambda$ at each $\lambda>0$ is
\begin{align*}
\frac{\partial R_{g}(\hat{\delta}_{S,\lambda},P)}{\partial\lambda} & =\mathbb{E}_{P^{n}}\left[g'((1-\lambda)Reg(\hat{\delta}_{S})+\lambda Reg(1-\hat{\delta}_{S}))\left(Reg(1-\hat{\delta}_{S})-Reg(\hat{\delta}_{S})\right)\right]\\
 & =\mathbb{E}_{P^{n}}\left[g'((1-\lambda)Reg(\hat{\delta}_{S})+\lambda Reg(1-\hat{\delta}_{S}))\tau(2\hat{\delta}_{S}-1)\right]\\
 & =\tau\left[g'((1-\lambda)Reg(1)+\lambda Reg(0))\mathbb{P}_{P^{n}}(\hat{\delta}_{S}=1)\right]\\
 & -\tau\left[g'((1-\lambda)Reg(0)+\lambda Reg(1))\mathbb{P}_{P^{n}}(\hat{\delta}_{S}=0)\right].
\end{align*}
As $\hat{\delta}_{S}=\hat{\delta}_{S,0}$, it also holds $R_{g}(\hat{\delta}_{S},P)=R_{g}(\hat{\delta}_{S,0},P)$.
Thus, by Taylor's theorem, we may write $R_{g}(\hat{\delta}_{S,\lambda},P)$
as
\begin{align*}
R_{g}(\hat{\delta}_{S,\lambda},P) & =R_{g}(\hat{\delta}_{S},P)+\frac{\partial R_{g}(\hat{\delta}_{S,\lambda},P)}{\partial\lambda}\mid_{\lambda=\tilde{\lambda}}\cdotp\lambda,
\end{align*}
where $\tilde{\lambda}:=\tilde{\lambda}(\lambda,\hat{\delta}_{S},P)$
is a number between $0$ and $\lambda$ that depends on $\lambda$,
$\hat{\delta}_{S}$ and $P$.

If $P=P_{1}$, then $\tau(P_{1})=c>0$, $Reg(1)=0$, $Reg(0)=c$. As $\hat{\delta}_{S}\in\mathcal{D}_{S}^1$,
note $p_{c,\hat{\delta}_{S}}^{+}:=\mathbb{P}_{P_1^{n}}(\hat{\delta}_{S}=1)<1$ and   $\frac{1-p_{c,\hat{\delta}_{S}}^{+}}{p_{c,\hat{\delta}_{S}}^{+}}>0$. Then,
\[
\begin{aligned}\frac{\partial R_{g}(\hat{\delta}_{S,\lambda},P_{1})}{\partial\lambda}\mid_{\lambda=\tilde{\lambda}} & =c\left[g'(\tilde{\lambda}c)\mathbb{P}_{P_1^{n}}(\hat{\delta}_{S}=1)\right]-c\left[g'((1-\tilde{\lambda})c)\mathbb{P}_{P_1^{n}}(\hat{\delta}_{S}=0)\right]\\
 & =c\left[g'(\tilde{\lambda}c)p_{c,\hat{\delta}_{S}}^{+}-g'((1-\tilde{\lambda})c)(1-p_{c,\hat{\delta}_{S}}^{+})\right]<0\\
\end{aligned}
\]
if $\frac{g'(\tilde{\lambda}c)}{g'((1-\tilde{\lambda})c)}<\frac{1-p_{c,\hat{\delta}_{S}}^{+}}{p_{c,\hat{\delta}_{S}}^{+}}.$
Since $0\leq\tilde{\lambda}\leq\lambda$, $\tilde{\lambda}\rightarrow0$
as $\lambda\rightarrow0$. Also, $g'(\cdot)$ is continuously differentiable
on $\mathbb{R}^{+}$ and $c>0$, implying $\lim_{\lambda\rightarrow0}\frac{g'(\tilde{\lambda}c)}{g'((1-\tilde{\lambda})c)}=\frac{g'(0)}{g'(c)}=0$.
Thus, there exists some $0<\lambda<1$ small enough such that $\frac{g'(\tilde{\lambda}c)}{g'((1-\tilde{\lambda})c)}<\frac{1-p_{c,\hat{\delta}_{S}}^{+}}{p_{c,\hat{\delta}_{S}}^{+}}$, so that it holds
\begin{equation}
R_{g}(\hat{\delta}_{S,\lambda},P_{1})-R_{g}(\hat{\delta}_{S},P_{1})=\frac{\partial R_{g}(\hat{\delta}_{S,\lambda},P_{1})}{\partial\lambda}\mid_{\lambda=\tilde{\lambda}}\cdotp\lambda<0.\label{eq:incomplete.2}
\end{equation}
If $P=P_{2}$, $\tau(P)=-c<0$, $Reg(1)=c$, $Reg(0)=0$. Write $p_{c,\hat{\delta}_{S}}^{-}:=\mathbb{P}_{P_2^{n}}(\hat{\delta}_{S}=1)>0$. Then,
\[
\begin{aligned}\frac{\partial R_{g}(\hat{\delta}_{S,\lambda},P_{2})}{\partial\lambda}\mid_{\lambda=\tilde{\lambda}} & =c\left[g'(\tilde{\lambda}c)(1-p_{c,\hat{\delta}_{S}}^{-})-g'((1-\tilde{\lambda})c)p_{c,\hat{\delta}_{S}}^{-}\right].\end{aligned}
\]
Analogous arguments show that by choosing $0<\lambda<1$ small enough,
we have $\frac{g'(\tilde{\lambda}c)}{g'((1-\tilde{\lambda})c)}<\frac{p_{c,\hat{\delta}_{S}}^{-}}{1-p_{c,\hat{\delta}_{S}}^{-}}$
and
\begin{equation}
R_{g}(\hat{\delta}_{S,\lambda},P_{2})-R_{g}(\hat{\delta}_{S},P_{2})=\frac{\partial R_{g}(\hat{\delta}_{S,\lambda},P_{2})}{\partial\lambda}\mid_{\lambda=\tilde{\lambda}}\cdotp\lambda<0.\label{eq:incomplete.3}
\end{equation}
We then pick $\lambda>0$ small enough so that both (\ref{eq:incomplete.2})
and (\ref{eq:incomplete.3}) hold, implying 
\[
R_{g}(\hat{\delta}_{S,\lambda},P)<R_{g}(\hat{\delta}_{S},P),\text{ for each }P\in\left\{ P_{1},P_{2}\right\}.
\]

\subsection*{Proof of Theorem \ref{thm:inadmissibility}}
First note, for any rule $\hat{\delta}$ that is a function of $\hat{\tau}$
(including $\hat{\delta}_{t}$), we can rewrite their nonlinear
regret risk $R_{g}(\hat{\delta},P)$ as $R_{g}(\hat{\delta},\tau)$, since
the nonlinear regret risk of $\hat{\delta}$ depends on $P$ only
via $\tau$. Given a singleton rule $\hat{\delta}_{t}$, Lemma
\ref{lem:admit.2} establishes the existence of a fractional rule that dominates $\hat{\delta}_{t}$. The conclusion of
the theorem directly follows. 

\subsection*{Proof of Theorem \ref{thm:bayes.finite}}
As $g:\mathbb{R}^{+}\to\mathbb{R}^{+}$ and condition (i) of Theorem \ref{thm:bayes.finite} holds, it is straightforward to show (see, for example, Theorem 4.1.1 in \cite{lehmann2006theory}) that the Bayes optimal rule  $\hat{\delta}_{\pi}$ is such that 
\begin{equation}
\hat{\delta}_{\pi}\in \min_{\hat{\delta}\in[0,1]}\int g(Reg(\hat{\delta}))d\pi(P|z_{n}),\text{ for almost all }z_{n}\in\mathbf{Z}_{n},\label{eq:bayes.1}
\end{equation}
provided the solution of (\ref{eq:bayes.1}) exists for almost all $z_{n}\in\mathbf{Z}_{n}$.


Then the existence of $\hat{\delta}_{\pi}$ follows from continuity of the objective function (\ref{eq:bayes.1}) in $\hat{\delta}\in[0,1]$, which itself follows from the fact that $g$ is continuously differentiable. To see $0<\hat{\delta}_{\pi}<1$
for almost all $z_{n}\in\mathbf{Z}_{n}$, note $g^{\prime}(\tau)>0$
for all $\tau>0$ because $g$ is strictly increasing on $\mathbb{R}^+$ by Assumption G. Thus,
\begin{align*}
 & \left[\frac{\partial}{\partial\hat{\delta}}\int g(Reg(\hat{\delta}))d\pi(P|z_{n})\right]_{\hat{\delta}\searrow0}\\
= & -\int\left[\tau(P)g^{\prime}(\tau(P)\mathbf{1}\{\tau(P)\geq0\})\right]d\pi(P|z_{n})\\
= & -\left[\int_{P\in\mathcal{P}:\tau(P)>0}\left[\tau(P)g^{\prime}(\tau(P))\right]d\pi(P|z_{n})+g^{\prime}(0)\int_{P\in\mathcal{P}:\tau(P)<0}\tau(P)d\pi(P|z_{n})\right]\\
= & -\left[\int_{P\in\mathcal{P}:\tau(P)>0}\left[\tau(P)g^{\prime}(\tau(P))\right]d\pi(P|z_{n})\right]<0,
\end{align*}
where the last inequality follows from Assumption G and condition (ii).
Similarly,
\begin{align*}
 & \left[\frac{\partial}{\partial\hat{\delta}}\int g(Reg(\hat{\delta}))d\pi(P|z_{n})\right]_{\hat{\delta}\nearrow1}\\
= & -\int\left[\tau(P)g^{\prime}(\tau(P)(\mathbf{1}\{\tau(P)\geq0\}-1))\right]d\pi(P|z_{n})\\
= & -\left[g^{\prime}(0)\int_{P\in\mathcal{P}:\tau(P)>0}\tau(P)d\pi(P|z_{n})+\int_{P\in\mathcal{P}:\tau(P)<0}\left[g^{\prime}(-\tau(P))\tau(P)\right]d\pi(P|z_{n})\right]\\
= & -\int_{P\in\mathcal{P}:\tau(P)<0}\left[g^{\prime}(-\tau(P))\tau(P)\right]d\pi(P|z_{n})>0.
\end{align*}
 The above calculations imply that we can always reduce $\int g(Reg(\hat{\delta}))d\pi(P|z_{n})$ by moving $\hat{\delta}$ away from both 0 and 1 and toward an interior point. Therefore, $\hat{\delta}_{\pi}$ must be such that
$0<\hat{\delta}_{\pi}<1$, for almost all $z_{n}\in\mathbf{Z}_{n}$. Then, (\ref{eq:bayes.2}) follows from the first order condition for (\ref{eq:bayes.1}).

\subsection*{Proof of Proposition \ref{prop:bayes.flat}}
The proof is similar to the proof of statement (i) of Theorem \ref{thm:bayes.feasible} and thus omitted.

\subsection*{Proof of Theorem \ref{thm:minimax}}

We split the proof into three steps by adopting the `guess-and-verify' approach.

\emph{Step 1}: Guess a least favorable prior. 
Note the worst case mean
square regret of a minimax optimal rule is
\begin{equation}
\sup_{\tau\in\mathbb{R}}R_{sq}(\hat{\delta}^{*},P_{\tau}),\label{pf:main.1}
\end{equation}
where 
\begin{align*}
 R_{sq}(\hat{\delta}^{*},P_{\tau})
 & =\tau^{2}\mathbb{E}[1\{\tau\geq0\}-\hat{\delta}^{*}(\bar{Y}_{1})]^{2}\\
 & =\begin{cases}
\tau^{2}\mathbb{E}\left[1-\hat{\delta}^{*}(\bar{Y}_{1})\right]^{2} & \tau>0,\\
0 & \tau=0,\\
\tau^{2}\mathbb{E}\left[\hat{\delta}^{*}(\bar{Y}_{1})\right]^{2} & \tau<0.
\end{cases}
\end{align*}

By Lemma \ref{lem:1}, the support of the solution of (\ref{pf:main.1})
never contains zero. In Lemma \ref{lem:2}, we show that the support of the solution
of (\ref{pf:main.1}) must be symmetric, i.e., if the support of the solution
of (\ref{pf:main.1}) contains $\tau$ for some $0<\tau<\infty$, it must
also contain $-\tau$. Therefore, we conjecture that the least favorable
prior $\pi^{*}$ is two-point supported. Moreover, Lemma \ref{lem:3} shows that for a symmetric two-point prior to be least favorable, each point is equally likely to be realised. Thus, our guess for the least favorable prior
$\pi^{*}$ is such that
\[
\pi^{*}\left(\tau\right)=\frac{1}{2},\pi^{*}\left(-\tau\right)=\frac{1}{2}, \text{ for some } 0<\tau<\infty.
\]

\emph{Step 2}: Derive the Bayes optimal rule associated with the hypothesized least favorable prior. For each
$0<\tau<\infty$, let $\hat{\delta}_{\pi_{\tau}^{*}}$ be the Bayes
optimal rule with respect to the two-point symmetric prior
\[
\pi_{\tau}^{*}\left(\tau\right)=\frac{1}{2} \; \text{ and } \; \pi_{\tau}^{*}\left(-\tau\right)=\frac{1}{2}.
\]
Within the above set of candidate least favorable priors, we show: (1) the Bayes optimal rules admit the form $\hat{\delta}_{\pi^{*}_{\tau}}(\bar{Y_{1}})=\frac{\exp\left(2\tau\bar{Y}_{1}\right)}{\exp\left(2\tau\bar{Y}_{1}\right)+1}$; (2)  $r_{sq}(\hat{\delta}_{\pi^{*}_{\tau}},\pi^{*}_{\tau})$  follows the form in (\ref{eq:minimax.1}), and is equivalent to the form in (\ref{eq:minimax.2}). Thus, our guess for the least favorable prior is 
\[
\pi^{*}\left(\tau^{*}\right)=\frac{1}{2},\pi^{*}\left(-\tau^{*}\right)=\frac{1}{2},
\]
where $\tau^{*}$ solves (\ref{eq:minimax.1}) or (\ref{eq:minimax.2}).

Indeed, the functional form of $\hat{\delta}_{\pi_{\tau}^{*}}$ is derived by applying Theorem \ref{thm:bayes.finite}, 
\[
\hat{\delta}_{\pi_{\tau}^{*}}(\bar{y}_{1})=\frac{\int\tau^{2}1\{\tau\geq0\}d\pi_{\tau}^{*}(\tau|\bar{y}_{1})}{\int\tau^{2}d\pi_{\tau}^{*}(\tau|\bar{y}_{1})},
\]
where $\pi_{\tau}^{*}(\tau|\bar{y}_{1})$ is the posterior distribution
of $\pi_{\tau}^{*}$ conditional on $\bar{Y}_{1}=\bar{y}_{1}$ and
admits: 
\begin{align*}
\pi_{\tau}^{*}\{\tau|\bar{y}_{1}\}  =\frac{\frac{1}{2}f\{\bar{y}_{1}|\tau\}}{f\{\bar{y}_{1}\}}
 \; \text{ and } \; 
\pi_{\tau}^{*}\{-\tau|\bar{y}_{1}\}  =\frac{\frac{1}{2}f\{\bar{y}_{1}|-\tau\}}{f\{\bar{y}_{1}\}},
\end{align*}
where $f\{\bar{y}_{1}|\tau\}$ is the likelihood of $\tau$, $f\{\bar{y}_{1}|-\tau\}$ is the likelihood of $-\tau$, and $f\{\bar{y}_{1}\}$ is the marginal density
of $\bar{Y}_{1}$. Note $f\{\bar{y}_{1}|\tau\}=\sqrt{\frac{1}{2\pi}}\exp\left(-\frac{1}{2}\left[\left(\bar{y}_{1}-\tau\right)^{2}\right]\right)>0$, $
f\{\bar{y}_{1}|-\tau\}=\sqrt{\frac{1}{2\pi}}\exp\left(-\frac{1}{2}\left[\left(\bar{y}_{1}+\tau\right)^{2}\right]\right)>0$. 
It follows that
\begin{align*}
\hat{\delta}_{\pi_{\tau}^{*}}(\bar{y}_{1}) & =\frac{f\{\bar{y}_{1}|\tau\}}{f\{\bar{y}_{1}|\tau\}+f\{\bar{y}_{1}|-\tau\}}= \frac{\exp\left(2\tau\bar{y}_{1}\right)}{\exp\left(2\tau\bar{y}_{1}\right)+1}.
\end{align*}
Therefore,
the Bayes mean square regret of $\hat{\delta}_{\pi_{\tau}^{*}}$ admits the form in (\ref{eq:minimax.1}):
\begin{align*}
r_{sq}(\hat{\delta}_{\pi_{\tau}^{*}},\pi_{\tau}^{*})&=  \frac{1}{2}\tau^{2}\int\left(\frac{f\{\bar{y}_{1}|-\tau\}}{f\{\bar{y}_{1}|\tau\}+f\{\bar{y}_{1}|-\tau\}}\right)^{2}f\{\bar{y}_{1}|\tau\}d\bar{y}_{1}\\
&\;\;\; +  \frac{1}{2}\tau^{2}\int\left(\frac{f\{\bar{y}_{1}|\tau\}}{f\{\bar{y}_{1}|\tau\}+f\{\bar{y}_{1}|-\tau\}}\right)^{2}f\{\bar{y}_{1}|-\tau\}d\bar{y}_{1}\\
&=  \frac{1}{2}\tau^{2}\int\frac{f\{\bar{y}_{1}|-\tau\}f\{\bar{y}_{1}|\tau\}}{\left[f\{\bar{y}_{1}|\tau\}+f\{\bar{y}_{1}|-\tau\}\right]}d\bar{y}_{1}\\
&=  \frac{1}{2}\tau^{2}\mathbb{E}\left[\frac{1}{\exp\left(2\tau\bar{Y}_{1}\right)+1}\right].
\end{align*}
Since $\tau>0$ and $R_{sq}(\hat{\delta}_{\pi_{\tau}^{*}},P_{\tau}) = \tau^2 \mathbb{E}[1 - \hat{\delta}_{\pi_{\tau}^{*}}(\bar{Y}_1) ]^2$, we see that (\ref{eq:minimax.1}) is equivalent to (\ref{eq:minimax.2}):
\begin{align*}
R_{sq}(\hat{\delta}_{\pi_{\tau}^{*}},P_{\tau}) & =\tau^{2}\int\left[\frac{f\{\bar{y}_{1}|-\tau\}}{f\{\bar{y}_{1}|\tau\}+f\{\bar{y}_{1}|-\tau\}}\right]^{2}f(\bar{y}_{1}|\tau)d\bar{y}_{1}\\
 & =\tau^{2}\int\left[\frac{f\{\bar{y}_{1}|-\tau\}}{f\{\bar{y}_{1}|\tau\}+f\{\bar{y}_{1}|-\tau\}}\right]^{2}\sqrt{\frac{1}{2\pi}}\exp\left(-\frac{1}{2}\left[\left(\bar{y}_{1}-\tau\right)^{2}\right]\right)d\bar{y}_{1}\\
 & =\tau^{2}\int\left[\frac{1}{\exp\left(2\tau\bar{y}_{1}\right)+1}\right]^{2}\sqrt{\frac{1}{2\pi}}\exp\left(-\frac{1}{2}\left[\left(\bar{y}_{1}-\tau\right)^{2}\right]\right)d\bar{y}_{1}\\
 & =\tau^{2}\mathbb{E}\left[\frac{1}{\exp\left(2\tau\bar{Y}_{1}\right)+1}\right]^{2},
\end{align*}
and by a change of variables, 
\begin{align*}
R_{sq}(\hat{\delta}_{\pi_{\tau}^{*}},P_{-\tau}) & =\tau^{2}\int\left(\frac{f\{\bar{y}_{1}|\tau\}}{f\{\bar{y}_{1}|\tau\}+f\{\bar{y}_{1}|-\tau\}}\right)^{2}f\{\bar{y}_{1}|-\tau\}d\bar{y}_{1}\\
&=R_{sq}(\hat{\delta}_{\pi_{\tau}^{*}},P_{\tau}).
\end{align*}


\emph{Step 3}: For our guess of the least favorable prior, Lemma \ref{lem:5} further establishes that Condition \ref{cond:1} holds. Thus, Proposition \ref{prop:minimax} implies that $\hat{\delta}^{*}$ is indeed a minimax optimal rule and the two-point prior $\pi^{*}(\tau^{*})=\pi^{*}(-\tau^{*})=\frac{1}{2}$ is indeed least favorable.



\subsection*{Proof of Theorem \ref{thm:minimax.feasible}}

\subsubsection*{Proof of statement (i)}
Let $\hat{\delta}^{*}$ be a minimax optimal rule in the limit experiment.
That is, $\hat{\delta}^{*}$ solves
\[
\inf_{\hat{\delta}}\sup_{h}R_{sq}^{\infty}(\hat{\delta},h):=R^{*}.
\]
Following \cite{HiranoPorter2009}, consider slicing the parameter
space of $h$ in the following way: define
\[
h_{1}(b,h_{0})=h_{0}+\frac{b}{\dot{\tau}^{\prime}I_{0}^{-1}\dot{\tau}}I_{0}^{-1}\dot{\tau},
\]
where $h_{0}$ is such that $\dot{\tau}^{\prime}h_{0}=0$ (without
loss of generality) and $b\in\mathbb{R}.$ Hence, 
\[
\dot{\tau}^{\prime}h_{1}(b,h_{0})=b.
\]
Note for each $\hat{\delta}\in[0,1]$, the limit regret $Reg_{\infty}(\hat{\delta},h)$ only
depends on $h$ through $\dot{\tau}^{\prime}h$.
Thus, we can consider treatment rules of the form
\[
\hat{\delta}(\varDelta)=\hat{\delta}\left(\dot{\tau}^{\prime}\varDelta\right):=\hat{\delta}_{\tau}(\varDelta_{\tau}),
\]
where $\varDelta_{\tau}:=\dot{\tau}^{\prime}\varDelta\sim N(\dot{\tau}^{\prime}h,\dot{\tau}^{\prime}I_{0}^{-1}\dot{\tau})$.
Let $\hat{\delta}_{\tau}^{*}$ solve the simpler minimax exercise \[\inf_{\hat{\delta}_{\tau}}\sup_{b}R_{sq}^{\infty}(\hat{\delta}_{\tau},h_{1}(b,0))
\] among rules of form $\hat{\delta}_{\tau}(\varDelta_{\tau})$. It follows by Lemma \ref{lem:C1} that $\hat{\delta}_{\tau}^{*}$ is a minimax optimal rule.
Define 
\[
R_{sq}^{\infty}(\hat{\delta}_{\tau},b):=b^{2}\mathbb{E}_{\varDelta_{\tau}\sim N(b,\dot{\tau}^{\prime}I_{0}^{-1}\dot{\tau})}\left[1\left\{ b\geq0\right\} -\hat{\delta}_{\tau}\left(\varDelta_{\tau}\right)\right]^{2}.
\]
Lemma \ref{lem:C2} shows that $\hat{\delta}_{\tau}^{*}$ can be found by
solving $\inf_{\hat{\delta}_{\tau}}\sup_{b}R_{sq}^{\infty}(\hat{\delta}_{\tau},b)$, and Lemma \ref{lem:C3} establishes the the form of $\hat{\delta}_{\tau}^{*}$, which is a minimax optimal rule in the limit experiment.

\subsubsection*{Proof of statement (ii)}
By \citet[Lemma 3,][]{HiranoPorter2009}, $\sqrt{n}\frac{\tau(\hat{\theta})}{\hat{\sigma}_{\tau}}\overset{h}{\rightsquigarrow}N(\frac{\dot{\tau}^{\prime}h}{\sigma_{\tau}},1)$. Furthermore, using the continuous mapping theorem,
\[
\hat{\delta}_{F}^{*}(Z_{n})\overset{h}{\rightsquigarrow}\frac{\exp\left(2\tau^{*}N(\frac{\dot{\tau}^{\prime}h}{\sigma_{\tau}},1)\right)}{\exp\left(2\tau^{*}N(\frac{\dot{\tau}^{\prime}h}{\sigma_{\tau}},1)\right)+1}.
\]
Therefore, $\hat{\delta}_{F}^{*}$ is matched with $\hat{\delta}^{*}$
in the limit experiment in the sense of Proposition \ref{prop:limit}.
The desired conclusion follows via a similar argument to that in \citet[Lemma 4,][]{HiranoPorter2009}.

\subsection*{Proof of Theorem \ref{thm:bayes.feasible}}

\subsubsection*{Proof of statement (i)}
Applying Theorem \ref{thm:bayes.finite} to the limit Bayes mean square criterion $r_{sq}^{\infty}$
yields
\[
\hat{\delta}_{B}(\varDelta)=\frac{\int\left(\dot{\tau}^{\prime}h\right)^{2}\left(\mathbf{1}\left\{ \dot{\tau}^{\prime}h\geq0\right\} \right)d\pi(h|\varDelta)}{\int\left(\dot{\tau}^{\prime}h\right)^{2}d\pi(h|\varDelta)}.
\]
Notice in the limit experiment, $h$ has a flat prior. It follows that
the posterior distribution $\pi(h|\varDelta)$ is proportional to
a normal distribution with mean $\varDelta$ and variance $I_{0}^{-1}$.
Then
\begin{align*}
\hat{\delta}_{B}(\varDelta) & =\frac{\int\left(\dot{\tau}^{\prime}h\right)^{2}\left(\mathbf{1}\left\{ \dot{\tau}^{\prime}h\geq0\right\} \right)dN(h|\varDelta,I_{0}^{-1})}{\int\left(\dot{\tau}^{\prime}h\right)^{2}dN(h|\varDelta,I_{0}^{-1})}\\
 & =\frac{\int s^{2}\left(\mathbf{1}\left\{ s\geq0\right\} \right)dN(s|\dot{\tau}^{\prime}\varDelta,\dot{\tau}^{\prime}I_{0}^{-1}\dot{\tau})}{\int s^{2}dN(s|\dot{\tau}^{\prime}\varDelta,\dot{\tau}^{\prime}I_{0}^{-1}\dot{\tau})}\\
 & =\frac{\int_{s\geq0}s^{2}dN(s|\dot{\tau}^{\prime}\varDelta,\sigma_{\tau}^{2})}{\int s^{2}dN(s|\dot{\tau}^{\prime}\varDelta,\sigma_{\tau}^{2})}\\
 & =\int\mathbf{1}\left\{ s\geq0\right\} dN(s|\dot{\tau}^{\prime}\varDelta,\sigma_{\tau}^{2})\frac{\int s^{2}dN(s|\dot{\tau}^{\prime}\varDelta,\sigma_{\tau}^{2},S\geq0)}{\int s^{2}dN(s|\dot{\tau}^{\prime}\varDelta,\sigma_{\tau}^{2})},
\end{align*}
where $\int s^{2}dN(s|\dot{\tau}^{\prime}\varDelta,\sigma_{\tau}^{2},S\geq0)$
denotes the conditional expectation of a normal random variable $S$
with mean $\dot{\tau}^{\prime}\varDelta$ and variance $\sigma_{\tau}^{2}$
given $S\geq0$. By the properties of the normal distribution and truncated normal
distribution, 
\begin{align*}
\int\mathbf{1}\left\{ s\geq0\right\} dN(s|\dot{\tau}^{\prime}\varDelta,\sigma_{\tau}^{2}) & =\Phi(\frac{\dot{\tau}^{\prime}\varDelta}{\sigma_{\tau}}),\\
\int s^{2}dN(s|\dot{\tau}^{\prime}\varDelta,\sigma_{\tau}^{2}) & =\sigma_{\tau}^{2}+\left(\dot{\tau}^{\prime}\varDelta\right)^{2},\\
\int s^{2}dN(s|\dot{\tau}^{\prime}\varDelta,\sigma_{\tau}^{2},S\geq0) & =\sigma_{\tau}^{2}\left(1-\frac{\dot{\tau}^{\prime}\varDelta}{\sigma_{\tau}}\frac{\phi(\frac{\dot{\tau}^{\prime}\varDelta}{\sigma_{\tau}})}{\Phi(\frac{\dot{\tau}^{\prime}\varDelta}{\sigma_{\tau}})}-\frac{\phi^{2}(\frac{\dot{\tau}^{\prime}\varDelta}{\sigma_{\tau}})}{\Phi^{2}(\frac{\dot{\tau}^{\prime}\varDelta}{\sigma_{\tau}})}\right)+\left(\left(\dot{\tau}^{\prime}\varDelta\right)+\sigma_{\tau}\frac{\phi(\frac{\dot{\tau}^{\prime}\varDelta}{\sigma_{\tau}})}{\Phi(\frac{\dot{\tau}^{\prime}\varDelta}{\sigma_{\tau}})}\right)^{2}.
\end{align*} 
Statement (i) follows.

\subsubsection*{Proof of statement (ii)}
Similar to the argument in the proof of statement
(ii) of Theorem \ref{thm:minimax.feasible}, $\hat{\delta}_{B,F}^{*}$ is matched with $\hat{\delta}_{B}$
in the limit experiment in the sense of Proposition \ref{prop:limit}.
The conclusion follows via a similar argument to that in \citet[Lemma 1,] []{HiranoPorter2009}.

\bibliographystyle{ecta}
\bibliography{bibliography}

\newpage
\renewcommand{\thepage}{OA-\arabic{page}}
\setcounter{page}{1}

\part*{Appendix for Online Publication}

\section{Comparison with \texorpdfstring{\cite{manski2007admissible}}{Manski and Tetenov (2007)}}\label{sec:app.compare} 

In this section, we clarify the differences between our approach of treatment choice with nonlinear regret criteria and the approach of risk averse welfare criteria taken by \cite{manski2007admissible}. To elaborate, let $f(\cdotp):\mathbb{R}\rightarrow\mathbb{R}$ be a concave function. A concave
transformation of $W(\hat{\delta})$ is $f(W(\hat{\delta}))$.
For the concave transformation $f$, the regret of treatment rule $\hat{\delta}$ defined in terms of \emph{nonlinear welfare} is 
\[
Reg^f(\hat{\delta}) =f(W(\delta^{*}))-f(W(\hat{\delta}))
=f(\mu_{0}+\delta^{*}\tau)-f(\mu_{0}+\hat{\delta}\tau).
\]
In contrast, our paper considers a nonlinear (possibly convex) transformation of regret measured in terms of the original welfare:
\[
g(Reg(\hat{\delta}))  =g\left(\left(\mu_{0}+\delta^{*}\tau\right)-(\mu_{0}+\hat{\delta}\tau)\right),
\]
where $g:\mathbb{R}^{+}\rightarrow\mathbb{R}^{+}$ is a nonlinear function
that does not depend on $\hat{\delta}$, $\mu_{0}$ or $\mu_{1}$. In other words, the loss function in \cite{manski2007admissible} is $Reg^f(\hat{\delta})$ while in our paper the loss function is $g(Reg(\hat{\delta}))$. Note we may write
(assuming $f(\cdotp)$ is sufficiently differentiable)
\[
f(W(\hat{\delta}))=f(W(\delta^{*}))+\sum_{j=1}^{\infty}\frac{\partial^{j}f(x)}{\partial x^{j}}\mid_{x=W(\delta^{*})}\left(W(\hat{\delta})-W(\delta^{*})\right)^{j}.
\]
Since $W(\hat{\delta})-W(\delta^{*})=-Reg(\hat{\delta})$, the \textit{expected
regret} based on \textit{nonlinear welfare} is
\begin{align*}
\mathbb{E}_{P^{n}}[Reg^{f}(\hat{\delta})] & =\sum_{j=1}^{\infty}\underset{\text{state-dependent weights}}{\underbrace{\left[\left(-1\right)^{j+1}\frac{\partial^{j}f(x)}{\partial x^{j}}\mid_{x=W(\delta^{*})}\right]}}\underset{\text{nonlinear regret}}{\underbrace{\mathbb{E}_{P^{n}}\left[\left(Reg(\hat{\delta})\right)^{j}\right]}}.
\end{align*}
Thus, a risk-averse decision maker in the sense of \citet{manski2007admissible}
also cares about the regret distribution beyond the mean. However,
the nonlinear welfare approach effectively considers all moments
of the regret (up to some limit determined by the specific $f$) additively. Moreover, how the decision maker weighs
different moments of the regret is state-dependent. That is, different values of $\tau$ results in different
relative importance of each moment. While our nonlinear
regret approach also aims to bring other features of the regret distribution
into consideration, a typical nonlinear regret criterion often only
focuses on one particular moment of regret and is less convoluted. 
The following example further illustrates the differences.
\begin{example}[expected regret with a quadratic and concave transformation of welfare]
Suppose $\mu_{0}=0$, $\tau\in[-1,1]$, and $f(x)=-x^{2}+2x$ is a
quadratic and concave function. As $\frac{\partial f(x)}{\partial x}=-2x+2$,
$\frac{\partial^{2}f(x)}{\partial x^{2}}=-2$, and $\frac{\partial^{j}f(x)}{\partial x^{j}}=0$
for all $j\geq3$, it follows that 
\begin{align*}
\frac{1}{2}\mathbb{E}_{P^{n}}[Reg^{f}(\hat{\delta})] & =\left[1-\mathbf{1}\left\{ \tau\geq0\right\} \tau\right]\underset{\text{mean regret}}{\underbrace{\mathbb{E}_{P^{n}}\left[Reg(\hat{\delta})\right]}}+\underset{\text{mean square regret}}{\underbrace{\mathbb{E}_{P^{n}}\left[\left(Reg(\hat{\delta})\right)^{2}\right]}}\\
 & =\begin{cases}
\left[1-\tau\right]\mathbb{E}_{P^{n}}\left[Reg(\hat{\delta})\right]+\mathbb{E}_{P^{n}}\left[\left(Reg(\hat{\delta})\right)^{2}\right] & \tau>0\\
\mathbb{E}_{P^{n}}\left[Reg(\hat{\delta})\right]+\mathbb{E}_{P^{n}}\left[\left(Reg(\hat{\delta})\right)^{2}\right] & \tau\leq0
\end{cases}.
\end{align*}
Thus, compared to our mean square regret, the criterion of expected
regret based on a quadratic, concave transformation of welfare always
gives more weight to the mean of regret when it comes to the ranking
of different decision rules. In addition, with a positive treatment
effect, $\mathbb{E}_{P^{n}}[Reg^{f}(\hat{\delta})]$ would put more
weight on $\mathbb{E}_{P^{n}}[Reg(\hat{\delta})]$
when $\tau$ is smaller. If the treatment effect is negative, the
relative importance of $\mathbb{E}_{P^{n}}[Reg(\hat{\delta})]$
and $\mathbb{E}_{P^{n}}[(Reg(\hat{\delta}))^{2}]$
stays the same. 
\end{example}
Therefore, both our approach and that of \citet{manski2007admissible}
share a similar motivation: when it comes to the selection of a good
rule, decision makers may wish to take other aspects of the regret
distribution into account. As a result, fractional rules arise. However, how these two approaches operate
is dramatically different, and it seems our approach offers a more
tractable framework to work out different optimal rules. Below, we take a step further to show that these two approaches
are inherently two different ways of ranking decision rules. 

\begin{prop} \label{prop:nonlinear}
Consider the following statement
\begin{equation}\label{eq:equal.risk}
\mathbb{E}_{P^{n}}{[}Reg^{f}(\hat{\delta}){]}=\mathbb{E}_{P^{n}}{[}g(Reg(\hat{\delta})){]},\text{ for all }\hat{\delta},\,\mu_{0}\text{ and }\mu_{1}.
\end{equation}
Then (\ref{eq:equal.risk}) holds for some concave function $f(\cdotp):\mathbb{R}\rightarrow\mathbb{R}$
and some function $g(\cdotp):\mathbb{R}^{+}\rightarrow\mathbb{R}$
if and only if $f(x)=ax+b$ and $g(x)=ax$ for some constants $a$
and $b$.
\end{prop}
\begin{proof}
The \emph{if} part is straightforward to show. We focus on the
\emph{only if} part. Let $\mathbb{F}(\cdotp):=\mathbb{E}_{P^{n}}[f(\cdotp)]$
and $\mathbb{G}(\cdotp):=\mathbb{E}_{P^{n}}[g(\cdotp)]$. Since convexity
and concavity are preserved under the expectation operator, it holds
that $\mathbb{F}(\cdotp)$ is concave too. Then, by assumption, 
\begin{equation}
\mathbb{F}(\mu_{0}+\delta^{*}\tau)-\mathbb{F}(\mu_{0}+\hat{\delta}\tau)=\mathbb{G}((\mu_{0}+\delta^{*}\tau)-(\mu_{0}+\hat{\delta}\tau))\label{pf:FG.1}
\end{equation}
for all $\mu_{0},$ $\mu_{1}$ and $\hat{\delta}$, implying 
\begin{equation}
\mathbb{F}(x)-\mathbb{F}(y)=\mathbb{G}\left(x-y\right),\forall x\geq y.\label{eq:FG.2}
\end{equation}
Fixing $y=0$, (\ref{eq:FG.2}) implies 
\begin{equation}
\mathbb{F}(x)-\mathbb{F}(0)=\mathbb{G}(x),\forall x\geq0.\label{eq:FG.3}
\end{equation}
Since $\mathbb{F}$ is concave, (\ref{eq:FG.3}) implies $\mathbb{G}(x)$
is concave as well for all $x\geq0$. Conversely, fixing $x=0$,
(\ref{eq:FG.2}) implies 
\begin{equation}
\mathbb{F}(0)-\mathbb{F}(y)=\mathbb{G}\left(-y\right),\forall y\leq0,\label{eq:FG.4}
\end{equation}
or, equivalently, 
\begin{equation}
\mathbb{F}(0)-\mathbb{F}(-x)=\mathbb{G}(x),\forall x\geq0.\label{eq:FG.5}
\end{equation}

Since $\mathbb{F}$ is concave, (\ref{eq:FG.5}) implies $\mathbb{G}(x)$
is convex for all $x\geq0$. Thus, $\mathbb{G}(x)$ must be both concave
and convex for $x\geq0$, implying $\mathbb{G}(x)$ is an affine function
for all $x\geq0$. This implies $g$ is affine and admits $g(x)=ax+t$
for some constants $a$ and $t$. Since $\mathbb{G}(x)$ is affine,
(\ref{eq:FG.3}) implies that $\mathbb{F}(x)$ is affine for $x\geq0$
and $f(x)=ax+t+\mathbb{F}(0)$ for $x\geq0$. Furthermore, for all $y\leq0$,
(\ref{eq:FG.4}) implies $\mathbb{F}(y)=\mathbb{F}(0)-\mathbb{G}(-y)$,
i.e., $\mathbb{F}(y)$ is affine for $y\leq0$ as well, and $f(y)=ay-t+\mathbb{F}(0)$
for $y\leq0$. At $x=0$, $t+\mathbb{F}(0)=-t+\mathbb{F}(0)$ must hold,
implying $t=0$. Thus, $g(x)=ax$ and $f(x)=ax+\mathbb{F}(0)$ must hold
or, equivalently, $f(x)=ax+b$ for some constants $a$ and $b$.
\end{proof}
Given a concave transformation $f$ of the welfare considered in \cite{manski2007admissible}, Proposition \ref{prop:nonlinear} shows that we cannot find a nonlinear transformation $g$ of the original regret such that the regret of nonlinear welfare defined in \cite{manski2007admissible} equals our nonlinear regret risk for all rules and all states of the world. The results of Proposition \ref{prop:nonlinear} can be extended in
several ways. Firstly, Proposition \ref{lem:A1} shows that even if
we consider either $f$ or $g$ to be convex, or we restrict the domain
of $f$ to be positive, the results of Proposition \ref{prop:nonlinear}
continue to hold. For instance, suppose $g(r)=r^{2}$ and a nonlinear
welfare transformation $f$ were to exist so that
\begin{equation}
\mathbb{E}_{P^{n}}[f(W(\delta^{*}))-f(W(\hat{\delta}))]=\mathbb{E}_{P^{n}}{[}(W(\delta^{*})-W(\hat{\delta}))^{2}{]},\forall\hat{\delta},\mu_{0},\text{and }\mu_{1}.\label{eq:FG.example}
\end{equation}
Proposition \ref{lem:A1} shows that such an $f$ does not exist.
Secondly, one might argue that even though (\ref{eq:equal.risk})
does not hold, the risks of the two approaches could be affine transformations
of each other, so that the optimal rules are the same. In Proposition
\ref{lem:A2}, we show that even in such a scenario, both $f$ and
$g$ also have to be affine. Our approach in introducing nonlinear $g(\cdot)$ is inherently different from that of \cite{manski2007admissible}.

\begin{prop}\label{lem:A1}
\begin{itemize}
\item[(i)] (\ref{eq:equal.risk}) holds for some convex function $f(\cdotp):\mathbb{R}\rightarrow\mathbb{R}$ and some function $g(\cdotp):\mathbb{R}^{+}\rightarrow\mathbb{R}$ if and only if $f(x)=ax+b$ and $g(x)=ax$ for some constants $a$
and $b$.

\item[(ii)] (\ref{eq:equal.risk}) holds for some function $f(\cdotp):\mathbb{R}\rightarrow\mathbb{R}$ and some convex function $g(\cdotp):\mathbb{R}^{+}\rightarrow\mathbb{R}$ if and only if $f(x)=ax+b$ and $g(x)=ax$ for some constants $a$
and $b$.

\item[(iii)] (\ref{eq:equal.risk}) holds  for some concave function $f(\cdotp):\mathcal{C}\rightarrow\mathbb{R}$, where $\mathcal{C}\subseteq\mathbb{R}^{+}$ is a compact interval, and some function $g(\cdotp):\mathbb{R}^{+}\rightarrow\mathbb{R}$  if and only if $f(x)=ax+b$ and $g(x)=ax$ for some constants $a$
and $b$.
\end{itemize}
\end{prop}
\begin{proof}
Statement (i): the proof is the same as that of Proposition \ref{prop:nonlinear}.

Statement (ii): We only show the \emph{only if} part. Note (\ref{eq:FG.2}) still holds. Fix $T\in\mathbb{R}$.
It holds that 
\[
\mathbb{F}(T)-\mathbb{F}(y)=\mathbb{G}(T-y),\forall y\leq T,
\]
or, equivalently, that
\begin{equation}
\mathbb{F}(y)=\mathbb{F}(T)-\mathbb{G}(T-y),\forall y\leq T.\label{pf:lemmaA1.1}
\end{equation}
Since $\mathbb{G}$ is convex, (\ref{pf:lemmaA1.1}) implies $\mathbb{F}(y)$
is concave for all $\forall y\leq T$. Letting $T\rightarrow\infty$
implies $\mathbb{F}(y)$ is a concave function in $\mathbb{R}$. The
rest of the proof then follows that of Proposition \ref{prop:nonlinear}.

Statement (iii). We only show the \emph{only if} part. Without loss of generality, suppose $f(\cdotp):[0,1]\rightarrow\mathbb{R}$.
Note (\ref{eq:FG.2}) still holds for all $1\geq x \geq y \geq 0$, implying
\begin{equation}
\mathbb{F}(x)-\mathbb{F}(0.5) =\mathbb{G}(x-0.5),\forall0.5\leq x\leq1,\label{pf:lemmaA1.2}
\end{equation}
i.e.,  $\mathbb{G}$ is concave on the interval $[0,0.5]$. Conversely, (\ref{eq:FG.2}) also implies 
\[
\mathbb{F}(0.5)-\mathbb{F}(y)=\mathbb{G}(0.5-y),\forall0\leq y\leq0.5,
\]
which means $\mathbb{G}$ is convex on $[0,0.5]$. Thus,
$\mathbb{G}$ must be affine on $[0,0.5]$, and $g(x)=ax+t$ for some constants $a$ and $t$, for each $0\leq x \leq 0.5$. Further note $\mathbb{F}(x)-\mathbb{F}(0)	=\mathbb{G}(x),\forall0\leq x\leq1$. Thus, combining this with (\ref{pf:lemmaA1.2}), we find 
\[
\mathbb{G}(x)=\mathbb{F}(x)-\mathbb{F}(0)=\mathbb{G}(x-0.5)+\mathbb{F}(0.5)-\mathbb{F}(0),\forall0.5\leq x\leq1.
\] 
In particular, at $x=0.5$, $\mathbb{G}(0.5)=\mathbb{G}(0)+\mathbb{F}(0.5)-\mathbb{F}(0)$. Hence, 
\[
\mathbb{G}(x)=\mathbb{G}(x-0.5)+\mathbb{G}(0.5)-\mathbb{G}(0),\forall0.5\leq x\leq1,
\]
implying $\mathbb{G}(x)$ is affine and  $g(x)=ax+t$ in $[0.5,1]$ as well. Thus,  $g(x)=ax+t$ for $0\leq x \leq 1$. But then plugging $x=0.5$ into ($\ref{pf:lemmaA1.2}$) implies $t=0$. Then, it is easy to see that $f(x)=ax+b$ for some constant $b$.
\end{proof}
\begin{prop}
\label{lem:A2} Let $A>0$ and $B\in\mathbb{R}$ be some constants.
Consider the following statement 
\begin{equation}
\mathbb{E}_{P^{n}}{[}Reg^{f}(\hat{\delta}){]}=A\mathbb{E}_{P^{n}}{[}g(Reg(\hat{\delta})){]}+B,\text{ for all }\hat{\delta},\,\mu_{0}\text{ and }\mu_{1}.\label{eq:equal.ranking.risk}
\end{equation}
Then (\ref{eq:equal.ranking.risk}) holds for some concave function
$f(\cdotp):\mathbb{R}\rightarrow\mathbb{R}$, some function $g(\cdotp):\mathbb{R}^{+}\rightarrow\mathbb{R}$
and some constants $A>0$ and $B$ if and only if $f(x)=ax+b$ and
$g(x)=\frac{a}{A}x-\frac{B}{A}$ for some constants $a$, $b$, $A>0$
and $B$.
\end{prop}

\begin{proof}
To see the \emph{only if} part, let $\tilde{g}(x)=Ag(x)+B$. (\ref{eq:equal.ranking.risk})
implies 
\[
\mathbb{E}_{P^{n}}{[}Reg^{f}(\hat{\delta}){]}=\mathbb{E}_{P^{n}}{[}\tilde{g}(Reg(\hat{\delta})){]},\text{ for all }\hat{\delta},\mu_{0}\text{ and }\mu_{1}.
\]
Applying the results of Proposition \ref{prop:nonlinear} yields that
$Ag(x)+B=ax$ and $f(x)=ax+b$ for some constants $a$, $b$, $A>0$ and $B$. That
is, $g(x)=\frac{a}{A}x-\frac{B}{A}$. The $if$ part is straightforward to show and omitted.
\end{proof}

\section{Lemmas for Section \ref{sec:admit}}

\begin{lem}\label{lem:admit.1}
Under the conditions of Theorem \ref{thm:incomplete}, there exists a rule  $\tilde{\delta}$ that is not dominated in $\{P_1,P_2\}$, and is such that \eqref{eq:different.risk.profile} holds.
\end{lem}

\begin{proof}
Consider a Bayes optimal rule with respect to a prior $\pi$ such
that $\pi(P_{1})=\pi(P_{2})=\frac{1}{2}$. It follows from \citet[Theorem 2, Chapter 2.3]{ferguson1967mathematical}
that this Bayes optimal rule is not dominated in $\{P_{1},P_{2}\}$.
Below, we show that no rule in $\mathcal{D}_{S}^{2}$
or $\mathcal{D}_{S}^{3}$ can be Bayes optimal
with respect to such a prior, implying that the Bayes optimal rule
must satisfy \eqref{eq:different.risk.profile}. To see this, consider
a rule $\hat{\delta}_{\varepsilon}$ such that $\mathbb{P}_{P_{1}^{n}}\left\{ \hat{\delta}_{\varepsilon}=\varepsilon\right\} =\mathbb{P}_{P_{2}^{n}}\left\{ \hat{\delta}_{\varepsilon}=\varepsilon\right\} =1$
for some small $\varepsilon>0$. The Bayes nonlinear regret risk of
$\hat{\delta}_{\varepsilon}$ with respect to prior $\pi(P_{1})=\pi(P_{2})=\frac{1}{2}$
is
\begin{align*}
r_{g}\left(\hat{\delta}_{\varepsilon},\pi\right) & =\frac{1}{2}\mathbb{E}_{P_{1}^{n}}\left[g\left(c(1-\hat{\delta}_{\varepsilon})\right)\right]+\frac{1}{2}\mathbb{E}_{P_{2}^{n}}\left[g\left(c\hat{\delta}_{\varepsilon}\right)\right]\\
 & =\frac{1}{2}g\left(c(1-\varepsilon)\right)+\frac{1}{2}g\left(c\varepsilon\right).
\end{align*}
The directional derivative of $r_{g}\left(\hat{\delta}_{\varepsilon},\pi\right)$
(from above) with respect to $\varepsilon$ when $\varepsilon=0$
is 
\begin{equation}
\frac{\partial r_{g}\left(\hat{\delta}_{\varepsilon},\pi\right)}{\partial\varepsilon}\mid_{\varepsilon\downarrow0}=-\frac{1}{2}cg^{\prime}\left(c\right)+\frac{1}{2}cg^{\prime}\left(0\right)=-\frac{cg^{\prime}\left(c\right)}{2}<0,\label{eq:lem.c.1.1}
\end{equation}
where we have used the assumption that $g^{\prime}\left(0\right)=0$,
$g^{\prime}\left(x\right)>0$ for $x>0$, and $c>0$. (\ref{eq:lem.c.1.1})
implies that no rule in $\mathcal{D}_{S}^{3}$
can be Bayes optimal respect to prior $\pi(P_{1})=\pi(P_{2})=\frac{1}{2}$.
Applying analogous arguments shows that any rule in $\mathcal{D}_{S}^{2}$
is also not Bayes optimal with respect to the prior $\pi(P_{1})=\pi(P_{2})=\frac{1}{2}$. 
\end{proof}

\begin{lem}\label{lem:admit.2}
Under conditions of Theorem \ref{thm:inadmissibility}, for every singleton rule $\hat{\delta}_{t}$ considered in \eqref{eq:singelton.exponential.family}, there exists a fractional rule $\hat{\delta}_{t,\lambda}:=(1-\lambda)\hat{\delta}_{t}+\lambda(1-\hat{\delta}_{t})$
for some $0<\lambda<1$, such that $R_{g}(\hat{\delta}_{t},\tau)\geq R_{g}(\hat{\delta}_{t,\lambda},\tau)$
for all $\tau\in[-\bar{\tau},\bar{\tau}]$ with the inequality strict
for $\tau\neq0$.
\end{lem}

\begin{proof}
When $\tau=0$, it must hold that $R_{g}(\hat{\delta}_{t,\lambda},\tau)=R_{g}(\hat{\delta}_{t},\tau)$
for all $0<\lambda<1$. Thus, it suffices to show that there exists
some $\hat{\delta}_{t,\lambda}$, where $0<\lambda<1$, such
that $R_{g}(\hat{\delta}_{t},\tau)>R_{g}(\hat{\delta}_{t,\lambda},\tau)$
for each $\tau\in[-\bar{\tau},\bar{\tau}]$ not equal to zero. Emulating
the proof for Lemma \ref{lem:main}, for each $0<\lambda<1$, we can write 
\begin{align*}
R_{g}(\hat{\delta}_{t,\lambda},\tau) & =R_{g}(\hat{\delta}_{t},\tau)+\frac{\partial R_{g}(\hat{\delta}_{t,\lambda},\tau)}{\partial\lambda}\mid_{\lambda=\tilde{\lambda}}\cdotp\lambda,
\end{align*}
where $\tilde{\lambda}:=\tilde{\lambda}(\lambda,\tau,t)$ is
a number between $0$ and $\lambda$ that depends on $\lambda$, $\tau$
and $t$, and 
\[
\begin{aligned}\frac{\partial R_{g}(\hat{\delta}_{t,\lambda},\tau)}{\partial\lambda}\mid_{\lambda=\tilde{\lambda}}= & \tau\left[g'((1-\tilde{\lambda})Reg(1)+\tilde{\lambda}Reg(0))Pr(\hat{\delta}_{t}=1)\right]\\
- & \tau\left[g'((1-\tilde{\lambda})Reg(0)+\tilde{\lambda}Reg(1))Pr(\hat{\delta}_{t}=0)\right].
\end{aligned}
\]
As $\lambda>0$, it suffices to show that there exists some $\hat{\delta}_{t,\lambda}$
such that 
\begin{equation}
\frac{\partial R_{g}(\hat{\delta}_{t,\lambda},\tau)}{\partial\lambda}\mid_{\lambda=\tilde{\lambda}}<0,\text{ for all }\tau\in[-\bar{\tau},0)\cup(0,\bar{\tau}].\label{eq:negative.derivative}
\end{equation}
By Lemma \ref{lem:admit.3}, there exists some constant $p^{+}(t,\bar{\tau})>0$
such that $Pr(\hat{\delta}_{t}=0)\geq p^{+}(t,\bar{\tau})$ for all $\tau\in[0,\bar{\tau}]$, and there exists some $p^{-}(t,\bar{\tau})>0$
such that $Pr(\hat{\delta}_{t}=1)\geq p^{-}(t,\bar{\tau})$
for all $\tau\in[-\bar{\tau},0]$. Let 
\[
\lambda^{*}=\frac{\left(\frac{\min\left\{ p^{+}(t,\bar{\tau}),p^{-}(t,\bar{\tau})\right\} }{1-\min\left\{ p^{+}(t,\bar{\tau}),p^{-}(t,\bar{\tau})\right\} }\right)^{\frac{1}{\alpha-1}}}{1+\left(\frac{\min\left\{ p^{+}(t,\bar{\tau}),p^{-}(t,\bar{\tau})\right\} }{1-\min\left\{ p^{+}(t,\bar{\tau}),p^{-}(t,\bar{\tau})\right\} }\right)^{\frac{1}{\alpha-1}}}>0.
\]
Below, we show that for any $0<\lambda<\lambda^{*}$, the corresponding
fractional rule $\hat{\delta}_{t,\lambda}=(1-\lambda)\hat{\delta}_{t}+\lambda(1-\hat{\delta}_{t})$
satisfies (\ref{eq:negative.derivative}), which completes the proof.

\emph{Case 1}: If $0<\tau\leq\bar{\tau}$, $Reg(1)=0$, $Reg(0)=\tau$. Therefore, for each $0<\lambda<1$,
\[
\begin{aligned}\frac{\partial R_{g}(\hat{\delta}_{t,\lambda},\tau)}{\partial\lambda}\mid_{\lambda=\tilde{\lambda}} & =\tau\left[g'(\tilde{\lambda}\tau)Pr(\hat{\delta}_{t}=1)-g'((1-\tilde{\lambda})\tau)Pr(\hat{\delta}_{t}=0)\right].\end{aligned}
\]
As $g$ is nonzero and homogeneous of degree $\alpha>1$, $g'$ is
nonzero as well and homogeneous of degree $\alpha-1>0$, and can be
written as $g'(x)=kx^{\alpha-1}$ for some $k>0$. Thus,
\begin{align*}
 & \frac{\partial R_{g}(\hat{\delta}_{t,\lambda},\tau)}{\partial\lambda}\mid_{\lambda=\tilde{\lambda}}\\
= & \tau\left[g'(\tilde{\lambda}\tau)Pr(\hat{\delta}_{t}=1)-g'((1-\tilde{\lambda})\tau)Pr(\hat{\delta}_{t}=0)\right]\\
= & k\tau^{\alpha}\left[\tilde{\lambda}^{\alpha-1}Pr(\hat{\delta}_{t}=1)-(1-\tilde{\lambda})^{\alpha-1}Pr(\hat{\delta}_{t}=0)\right]\\
= & k\tau^{\alpha}\left[\tilde{\lambda}^{\alpha-1}-(\tilde{\lambda}^{\alpha-1}+(1-\tilde{\lambda})^{\alpha-1})Pr(\hat{\delta}_{t}=0)\right]\\
< & 0
\end{align*}
if 
\begin{equation}
Pr(\hat{\delta}_{t}=0)>\frac{\tilde{\lambda}^{\alpha-1}}{\tilde{\lambda}^{\alpha-1}+(1-\tilde{\lambda})^{\alpha-1}}.\label{eq:sign.1}
\end{equation}
As 
\[
\begin{aligned}Pr(\hat{\delta}_{t}=0) & \geq p^{+}(t,\bar{\tau})\geq\min\left\{ p^{+}(t,\bar{\tau}),p^{-}(t,\bar{\tau})\right\}>0 \end{aligned}
\]
and $\tilde{\lambda}\leq\lambda$, we deduce that by choosing $0<\lambda<\lambda^{*},$
(\ref{eq:sign.1}) holds true for all $0\leq\tilde{\lambda}\leq\lambda$,
implying (\ref{eq:negative.derivative}) is true for all $0<\tau\leq\bar{\tau}$.

\emph{Case 2}: If $-\bar{\tau}\leq\tau<0$, $Reg(1)=-\tau$, $Reg(0)=0$.
Similar calculations yield that for all  $0<\lambda<\lambda^{*}$,
\[
\begin{aligned}\frac{\partial R_{g}(\hat{\delta}_{t,\lambda},\tau)}{\partial\lambda}\mid_{\lambda=\tilde{\lambda}} & =\tau\left[g'(-(1-\tilde{\lambda})\tau)Pr(\hat{\delta}_{t}=1)\right]-\tau\left[g'(-\tilde{\lambda}\tau)Pr(\hat{\delta}_{t}=0)\right]\\
 & =-\tau\left[g'(-\tilde{\lambda}\tau)Pr(\hat{\delta}_{t}=0)-g'(-(1-\tilde{\lambda})\tau)Pr(\hat{\delta}_{t}=1)\right]\\
 & =k\left(-\tau\right)^{\alpha}\left[\tilde{\lambda}^{\alpha-1}-(\tilde{\lambda}^{\alpha-1}+(1-\tilde{\lambda})^{\alpha-1})Pr(\hat{\delta}_{t}=1)\right]\\
 & <0.
\end{aligned}
\]
This concludes the verification
of (\ref{eq:negative.derivative}).    
\end{proof}

\begin{lem}\label{lem:admit.3}
Under the conditions of Theorem \ref{thm:inadmissibility}, there exists some constant $p^{+}(t,\bar{\tau})>0$
such that $Pr(\hat{\delta}_{t}=0)\geq p^{+}(t,\bar{\tau})$ for all $\tau\in[0,\bar{\tau}]$, and there exists some $p^{-}(t,\bar{\tau})>0$
such that $Pr(\hat{\delta}_{t}=1)\geq p^{-}(t,\bar{\tau})$
for all $\tau\in[-\bar{\tau},0]$.
\end{lem}
\begin{proof}
We focus on the case when $\hat{\tau}$ is a continuous random variable.
For discrete case, the argument still holds by replacing the integration
with suitable summation. As the distribution $f(x|\tau)$ satisfies the 
monotone likelihood ratio property, it follows by \citet[][Theorem 4(i), p. 527]{berger1985statistical} that $Pr(\hat{\delta}_{t}=0)$ in non-increasing
in $\tau$. Therefore,
\begin{align*}
  p^{+}(t,\bar{\tau}):=& \min_{0\leq\tau\leq\bar{\tau}}Pr(\hat{\delta}_{t}=0)\\
= & \min_{0\leq\tau\leq\bar{\tau}}A(\tau)\int_{-\infty}^{t}\exp^{x\tau}h(x)dx\\
= & A(\bar{\tau})\int_{-\infty}^{t}\exp^{x\bar{\tau}}h(x)dx>0,
\end{align*}
as $t$ is in the interior of the support of $\hat{\tau}$.  And we conclude $\begin{aligned}Pr(\hat{\delta}_{t}=0) & \geq p^{+}(t,\bar{\tau})>0\end{aligned}
$ when $\tau\in[0,\bar{\tau}]$. Analogous arguments yield $Pr(\hat{\delta}_{t}=1)\geq p^{-}(t,\bar{\tau})>0$
when $\tau\in[-\bar{\tau},0]$, where 
\begin{align*}
p^{-}(t,\bar{\tau}) & :=\min_{-\bar{\tau}\leq\tau\leq0}Pr(\hat{\delta}_{t}=1)=A(-\bar{\tau})\int_{t}^{\infty}\exp^{-\bar{\tau}x}h(x)dx.
\end{align*}
\end{proof}

\section{Lemmas for Section \ref{sec:finite}}
\begin{lem}
\label{lem:1} $\tau=0$ is never a solution of (\ref{pf:main.1}).
\end{lem}
\begin{proof}
Note at $\tau=0$, the squared regret is 0. Suppose it is a solution
of (\ref{pf:main.1}), then it must hold that 
\begin{equation}
\mathbb{E}\left[1-\hat{\delta}^{*}(\bar{Y}_{1})\right]^{2}=0\text{ and }\mathbb{E}\left[\hat{\delta}^{*}(\bar{Y}_{1})\right]^{2}=0.\label{pf:lem1.1}
\end{equation}
Since $\hat{\delta}^{*}(\bar{y}_{1})\in[0,1]$ for all $\bar{y}_{1}$,
(\ref{pf:lem1.1}) implies 
\[
1-\hat{\delta}^{*}(\bar{y}_{1})=0\text{ and }\hat{\delta}^{*}(\bar{y}_{1})=0,\text{ for all }\bar{y}_{1}\text{ a.s.},
\]
which cannot be true. This implies $\tau=0$ is never a solution of (\ref{pf:main.1}).
\end{proof}
\begin{lem}
\label{lem:2}The solution of (\ref{pf:main.1}) is symmetric, i.e., if
some $\tau \in(0,\infty)$ solves (\ref{pf:main.1}), then it also holds that $-\tau$
solves (\ref{pf:main.1}).
\end{lem}
\begin{proof}
Suppose $\tau \in(0,\infty)$ solves (\ref{pf:main.1}) but $-\tau$ does not. Note the mean square regret of $\hat{\delta}^{*}$
at $\tau$ is
\begin{align*}
R_{sq}(\hat{\delta}^{*},P_{\tau}) & =\tau^{2}\mathbb{E}\left[1-\hat{\delta}^{*}(\bar{Y}_{1})\right]^{2}\\
 & =\tau^{2}\int\left[1-\hat{\delta}^{*}(\bar{y}_{1})\right]^{2}\sqrt{\frac{1}{2\pi}}\exp\left(-\frac{1}{2}\left[\left(\bar{y}_{1}-\tau\right)^{2}\right]\right)d\bar{y}_{1},
\end{align*}
while the mean square regret at $-\tau$ is 
\begin{align*}
R_{sq}(\hat{\delta}^{*},P_{-\tau})= & \tau^{2}\mathbb{E}\left[\hat{\delta}^{*}(\bar{Y}_{1})\right]^{2}\\
= & \tau^{2}\int\left[\hat{\delta}^{*}(\bar{y}_{1})\right]^{2}\sqrt{\frac{1}{2\pi}}\exp\left(-\frac{1}{2}\left[\left(\bar{y}_{1}+\tau\right)^{2}\right]\right)d\bar{y}_{1}\\
= & \tau^{2}\int\left[\hat{\delta}^{*}(-\tilde{y}_{1})\right]^{2}\sqrt{\frac{1}{2\pi}}\exp\left(-\frac{1}{2}\left[\left(\tau-\tilde{y}_{1}\right)^{2}\right]\right)d\tilde{y}_{1}\\
= & \tau^{2}\int\left[\hat{\delta}^{*}(-\bar{y}_{1})\right]^{2}\sqrt{\frac{1}{2\pi}}\exp\left(-\frac{1}{2}\left[\left(\tau-\bar{y}_{1}\right)^{2}\right]\right)d\bar{y}_{1},
\end{align*}
where the third equality uses the change of variable $\tilde{y}_{1}=-\bar{y}_{1}$,
and the fourth equality changes the variable of integration from $\tilde{y}_{1}$
to $\bar{y}_{1}$.

If $\tau$ solves (\ref{pf:main.1}) but $-\tau$ does not, then there must exist some $\bar{y}_{1}\in\mathbb{R}$ such that $1-\hat{\delta}^{*}(\bar{y}_{1})\neq\hat{\delta}^{*}(-\bar{y}_{1})$.
Let
\[
\mathbf{S}=\left\{ \bar{y}_{1}\in\mathbb{R}:1-\hat{\delta}^{*}(\bar{y}_{1})\neq\hat{\delta}^{*}(-\bar{y}_{1})\right\} 
\]
be the collection of all $\bar{y}_{1}$ such that $1-\hat{\delta}^{*}(\bar{y}_{1})\neq\hat{\delta}^{*}(-\bar{y}_{1})$.\footnote{Notice the set $\mathbf{S}$ must be symmetric, that is, if 
\[
1-\hat{\delta}^{*}(\bar{y}_{1})\neq\hat{\delta}^{*}(-\bar{y}_{1})
\]
holds then
\[
1-\hat{\delta}^{*}(-\bar{y}_{1})\neq\hat{\delta}^{*}(\bar{y}_{1})
\]
also holds.
} The contribution of the elements of $\mathbf{S}$ to the mean square
regret at $\tau$ is 
\[
\tau^{2}\int_{\mathbf{S}}\left[1-\hat{\delta}^{*}(\bar{y}_{1})\right]^{2}\sqrt{\frac{1}{2\pi}}\exp\left(-\frac{1}{2}\left[\left(\bar{y}_{1}-\tau\right)^{2}\right]\right)d\bar{y}_{1}
\]
while the contribution of the elements in $\mathbf{S}$ to the mean square
regret at $-\tau$ is 
\[
\tau^{2}\int_{\mathbf{S}}\left[\hat{\delta}^{*}(-\bar{y}_{1})\right]^{2}\sqrt{\frac{1}{2\pi}}\exp\left(-\frac{1}{2}\left[\left(\tau-\bar{y}_{1}\right)^{2}\right]\right)d\bar{y}_{1}.
\]
Since $\tau$ solves (\ref{pf:main.1}) but not $-\tau$, it holds that 
\begin{align}
 & \tau^{2}\int_{\mathbf{S}}\left[1-\hat{\delta}^{*}(\bar{y}_{1})\right]^{2}\sqrt{\frac{1}{2\pi}}\exp\left(-\frac{1}{2}\left[\left(\bar{y}_{1}-\tau\right)^{2}\right]\right)d\bar{y}_{1}\nonumber \\
> & \tau^{2}\int_{\mathbf{S}}\left[\hat{\delta}^{*}(-\bar{y}_{1})\right]^{2}\sqrt{\frac{1}{2\pi}}\exp\left(-\frac{1}{2}\left[\left(\tau-\bar{y}_{1}\right)^{2}\right]\right)d\bar{y}_{1}.\label{pf:lem2.1}
\end{align}
If (\ref{pf:lem2.1}) holds though, we can strictly reduce
the mean square regret for $\tau$ by switching to an alternative policy
$\bar{\delta}$, where 
\[
\bar{\delta}(\bar{y}_{1})=\begin{cases}
\hat{\delta}^{*}(\bar{y}_{1}) & \text{if }\bar{y}_{1}\notin\mathbf{S},\\
1-\hat{\delta}^{*}(-\bar{y}_{1}) & \text{if }\bar{y}_{1}\in\mathbf{S}.
\end{cases}
\]
This contradicts the assumption that $\hat{\delta}^{*}$ is a minimax optimal rule, i.e.,
\[
R_{sq}(\hat{\delta}^{*},P_{\tau})=\inf_{\hat{\delta}\in\mathcal{D}}R_{sq}(\hat{\delta},P_{\tau}).
\]\end{proof}

\begin{lem}
\label{lem:3} A least favorable prior distribution $\pi^{*}$ is such that
\[\pi^{*}\left(\tau\right)=\frac{1}{2},\pi^{*}\left(-\tau\right)=\frac{1}{2},
\]
for some $\tau\in(0,\infty)$.
\end{lem}
\begin{proof}
For each $\tau\in(0,\infty)$, consider the symmetric prior 
\begin{equation}
\pi^{*}(\tau)=p_{\tau},\pi^{*}(-\tau)=1-p_{\tau},\text{ where }p_{\tau}\in[0,1].\label{pf:lem3.prior}
\end{equation}
If (\ref{pf:lem3.prior}) is indeed the least favorable prior, then
$\hat{\delta}^{*}(\bar{y}_{1})=\frac{(1-p_{\tau})f\{\bar{y}_{1}|-\tau\}}{p_{\tau}f\{\bar{y}_{1}|\tau\}+(1-p_{\tau})f\{\bar{y}_{1}|-\tau\}}$,
and the mean square regret of $\hat{\delta}^{*}$ at $P_{\tau}$
is 
\begin{align}
R_{sq}(\hat{\delta}^{*},P_{\tau}) & =\tau^{2}\int\frac{(1-p_{\tau})^{2}f^{2}\{\bar{y}_{1}|-\tau\}f\{\bar{y}_{1}|\tau\}}{\left[p_{\tau}f\{\bar{y}_{1}|\tau\}+(1-p_{\tau})f\{\bar{y}_{1}|-\tau\}\right]^{2}}d\bar{y}_{1}.\label{pf:2}
\end{align}
The mean square regret of $\hat{\delta}^{*}$ at $P_{-\tau}$ is 
\begin{align}
R_{sq}(\hat{\delta}^{*},P_{-\tau}) & =\tau^{2}\int\frac{p_{\tau}^{2}f^{2}\{\bar{y}_{1}|\tau\}f\{\bar{y}_{1}|-\tau\}}{\left[p_{\tau}f\{\bar{y}_{1}|\tau\}+(1-p_{\tau})f\{\bar{y}_{1}|-\tau\}\right]^{2}}d\bar{y}_{1}.\label{pf:3}
\end{align}
By Lemma \ref{lem:2}, $\tau$ and $-\tau$ yield the same mean square
regret at $\hat{\delta}^{*}$, so $p_{\tau}$ must be such that 
\[
(\ref{pf:2})=(\ref{pf:3}).
\]
For each $\bar{y}_{1}$, the numerator of the integrand in (\ref{pf:2})
is 
\begin{align*}
 & (1-p_{\tau})^{2}f^{2}\{\bar{y}_{1}|-\tau\}f\{\bar{y}_{1}|\tau\}\\
= & (1-p_{\tau})^{2}\left(\frac{1}{2\pi}\right)^{\frac{3}{2}}\exp\left(-\left[\left(\bar{y}_{1}+\tau\right)^{2}\right]-\frac{1}{2}\left[\left(\bar{y}_{1}-\tau\right)^{2}\right]\right)
\end{align*}
while for each $\bar{y}_{1}$, the numerator of the integrand in (\ref{pf:3})
is 
\[
p_{\tau}^{2}\left(\frac{1}{2\pi}\right)^{\frac{3}{2}}\exp\left(-\left[\left(\bar{y}_{1}-\tau\right)^{2}\right]-\frac{1}{2}\left[\left(\bar{y}_{1}+\tau\right)^{2}\right]\right).
\]
Therefore, (\ref{pf:3}) can be written as
\begin{align}
 & R_{sq}(\hat{\delta}^{*},P_{-\tau})\nonumber \\
= & \tau^{2}\int\frac{p_{\tau}^{2}\left(\frac{1}{2\pi}\right)^{\frac{3}{2}}\exp\left(-\left[\left(\bar{y}_{1}-\tau\right)^{2}\right]-\frac{1}{2}\left[\left(\bar{y}_{1}+\tau\right)^{2}\right]\right)}{\left[p_{\tau}f\{\bar{y}_{1}|\tau\}+(1-p_{\tau})f\{\bar{y}_{1}|-\tau\}\right]^{2}}d\bar{y}_{1}\nonumber \\
= & \tau^{2}\int\frac{p_{\tau}^{2}\left(\frac{1}{2\pi}\right)^{\frac{3}{2}}\exp\left(-\left[\left(\bar{y}_{1}+\tau\right)^{2}\right]-\frac{1}{2}\left[\left(\bar{y}_{1}-\tau\right)^{2}\right]\right)}{\left[p_{\tau}f\{-\bar{y}_{1}|\tau\}+(1-p_{\tau})f\{-\bar{y}_{1}|-\tau\}\right]^{2}}d\bar{y}_{1},\label{eq:4}
\end{align}
where the second equality follows from a change of variable. (\ref{pf:2})
admits 
\begin{equation}
R_{sq}(\hat{\delta}^{*},P_{\tau})=\tau^{2}\int\frac{(1-p_{\tau})^{2}\left(\frac{1}{2\pi}\right)^{\frac{3}{2}}\exp\left(-\left[\left(\bar{y}_{1}+\tau\right)^{2}\right]-\frac{1}{2}\left[\left(\bar{y}_{1}-\tau\right)^{2}\right]\right)}{\left[p_{\tau}f\{\bar{y}_{1}|\tau\}+(1-p_{\tau})f\{\bar{y}_{1}|-\tau\}\right]^{2}}d\bar{y}_{1}. \label{eq:5}
\end{equation}
Hence, $p_{\tau}$ must be such that 
\[
(\ref{eq:4})=(\ref{eq:5}),
\]
which is satisfied if $p_{\tau}=\frac{1}{2}$. Indeed, when $p_{\tau}=\frac{1}{2}$,
(\ref{eq:4}) and (\ref{eq:5}) only differ in their denominators.
Furthermore, for (\ref{eq:5}), the denominator of the integrand can be written as 
\begin{align*}
 & \left[\frac{1}{2}f\{\bar{y}_{1}|\tau\}+\frac{1}{2}f\{\bar{y}_{1}|-\tau\}\right]^{2}\\
= & \left[\frac{1}{2}\sqrt{\frac{1}{2\pi}}\exp\left(-\frac{1}{2}\left(\bar{y}_{1}-\tau\right)^{2}\right)+\frac{1}{2}\sqrt{\frac{1}{2\pi}}\exp\left(-\frac{1}{2}\left(\bar{y}_{1}+\tau\right)^{2}\right)\right]^{2}
\end{align*}
while for (\ref{eq:4}), the corresponding restatement is
\begin{align*}
 & \left[\frac{1}{2}f\{-\bar{y}_{1}|\tau\}+\frac{1}{2}f\{-\bar{y}_{1}|-\tau\}\right]^{2}\\
= & \left[\frac{1}{2}\sqrt{\frac{1}{2\pi}}\exp\left(-\frac{1}{2}\left(-\bar{y}_{1}-\tau\right)^{2}\right)+\frac{1}{2}\sqrt{\frac{1}{2\pi}}\exp\left(-\frac{1}{2}\left(-\bar{y}_{1}+\tau\right)^{2}\right)\right]^{2}\\
= & \left[\frac{1}{2}\sqrt{\frac{1}{2\pi}}\exp\left(-\frac{1}{2}\left(\bar{y}_{1}+\tau\right)^{2}\right)+\frac{1}{2}\sqrt{\frac{1}{2\pi}}\exp\left(-\frac{1}{2}\left(\bar{y}_{1}-\tau\right)^{2}\right)\right]^{2}
\end{align*}
which is equivalent. 
\end{proof}
\begin{lem}
\label{lem:4} If
\[
\tau^{*}\in\arg\sup\limits _{\tau\in[0,\infty)}\tau^{2}\int\left(\frac{f\{\bar{y}_{1}|-\tau\}}{f\{\bar{y}_{1}|\tau\}+f\{\bar{y}_{1}|-\tau\}}\right)^{2}f\{\bar{y}_{1}|\tau\}d\bar{y}_{1},
\]
then $\left(\frac{\partial}{\partial\tau}R_{sq}(\hat{\delta}^{*},P_{\tau})\right)|_{\tau=\tau^{*}}=0$. 
\end{lem}
\begin{proof}
Since $\tau^{*}\in\arg\sup\limits _{\tau\in[0,\infty)}\tau^{2}\int\left(\frac{f\{\bar{y}_{1}|-\tau\}}{f\{\bar{y}_{1}|\tau\}+f\{\bar{y}_{1}|-\tau\}}\right)^{2}f\{\bar{y}_{1}|\tau\}d\bar{y}_{1}$ and the objective function is
continuously differentiable, it holds that 
\begin{equation}
\left[\frac{\partial}{\partial\tau}\left(\tau^{2}\int\left(\frac{f\{\bar{y}_{1}|-\tau\}}{f\{\bar{y}_{1}|\tau\}+f\{\bar{y}_{1}|-\tau\}}\right)^{2}f\{\bar{y}_{1}|\tau\}d\bar{y}_{1}\right)\right]\mid_{\tau=\tau^{*}}=0.\label{pf:lem.4.1}
\end{equation}
On the other hand, 
\begin{align}
 & \left(\frac{\partial}{\partial\tau}R_{sq}(\hat{\delta}^{*},P_{\tau})\right)\nonumber \\
= & \frac{\partial}{\partial\tau}\left(\tau^{2}\int\left[\frac{f\{\bar{y}_{1}|-\tau^{*}\}}{f\{\bar{y}_{1}|\tau^{*}\}+f\{\bar{y}_{1}|-\tau^{*}\}}\right]^{2}f(\bar{y}_{1}|\tau)d\bar{y}_{1}\right).\label{pf:lem.4.2}
\end{align}
Observing the objective function in (\ref{pf:lem.4.1}) and (\ref{pf:lem.4.2}),
$\left(\frac{\partial}{\partial\tau}R_{sq}(\hat{\delta}^{*},P_{\tau})\right)|_{\tau=\tau^{*}}=0$
holds if 
\begin{equation}
\left[\frac{\partial}{\partial\tau}\left((\tau^{*})^{2}\int\left(\frac{f\{\bar{y}_{1}|-\tau\}}{f\{\bar{y}_{1}|\tau\}+f\{\bar{y}_{1}|-\tau\}}\right)^{2}f\{\bar{y}_{1}|\tau^{*}\}d\bar{y}_{1}\right)\right]\mid_{\tau=\tau^{*}}=0.\label{pf:lem.4.3}
\end{equation}
In what follows, we verify that (\ref{pf:lem.4.3}) indeed holds.
Note 
\begin{align*}
 & \frac{\partial}{\partial\tau}\left((\tau^{*})^{2}\int\left(\frac{f\{\bar{y}_{1}|-\tau\}}{f\{\bar{y}_{1}|\tau\}+f\{\bar{y}_{1}|-\tau\}}\right)^{2}f\{\bar{y}_{1}|\tau^{*}\}d\bar{y}_{1}\right)\\
= & (\tau^{*})^{2}\int2\left(\frac{f\{\bar{y}_{1}|-\tau\}}{f\{\bar{y}_{1}|\tau\}+f\{\bar{y}_{1}|-\tau\}}\right)\frac{\partial}{\partial\tau}\left(\frac{f\{\bar{y}_{1}|-\tau\}}{f\{\bar{y}_{1}|\tau\}+f\{\bar{y}_{1}|-\tau\}}\right)f\{\bar{y}_{1}|\tau^{*}\}d\bar{y}_{1}\\
= & 2(\tau^{*})^{2}\int\left(\frac{f\{\bar{y}_{1}|-\tau\}f\{\bar{y}_{1}|\tau^{*}\}}{\left[f\{\bar{y}_{1}|\tau\}+f\{\bar{y}_{1}|-\tau\}\right]^{3}}\right)\left(\frac{\partial f\{\bar{y}_{1}|-\tau\}}{\partial\tau}f\{\bar{y}_{1}|\tau\}-f\{\bar{y}_{1}|-\tau\}\frac{\partial f\{\bar{y}_{1}|\tau\}}{\partial\tau}\right)d\bar{y}_{1}.
\end{align*}
Since $f\{\bar{y}_{1}|\tau\}=\sqrt{\frac{1}{2\pi}}\exp\left(-\frac{1}{2}\left[\left(\bar{y}_{1}-\tau\right)^{2}\right]\right)$ and
$f\{\bar{y}_{1}|-\tau\}=\sqrt{\frac{1}{2\pi}}\exp\left(-\frac{1}{2}\left[\left(\bar{y}_{1}+\tau\right)^{2}\right]\right)$, we have that
\begin{align*}
\frac{\partial f\{\bar{y}_{1}|\tau\}}{\partial\tau}  =f\{\bar{y}_{1}|\tau\}(\bar{y}_{1}-\tau)
\ \ \text{ and } \ \
\frac{\partial f\{\bar{y}_{1}|-\tau\}}{\partial\tau}  =-f\{\bar{y}_{1}|-\tau\}(\bar{y}_{1}+\tau).
\end{align*}
It follows that 
\begin{align}
 & \frac{\partial}{\partial\tau}\left((\tau^{*})^{2}\int\left(\frac{f\{\bar{y}_{1}|-\tau\}}{f\{\bar{y}_{1}|\tau\}+f\{\bar{y}_{1}|-\tau\}}\right)^{2}f\{\bar{y}_{1}|\tau^{*}\}d\bar{y}_{1}\right)\nonumber \\
= & -4(\tau^{*})^{2}\int\left(\frac{f\{\bar{y}_{1}|-\tau\}f\{\bar{y}_{1}|\tau^{*}\}f\{\bar{y}_{1}|-\tau\}f\{\bar{y}_{1}|\tau\}}{\left[f\{\bar{y}_{1}|\tau\}+f\{\bar{y}_{1}|-\tau\}\right]^{3}}\right)\bar{y}_{1}d\bar{y}_{1}.\label{pf:lem.4.4}
\end{align}
Evaluating (\ref{pf:lem.4.4}) at $\tau=\tau^{*}$ yields 
\begin{align}
 & \left[\frac{\partial}{\partial\tau}\left((\tau^{*})^{2}\int\left(\frac{f\{\bar{y}_{1}|-\tau\}}{f\{\bar{y}_{1}|\tau\}+f\{\bar{y}_{1}|-\tau\}}\right)^{2}f\{\bar{y}_{1}|\tau^{*}\}d\bar{y}_{1}\right)\right]\mid_{\tau=\tau^{*}}\nonumber \\
= & -4(\tau^{*})^{2}\int\left[\frac{\left(f\{\bar{y}_{1}|-\tau^{*}\}f\{\bar{y}_{1}|\tau^{*}\}\right)^{2}}{\left[f\{\bar{y}_{1}|\tau^{*}\}+f\{\bar{y}_{1}|-\tau^{*}\}\right]^{3}}\right]\bar{y}_{1}d\bar{y}_{1}\nonumber \\
= & -4(\tau^{*})^{2}\int w(\bar{y}_{1})\bar{y}_{1}d\bar{y}_{1},\label{pf:lem.4.5}
\end{align}
where $w(\bar{y}_{1})=\frac{\left(f\{\bar{y}_{1}|-\tau^{*}\}f\{\bar{y}_{1}|\tau^{*}\}\right)^{2}}{\left[f\{\bar{y}_{1}|\tau^{*}\}+f\{\bar{y}_{1}|-\tau^{*}\}\right]^{3}}$.
However, notice for each $\bar{y}_{1}\in\mathbb{R}$: 
\[
w(\bar{y}_{1})=\frac{\left(\sqrt{\frac{1}{2\pi}}\exp\left(-\frac{1}{2}\left[\left(\bar{y}_{1}+\tau^{*}\right)^{2}\right]\right)\sqrt{\frac{1}{2\pi}}\exp\left(-\frac{1}{2}\left[\left(\bar{y}_{1}-\tau^{*}\right)^{2}\right]\right)\right)^{2}}{\left[\sqrt{\frac{1}{2\pi}}\exp\left(-\frac{1}{2}\left[\left(\bar{y}_{1}-\tau^{*}\right)^{2}\right]\right)+\sqrt{\frac{1}{2\pi}}\exp\left(-\frac{1}{2}\left[\left(\bar{y}_{1}+\tau^{*}\right)^{2}\right]\right)\right]^{3}}
\]
while 
\begin{align*}
w(-\bar{y}_{1}) & =\frac{\left(\sqrt{\frac{1}{2\pi}}\exp\left(-\frac{1}{2}\left[\left(-\bar{y}_{1}+\tau^{*}\right)^{2}\right]\right)\sqrt{\frac{1}{2\pi}}\exp\left(-\frac{1}{2}\left[\left(-\bar{y}_{1}-\tau^{*}\right)^{2}\right]\right)\right)^{2}}{\left[\sqrt{\frac{1}{2\pi}}\exp\left(-\frac{1}{2}\left[\left(-\bar{y}_{1}-\tau^{*}\right)^{2}\right]\right)+\sqrt{\frac{1}{2\pi}}\exp\left(-\frac{1}{2}\left[\left(-\bar{y}_{1}+\tau^{*}\right)^{2}\right]\right)\right]^{3}}\\
 & =\frac{\left(\sqrt{\frac{1}{2\pi}}\exp\left(-\frac{1}{2}\left[\left(\bar{y}_{1}-\tau^{*}\right)^{2}\right]\right)\sqrt{\frac{1}{2\pi}}\exp\left(-\frac{1}{2}\left[\left(\bar{y}_{1}+\tau^{*}\right)^{2}\right]\right)\right)^{2}}{\left[\sqrt{\frac{1}{2\pi}}\exp\left(-\frac{1}{2}\left[\left(\bar{y}_{1}+\tau^{*}\right)^{2}\right]\right)+\sqrt{\frac{1}{2\pi}}\exp\left(-\frac{1}{2}\left[\left(\bar{y}_{1}-\tau^{*}\right)^{2}\right]\right)\right]^{3}}\\
 & =\frac{\left(f\{\bar{y}_{1}|\tau^{*}\}f\{\bar{y}_{1}|-\tau^{*}\}\right)^{2}}{\left[f\{\bar{y}_{1}|-\tau^{*}\}+f\{\bar{y}_{1}|\tau^{*}\}\right]^{3}}\\
 & =w(\bar{y}_{1}).
\end{align*}
Therefore 
\[
(\ref{pf:lem.4.5})=-4(\tau^{*})^{2}\int w(\bar{y}_{1})\bar{y}_{1}d\bar{y}_{1}=0
\]
and the conclusion of the lemma follows.
\end{proof}

\begin{lem}
\label{lem:5} $\tau^{*}$ is the unique solution of $\sup\limits _{\tau\in[0,\infty)}R_{sq}(\hat{\delta}^{*},P_{\tau})$. 
\end{lem}
\begin{proof}
Write $\omega^{*}(\bar{y}_{1}):=1-\hat{\delta}^{*}(\bar{y}_{1})$.
We evaluate the first derivative of 
\[
R_{sq}(\hat{\delta}^{*},P_{\tau})=\tau^{2}\int\left[\omega^{*}(\bar{y}_{1})\right]^{2}f\{\bar{y}_{1}|\tau\}d\bar{y}_{1}
\]
as a function of $\tau\in[0,\infty)$. Notice for each $\bar{y}_{1}\in\mathbb{R}$
and each $\tau\in[0,\infty)$, 
\[
\frac{\partial}{\partial\tau}f\{\bar{y}_{1}|\tau\}=f\{\bar{y}_{1}|\tau\}\left(\bar{y}-\tau\right)=-\frac{\partial}{\partial\bar{y}_{1}}f\{\bar{y}_{1}|\tau\}.
\]
Therefore, using integration by parts twice yields 
\begin{align*}
R_{sq}^{(1)}(\tau)  &:=\frac{\partial}{\partial\tau}R_{sq}(\hat{\delta}^{*},P_{\tau})\\
 & =2\tau\int\left[\omega^{*}(\bar{y}_{1})\right]^{2}f(\bar{y}_{1}|\tau)d\bar{y}_{1} 
 +\tau^{2}\int\left[\omega^{*}(\bar{y}_{1})\right]^{2}\frac{\partial}{\partial\tau}\left(f(\bar{y}_{1}|\tau)\right)d\bar{y}_{1}\\
 & =2\tau\int\left[\omega^{*}(\bar{y}_{1})\right]^{2}f(\bar{y}_{1}|\tau)d\bar{y}_{1} 
 -\tau^{2}\int\left[\omega^{*}(\bar{y}_{1})\right]^{2}\frac{\partial}{\partial\bar{y}_{1}}f(\bar{y}_{1}|\tau)d\bar{y}_{1}\\
 & =2\tau\int\left[\omega^{*}(\bar{y}_{1})\right]^{2}f(\bar{y}_{1}|\tau)d\bar{y}_{1} 
 -\tau^{2}\int\left[\omega^{*}(\bar{y}_{1})\right]^{2}df(\bar{y}_{1}|\tau)\\
 & =2\tau\int\left[\omega^{*}(\bar{y}_{1})\right]^{2}f(\bar{y}_{1}|\tau)d\bar{y}_{1}+2\tau^{2}\left(\int\omega^{*}(\bar{y}_{1})\frac{\partial\left(\omega^{*}(\bar{y}_{1})\right)}{\partial\bar{y}_{1}}f(\bar{y}_{1}|\tau)d\bar{y}_{1}\right)\\
 & =2\tau\left\{ \int\left[\omega^{*}(\bar{y}_{1})\right]^{2}f(\bar{y}_{1}|\tau)d\bar{y}_{1}+\int\omega^{*}(\bar{y}_{1})\frac{\partial\left(\omega^{*}(\bar{y}_{1})\right)}{\partial\bar{y}_{1}}\tau f(\bar{y}_{1}|\tau)d\bar{y}_{1}\right\} \\
 & =2\tau\bigg\{ \int\left[\omega^{*}(\bar{y}_{1})\right]^{2}f(\bar{y}_{1}|\tau)d\bar{y}_{1}+\int\omega^{*}(\bar{y}_{1})\frac{\partial\left(\omega^{*}(\bar{y}_{1})\right)}{\partial\bar{y}_{1}}\left(\tau-\bar{y}_{1}\right)f(\bar{y}_{1}|\tau)d\bar{y}_{1}\\
 &\;\;\;\;\;\;\;\;\;\;\;\; +\int\omega^{*}(\bar{y}_{1})\frac{\partial\left(\omega^{*}(\bar{y}_{1})\right)}{\partial\bar{y}_{1}}\bar{y}_{1}f(\bar{y}_{1}|\tau)d\bar{y}_{1}\bigg\} \\
 & =2\tau\bigg\{ \int\left[\omega^{*}(\bar{y}_{1})\right]^{2}f(\bar{y}_{1}|\tau)d\bar{y}_{1}+\int\omega^{*}(\bar{y}_{1})\frac{\partial\left(\omega^{*}(\bar{y}_{1})\right)}{\partial\bar{y}_{1}}df(\bar{y}_{1}|\tau)\\
 &\;\;\;\;\;\;\;\;\;\;\;\; +\int\omega^{*}(\bar{y}_{1})\frac{\partial\left(\omega^{*}(\bar{y}_{1})\right)}{\partial\bar{y}_{1}}\bar{y}_{1}f(\bar{y}_{1}|\tau)d\bar{y}_{1}\bigg\} \\
 & =2\tau\bigg\{ \int\left[\omega^{*}(\bar{y}_{1})\right]^{2}f(\bar{y}_{1}|\tau)d\bar{y}_{1}-\int\frac{\partial}{\partial\bar{y}_{1}}\left[\omega^{*}(\bar{y}_{1})\frac{\partial\left(\omega^{*}(\bar{y}_{1})\right)}{\partial\bar{y}_{1}}\right]f(\bar{y}_{1}|\tau)d\bar{y}_{1}\\
 &\;\;\;\;\;\;\;\;\;\;\;\; +\int\omega^{*}(\bar{y}_{1})\frac{\partial\left(\omega^{*}(\bar{y}_{1})\right)}{\partial\bar{y}_{1}}\bar{y}_{1}f(\bar{y}_{1}|\tau)d\bar{y}_{1}\bigg\} \\
 & =2\tau\bigg\{ \int\bigg\{ \left[\omega^{*}(\bar{y}_{1})\right]^{2}-\left(\frac{\partial\left(\omega^{*}(\bar{y}_{1})\right)}{\partial\bar{y}_{1}}\right)^{2}-\omega^{*}(\bar{y}_{1})\frac{\partial^{2}\left(\omega^{*}(\bar{y}_{1})\right)}{\partial\left(\bar{y}_{1}\right)^{2}} \\
 &\;\;\;\;\;\;\;\;\;\;\;\;\;\;\;\;\;\; +\omega^{*}(\bar{y}_{1})\frac{\partial\left(\omega^{*}(\bar{y}_{1})\right)}{\partial\bar{y}_{1}}\bar{y}_{1}\bigg\} f(\bar{y}_{1}|\tau)d\bar{y}_{1}\bigg\} .
\end{align*}
The sign of $R_{sq}^{(1)}(\tau)$ is determined by 
\[
\mathbf{R}(\tau):=\int\mathbf{w}(\bar{y}_{1})f\{\bar{y}_{1}|\tau\}d\bar{y}_{1},
\]
where 
\[
\mathbf{w}(\bar{y}_{1})=\left[\omega^{*}(\bar{y}_{1})\right]^{2}-\left(\frac{\partial\left(\omega^{*}(\bar{y}_{1})\right)}{\partial\bar{y}_{1}}\right)^{2}-\omega^{*}(\bar{y}_{1})\frac{\partial^{2}\left(\omega^{*}(\bar{y}_{1})\right)}{\partial\left(\bar{y}_{1}\right)^{2}}+\omega^{*}(\bar{y}_{1})\frac{\partial\left(\omega^{*}(\bar{y}_{1})\right)}{\partial\bar{y}_{1}}\bar{y}_{1}.
\]
We aim to show that $\mathbf{R}(\tau)$ has a unique sign change from
$+$ to $-$ at $\tau^{*}$, with the conclusion immediately following.

\emph{Step 1}: we show $\mathbf{R}(\tau)$ has at most one sign change from
$+$ to $-$. Notice $\omega^{*}(\bar{y}_{1})=\frac{1}{\exp\left(2\tau^{*}\bar{y}_{1}\right)+1}$.
Therefore, 
\begin{align*}
\frac{\partial\left(\omega^{*}(\bar{y}_{1})\right)}{\partial\bar{y}_{1}} & =-\left[\omega^{*}(\bar{y}_{1})\right]^{2}\exp\left(2\tau^{*}\bar{y}_{1}\right)2\tau^{*},\\
\frac{\partial^{2}\left(\omega^{*}(\bar{y}_{1})\right)}{\partial\left(\bar{y}_{1}\right)^{2}} & =2\left(\exp\left(2\tau^{*}\bar{y}_{1}\right)2\tau^{*}\right)^{2}\left[\omega^{*}(\bar{y}_{1})\right]^{3}-\left[\omega^{*}(\bar{y}_{1})\right]^{2}\exp\left(2\tau^{*}\bar{y}_{1}\right)\left(2\tau^{*}\right)^{2}
\end{align*}
Plugging in $\mathbf{w}(\bar{y}_{1})$ yields 
\begin{align*}
\mathbf{w}(\bar{y}_{1}) & =\left[\omega^{*}(\bar{y}_{1})\right]^{2}\bigg\{ 1-3\left(\omega^{*}(\bar{y}_{1})\exp\left(2\tau^{*}\bar{y}_{1}\right)2\tau^{*}\right)^{2}+\omega^{*}(\bar{y}_{1})\exp\left(2\tau^{*}\bar{y}_{1}\right)\left(2\tau^{*}\right)^{2}\\
&\;\;\;\;\;\;\;\;\;\;\;\;\;\;\;\;\;\;\;\; -\omega^{*}(\bar{y}_{1})\exp\left(2\tau^{*}\bar{y}_{1}\right)2\tau^{*}\bar{y}_{1}\bigg\} \\
 & =\left[\omega^{*}(\bar{y}_{1})\right]^{2}\hat{\delta}^{*}(\bar{y}_{1})\left\{\frac{1}{\hat{\delta}^{*}(\bar{y}_{1})}-3\left(2\tau^{*}\right)^{2}\hat{\delta}^{*}(\bar{y}_{1})+\left(2\tau^{*}\right)^{2}-2\tau^{*}\bar{y}_{1}\right\}.
\end{align*}
Since $\left[\omega^{*}(\bar{y}_{1})\right]^{2}>0$ and $\hat{\delta}^{*}(\bar{y}_{1})>0$ for all $\bar{y}_{1}$,
the sign of $\mathbf{w}(\bar{y}_{1})$ is determined by 
\begin{align*}
\tilde{\mathbf{w}}(\bar{y}_{1}) & :=\frac{1}{\hat{\delta}^{*}(\bar{y}_{1})}-3\left(2\tau^{*}\right)^{2}\hat{\delta}^{*}(\bar{y}_{1})+\left(2\tau^{*}\right)^{2}-2\tau^{*}\bar{y}_{1},
\end{align*}
and $\tilde{\mathbf{w}}(\bar{y}_{1})=0$
if and only if 
\begin{equation}
\frac{1}{\hat{\delta}^{*}(\bar{y}_{1})}-3\left(2\tau^{*}\right)^{2}\hat{\delta}^{*}(\bar{y}_{1})+\left(2\tau^{*}\right)^{2}=2\tau^{*}\bar{y}_{1}.\label{pf:lem.5.1}
\end{equation}
Moreover, it is straightforward to check that $\frac{\partial}{\partial\bar{y}_{1}}\hat{\delta}^{*}(\bar{y}_{1})>0$.
It follows the first derivative of the left hand side (LHS) of (\ref{pf:lem.5.1}) is
\[
\frac{\partial\text{LHS}}{\partial\bar{y}_{1}}=\left(-\frac{1}{\left(\hat{\delta}^{*}(\bar{y}_{1})\right)^{2}}-3\left(2\tau^{*}\right)^{2}\right)\frac{\partial}{\partial\bar{y}_{1}}\hat{\delta}^{*}(\bar{y}_{1})<0,
\]
which implies the LHS of (\ref{pf:lem.5.1}) is a decreasing function.
Also, the right hand side of (\ref{pf:lem.5.1}) is an increasing
function. Thus, (\ref{pf:lem.5.1}) has at most one sign change from
$+$ to $-$. Furthermore, note $\underset{\bar{y}_{1}\rightarrow-\infty}{\lim}\tilde{\mathbf{w}}(\bar{y}_{1})=1$
and $\underset{\bar{y}_{1}\rightarrow\infty}{\lim}\tilde{\mathbf{w}}(\bar{y}_{1})=-\infty$,
implying (\ref{pf:lem.5.1}) indeed has one and only one sign change
from $+$ to $-$. It follows from Theorem \ref{thm:Karlin} (i) that $\mathbf{R}(\tau)$
also has at most one sign change.

\emph{Step 2}: we show $\mathbf{R}(\tau)$ indeed has one sign change. Note
\begin{align*}
\mathbf{R}(\tau) & =\int\mathbf{w}(\bar{y}_{1})f\{\bar{y}_{1}|\tau\}d\bar{y}_{1}.
\end{align*}
Also, by Lemma \ref{lem:4}, $\mathbf{R}(\tau^{*})=0$. We may calculate
\begin{align*}
\frac{\partial}{\partial\tau}\mathbf{R}(\tau^*) & =\int\mathbf{w}(\bar{y}_{1})(\bar{y}_{1}-\tau^*)f\{\bar{y}_{1}|\tau^*\}d\bar{y}_{1}\\
&=\int\mathbf{w}(\bar{y}_{1})\bar{y}_{1}f\{\bar{y}_{1}|\tau^*\}d\bar{y}_{1}-\tau^*\underset{0}{\underbrace{\int\mathbf{w}(\bar{y}_{1})f\{\bar{y}_{1}|\tau^*\}d\bar{y}_{1}}}\\
&=\int\mathbf{w}(\bar{y}_{1})\bar{y}_{1}f\{\bar{y}_{1}|\tau^*\}d\bar{y}_{1}.
\end{align*}
Algebra shows:
\begin{align*}
\mathbf{w}(\bar{y}_1) & =w^{*}(y)^{2}\hat{\delta}^{*}(\bar{y}_1)\tilde{\mathbf{w}}(\bar{y}_1)\\
 & =\left(1-\hat{\delta}^{*}(\bar{y}_1)\right)^{2}\hat{\delta}^{*}(\bar{y}_1)\tilde{\mathbf{w}}(\bar{y}_1)\\
 & =\left(\frac{f\{\bar{y}_{1}|-\tau^*\}}{f\{\bar{y}_{1}|\tau^*\}+f\{\bar{y}_{1}|-\tau^*\}}\right)^{2}\frac{f\{\bar{y}_{1}|\tau^*\}}{f\{\bar{y}_{1}|\tau^*\}+f\{\bar{y}_{1}|-\tau^*\}}\tilde{\mathbf{w}}(\bar{y}_1).
\end{align*}
Thus, 
\begin{align*}
\frac{\partial}{\partial\tau}\mathbf{R}(\tau^*) &=\int w_{\tau^{*}}(\bar{y}_1)\tilde{\mathbf{w}}(\bar{y}_1)\bar{y}_1d\bar{y}_1,
\end{align*}
where $w_{\tau^{*}}(\bar{y}_1):=\frac{f^{2}\{\bar{y}_{1}|-\tau^{*}\}f^{2}\{\bar{y}_{1}|\tau^{*}\}}{\left(f\{\bar{y}_{1}|\tau^{*}\}+f\{\bar{y}_{1}|-\tau^{*}\}\right)^{3}}$>0 and is such that $w_{\tau^{*}}(-\bar{y}_1)=w_{\tau^{*}}(\bar{y}_1)$ for all $y\in\mathbb{R}$, and $\tilde{\mathbf{w}}(\bar{y}_1)$ is strictly decreasing and ranges from $+\infty$
to $-\infty$ as $\bar{y}_1$ increases from $-\infty$ to $+\infty$. Let $t^{*}$ be the unique point such that $\tilde{\mathbf{w}}(t^{*})=0$.
Suppose $t^{*}\geq0$. Then, we have the following decomposition
\begin{align*}
\frac{\partial}{\partial\tau}\mathbf{R}(\tau^*)  & =\int_{\bar{y}_1<-t^{*}}w_{\tau^{*}}(\bar{y}_1)\tilde{\mathbf{w}}(\bar{y}_1)\bar{y}_1d\bar{y}_1\\
 & +\int_{-t^{*}\leq \bar{y}_1\leq t^{*}}w_{\tau^{*}}(\bar{y}_1)\tilde{\mathbf{w}}(\bar{y}_1)\bar{y}_1d\bar{y}_1\\
 & +\int_{\bar{y}_1>t^{*}}w_{\tau^{*}}(\bar{y}_1)\tilde{\mathbf{w}}(\bar{y}_1)\bar{y}_1d\bar{y}_1,
\end{align*}
where all three terms above can be signed to be negative. A similar
decomposition also reveals that $\frac{\partial}{\partial\tau}\mathbf{R}(\tau^*) <0$ holds
true when $t^{*}<0$. Therefore, we conclude that $\tau^{*}$ is indeed a point of sign change. Thus, $\mathbf{R}(\tau)$
indeed has one and only one sign change by Theorem \ref{thm:Karlin} (i).

\emph{Step 3}: From Steps 1 and 2, Theorem \ref{thm:Karlin} (ii) further implies
that $\mathbf{R}(\tau)$ and $\tilde{\mathbf{w}}(\bar{y}_{1})$ changes
sign in the same order. Hence, we conclude that $\mathbf{R}(\tau)$
only has one sign change at $\tau^{*}$ from $+$ to $-$, i.e., $\tau^{*}$
is indeed a unique maximum of $\sup\limits _{\tau\in[0,\infty)}R_{sq}(\hat{\delta}^{*},P_{\tau})$. 
\end{proof}

\begin{thm}[Theorem 3 and Corollary 2, \citet{karlin1957polya}]\label{thm:Karlin}
 Let $p$ be strictly
P\'{o}lya type $\infty$ and assume that $p$ can be differentiated $n$
times with respect to $x$ for all $t$. Let $F$ be a measure on
the real line, and let $h$ be a function of $t$ which changes sign
$n$ times. 

(i) If 
\[
g(x)=\int p(x,t)h(t)dF(t)
\]
 can be differentiated $n$ times with respect to $x$ inside the
integral sign, then $g$ changes sign at most $n$ times and has at
most $n$ zeroes counting multiplicities, or is identically zero. The
function $g$ is identically zero if and only if the spectrum of $F$
is contained in the set of zeros of $h$.

(ii) If the number of sign changes of $g$ is $n$, then $g$ and
$h$ change signs in the same order.
\end{thm}

\section{Lemmas for Section \ref{sec:asymptotic}}

\begin{lem} \label{lem:C1}
Treatment rule $\hat{\delta}_{\tau}^{*}$ is a minimax optimal rule in the limit experiment, i.e.,
$\sup_{h}R_{sq}^{\infty}(\hat{\delta}_{\tau}^{*},h)=R^{*}$. 
\end{lem}
\begin{proof}
The mean square regret of a treatment rule $\hat{\delta}_{\tau}$
for each $h_{1}(b,h_{0})$ is
\begin{align}
R_{sq}^{\infty}(\hat{\delta}_{\tau},h_{1}(b,h_{0})) & =\left[\dot{\tau}^{\prime}h_{1}(b,h_{0})\right]^{2}\mathbb{E}_{\varDelta\sim N(h_{1}(b,h_{0}),I_{0}^{-1})}\left[1\left\{ \dot{\tau}^{\prime}h_{1}(b,h_{0})\geq0\right\} -\hat{\delta}\left(\dot{\tau}^{\prime}\varDelta\right)\right]^{2}\nonumber \\
 & =\left[\dot{\tau}^{\prime}h_{0}+b\right]^{2}\mathbb{E}_{\varDelta\sim N(h_{1}(b,h_{0}),I_{0}^{-1})}\left[1\left\{ \dot{\tau}^{\prime}h_{0}+b\geq0\right\} -\hat{\delta}\left(\dot{\tau}^{\prime}\varDelta\right)\right]^{2}\nonumber \\
 & =R_{sq}^{\infty}(\hat{\delta}_{\tau},h_{1}(b,0)),\label{pf:local.1}
\end{align}
where the last relation follows from $\dot{\tau}^{\prime}h_{1}(b,h_{0})=\dot{\tau}^{\prime}h_{0}+b$
and $\dot{\tau}^{\prime}h_{0}=0$ by construction. Thus, 
\begin{align*}
R^{*} & \leq\sup_{h_{1}}R_{sq}^{\infty}(\hat{\delta}_{\tau}^{*},h_{1})=\sup_{h_{0}}\sup_{b}R_{sq}^{\infty}(\hat{\delta}_{\tau}^{*},h_{1}(b,h_{0}))\\
 & =\sup_{b}R_{sq}^{\infty}(\hat{\delta}_{\tau}^{*},h_{1}(b,0))\leq\sup_{b}R_{sq}^{\infty}(\hat{\delta}^{*},h_{1}(b,0))\leq R^{*},
\end{align*}
where the first relation follows from the definition of $R^{*}$, the
second relation follows from definition of $h_{1}$, the third relation
follows from (\ref{pf:local.1}), the fourth relation follows by definition
of $\hat{\delta}_{\tau}^{*}$, and the final relation follows because $R^{*}$
is the worst case mean square regret of $\hat{\delta}^{*}$ and so
must be no smaller than $\sup_{b}R_{sq}^{\infty}(\hat{\delta}^{*},h_{1}(b,0))$.
\end{proof}

\begin{lem} \label{lem:C2}
$\hat{\delta}_{\tau}^{*}$ can be found by solving 
\begin{equation}
\inf_{\hat{\delta}_{\tau}}\sup_{b}R_{sq}^{\infty}(\hat{\delta}_{\tau},b),\label{eq:minimax.simple}
\end{equation}
where we recall that
\[
R_{sq}^{\infty}(\hat{\delta}_{\tau},b)=b^{2}\mathbb{E}_{\varDelta_{\tau}\sim N(b,\dot{\tau}^{\prime}I_{0}^{-1}\dot{\tau})}\left[1\left\{ b\geq0\right\} -\hat{\delta}_{\tau}\left(\varDelta_{\tau}\right)\right]^{2}.
\]
\end{lem}
\begin{proof}
Note for each $b\in\mathbb{R}$ and $\hat{\delta}_{\tau}=\hat{\delta}(\dot{\tau}^{\prime}\varDelta)$,
\begin{align*}
R_{sq}^{\infty}(\hat{\delta}_{\tau},h_{1}(b,0)) & =b^{2}\mathbb{E}_{\varDelta\sim N(h_{1}(b,h_{0}),I_{0}^{-1})}\left[1\left\{ b\geq0\right\} -\hat{\delta}\left(\dot{\tau}^{\prime}\varDelta\right)\right]^{2}\\
 & =b^{2}\mathbb{E}_{\varDelta_{\tau}\sim N(b,\dot{\tau}^{\prime}I_{0}^{-1}\dot{\tau})}\left[1\left\{ b\geq0\right\} -\hat{\delta}_{\tau}\left(\varDelta_{\tau}\right)\right]^{2}\\
 & =R_{sq}^{\infty}(\hat{\delta}_{\tau},b)
\end{align*}
where the second equality follows from
\[
\varDelta_{\tau}=\dot{\tau}^{\prime}\varDelta\sim N(\dot{\tau}^{\prime}h_{1}(b,h_{0}),\dot{\tau}^{\prime}I_{0}^{-1}\dot{\tau}),
\]
$\dot{\tau}^{\prime}h_{1}(b,h_{0})=b$ and we defined $\hat{\delta}_{\tau}(\varDelta_{\tau})=\hat{\delta}\left(\dot{\tau}^{\prime}\varDelta\right)$.
\end{proof}
\begin{lem} \label{lem:C3}
Under assumptions of Theorem \ref{thm:minimax.feasible}, the minimax optimal policy in the limit experiment is
\[
\hat{\delta}^{*}(\varDelta)=\frac{\exp\left(\frac{2\tau^{*}}{\sqrt{\dot{\tau}^{\prime}I_{0}^{-1}\dot{\tau}}}\dot{\tau}^{\prime}\varDelta\right)}{\exp\left(\frac{2\tau^{*}}{\sqrt{\dot{\tau}^{\prime}I_{0}^{-1}\dot{\tau}}}\dot{\tau}^{\prime}\varDelta\right)+1},
\]
where $\tau^{*}\approx1.23$ and solves (\ref{eq:minimax.1}) or (\ref{eq:minimax.2}).
\end{lem}
\begin{proof}
By Lemmas \ref{lem:C1} and \ref{lem:C2}, it suffices to find $\hat{\delta}_{\tau}^{*}$. Recall $\sigma_{\tau}^{2}=\dot{\tau}^{\prime}I_{0}^{-1}\dot{\tau}$. Thus,
\begin{align*}
R_{sq}^{\infty}(\hat{\delta}_{\tau},b) & =b^{2}\mathbb{E}_{\varDelta_{\tau}\sim N(b,\sigma_{\tau}^{2})}\left[1\left\{ b\geq0\right\} -\hat{\delta}_{\tau}\left(\varDelta_{\tau}\right)\right]^{2}\\
 & =b^{2}\int\left[1\left\{ b\geq0\right\} -\hat{\delta}_{\tau}\left(\varDelta_{\tau}\right)\right]^{2}\frac{1}{\sqrt{2\pi\sigma_{\tau}^{2}}}\exp\left(-\frac{(\varDelta_{\tau}-b)^{2}}{2\sigma_{\tau}^{2}}\right)d\varDelta_{\tau}\\
 & =b^{2}\int\left[1\left\{ b\geq0\right\} -\hat{\delta}_{\tau}\left(\sigma_{\tau}z\right)\right]^{2}\frac{1}{\sqrt{2\pi}}\exp\left(-\frac{(z-\frac{b}{\sigma_{\tau}})^{2}}{2}\right)dz\\
 & =\sigma_{\tau}^{2}\left(b_{\tau}\right)^{2}\int\left[1\left\{ b_{\tau}\geq0\right\} -\hat{\delta}_{1}(z)\right]^{2}\frac{1}{\sqrt{2\pi}}\exp\left(-\frac{(z-b_{\tau})^{2}}{2}\right)dz\\
 & =\sigma_{\tau}^{2}\left(b_{\tau}\right)^{2}\mathbb{E}_{Z\sim N(b_{\tau},1)}\left[1\left\{ b_{\tau}\geq0\right\} -\hat{\delta}_{1}(Z)\right]^{2}
\end{align*}
where the third line follows from the change of variable $z=\frac{\varDelta_{\tau}}{\sigma_{\tau}}$,
and the fourth line follows by letting $b_{\tau}:=\frac{b}{\sigma_{\tau}}$
and $\hat{\delta}_{1}(z):=\hat{\delta}_{\tau}\left(\sigma_{\tau}z\right)$.
Therefore, the minimax optimal rule $\hat{\delta}_{\tau}^{*}(\varDelta_{\tau})=\hat{\delta}_{1}^{*}(\frac{\varDelta_{\tau}}{\sigma_{\tau}})$,
where $\hat{\delta}_{1}^{*}(z)$ solves 
\begin{equation}
\min_{\hat{\delta}_{1}}\sup_{b_{\tau}}R_{sq}^{\infty}(\hat{\delta}_{1},b_{\tau}),\label{pf:minimax.limit.3}
\end{equation}
and where $R_{sq}^{\infty}(\hat{\delta}_{1},b_{\tau})=(b_{\tau})^2\mathbb{E}_{Z\sim N(b_{\tau},1)}\left[1\left\{ b_{\tau}\geq0\right\} -\hat{\delta}_{1}(Z)\right]^{2}$.
By Theorem \ref{thm:minimax}, we know the solution of (\ref{pf:minimax.limit.3})
is $\hat{\delta}_{1}(z)=\frac{\exp\left(2\tau^{*}z\right)}{\exp\left(2\tau^{*}z\right)+1}$,
where $\tau^{*}$ solves (\ref{eq:minimax.1}) or (\ref{eq:minimax.2}). Hence, $\hat{\delta}_{\tau}^{*}(\varDelta_{\tau})=\frac{\exp\left(2\tau^{*}\frac{\varDelta_{\tau}}{\sigma_{\tau}}\right)}{\exp\left(2\tau^{*}\frac{\varDelta_{\tau}}{\sigma_{\tau}}\right)+1}$.
Finally, note $\varDelta_{\tau}=\dot{\tau}^{\prime}\varDelta$. Thus,
the minimax optimal policy in the limit is
\[
\hat{\delta}^{*}(\varDelta)=\frac{\exp\left(\frac{2\tau^{*}}{\sigma_{\tau}}\dot{\tau}^{\prime}\varDelta\right)}{\exp\left(\frac{2\tau^{*}}{\sigma_{\tau}}\dot{\tau}^{\prime}\varDelta\right)+1}.
\]
\end{proof}

\end{document}